\documentclass[12pt]{article}
\usepackage[utf8]{inputenc}
\usepackage[english]{babel}
\usepackage{amsmath}
\usepackage{mathrsfs}
\usepackage{amssymb}
\usepackage{amsthm}
\usepackage{amsfonts}
\usepackage{xspace}
\usepackage[normalem]{ulem}
\usepackage{graphicx}
\usepackage{multirow}
\usepackage{appendix}
\usepackage{enumerate}
\usepackage{url} 
\usepackage{authblk}
\usepackage{float}
\usepackage{xcolor}
\usepackage{subcaption}
\usepackage{booktabs}
\usepackage{makecell}
\usepackage{xcolor}
\usepackage{soul}
\usepackage{tikz}
\usetikzlibrary{calc}

\usepackage[top=1in,bottom=2in,left=1in,right=1in]{geometry}
\usepackage{natbib}
\usepackage[unicode=true,bookmarks=true,bookmarksnumbered=false,bookmarksopen=false,breaklinks=false,pdfborder={0 0 1},backref=false,colorlinks=true]{hyperref}
\hypersetup{citecolor=blue}
\usepackage{url}

\newtheorem{theorem}{Theorem}

\newtheorem{corollary}{Corollary}

\newcommand{\bY}{\boldsymbol{Y}}

\newcommand{\bM}{\boldsymbol{M}}
\newcommand{\bX}{\boldsymbol{X}}
\newcommand{\bW}{\boldsymbol{W}}

\newcommand{\bV}{\mathbf{V}}

\newcommand{\bzero}{\boldsymbol{0}}

\newcommand{\bbeta}{\boldsymbol{\beta}}

\newcommand{\bvarepsilon}{\boldsymbol{\varepsilon}}

\newcommand{\bDelta}{\boldsymbol{\Delta}}
\newcommand{\bSigma}{\boldsymbol{\Sigma}}

\author[1]{Bingkai Wang}
\author[2]{Yu Du}

\affil[1]{\small The Statistics and Data Science Department of the Wharton School, University of Pennsylvania, PA, USA}
\affil[2]{\small Statistics, Data and Analytics, Eli Lilly and Company, IN, USA}

\title{\bf Improving the mixed model for repeated measures to robustly increase precision in randomized trials}

\begin{document}
\def\spacingset#1{\renewcommand{\baselinestretch}%
{#1}\small\normalsize} \spacingset{1}

\date{\vspace{-5ex}}

\maketitle
\begin{abstract}
In randomized trials, repeated measures of the outcome are routinely collected. The mixed model for repeated measures (MMRM) leverages the information from these repeated outcome measures, and is often used for the primary analysis to estimate the average treatment effect at the primary endpoint. MMRM, however, can suffer from bias and precision loss when it models intermediate outcomes incorrectly, and hence fails to use the post-randomization information harmlessly. This paper proposes an extension of the commonly used MMRM, called IMMRM, that improves the robustness and optimizes the precision gain from covariate adjustment, stratified randomization, and adjustment for intermediate outcome measures. Under regularity conditions and missing completely at random, we prove that the IMMRM estimator for the average treatment effect is robust to arbitrary model misspecification and is asymptotically equal or  more precise than the analysis of covariance (ANCOVA) estimator and the MMRM estimator. Under missing at random, IMMRM is less likely to be misspecified than MMRM, and we demonstrate via simulation studies that IMMRM continues to have less bias and smaller variance. Our results are further supported by a re-analysis of a randomized trial for the treatment of diabetes.

\end{abstract}
\noindent%
{\it Keywords: ANCOVA, heterogeneity, heteroscedasticity, repeated outcomes}  
\vfill

\newpage
\spacingset{1.5}
\setcounter{page}{1}

	\section{Introduction}\label{sec:intro}
``Intermediate outcomes'' in the context of randomized trials refer to the outcomes measured after treatment assignment but before the time point of interest. For example, for a trial studying the effect of a 12-month dietary plan on weight, the intermediate outcomes can be the weights at 3, 6, and 9 months after randomization. Intermediate outcomes are routinely collected and have been used in various ways, including trial monitoring \citep{shih1999planning}, decision making in the interim analysis  \citep{kunz2015comparison}, mediation analysis \citep{landau2018beyond}, etc.
We focus on a common but less studied purpose of intermediate outcomes, which is to robustly improve the precision of statistical inference on the average treatment effect at the primary endpoint.

The status quo for improving the statistical precision of longitudinal trials is adjusting for covariates and intermediate outcomes, for which the mixed model for repeated measures (MMRM, \citealp{mallinckrod2008recommendations}) is commonly used. MMRM involves a linear mixed model for the outcome measure at each visit on fixed effects (including the intercept, treatment indicators, and covariates, where the coefficients for covariates are often modeled as constant across visits) and a normal-distributed random intercept (which characterizes the correlation 
across outcome measures). 
If MMRM correctly specifies both fixed and random effects, it will yield consistent estimates for the average treatment effect at the primary endpoint and can improve precision by adjusting for intermediate outcomes \citep{marschner2001interim, galbraith2003interim, stallard2010confirmatory, hampson2013group, zhou2018predictive}. However, no analytical results are given when MMRM is arbitrarily misspecified. For example, it has been concerned that the outcomes of a randomized trial can be non-linear on covariates or non-normally distributed \citep{Kraemer2015}; in this case, MMRM may not capture the true relationship among outcomes, treatment, and covariates, thereby leaving open questions related to validity and precision of a misspecified MMRM.

In this article, we study whether a misspecified MMRM remains valid and precise. Under mild regularity conditions, we prove that the MMRM estimator for the average treatment effect at the primary endpoint is ``robust'', i.e., it is consistent and asymptotically normal under \textit{arbitrary} model misspecification, if there are no missing outcomes or the missingness is completely at random. This result also holds under stratified randomization \citep{Zelen1974} or given multiple treatments.
Under missing at random, however, a misspecified MMRM could lead to bias. 
In terms of precision, we contribute a surprising result that adjustment for intermediate outcomes by a misspecified MMRM could result in precision loss, even when visit-by-covariates interaction terms are incorporated. 
This result is an extension of \cite{ye2022toward}, which showed the potential precision loss with ANCOVA under a non-longitudinal setting, a special case of MMRM.
We support this finding by both theoretical derivations and simulations, indicating that the current MMRM may adjust for post-randomization information inappropriately and inadequately.

To address the above issues of MMRM, we propose ``Improved MMRM'' (IMMRM), an extension of MMRM that could augment robustness and precision by modeling treatment heterogeneity and heteroscedasticity. We show that IMMRM not only shares the same robustness property discussed above with MMRM, but is also less likely to be misspecified since MMRM is a special case of IMMRM, thereby increasing its reliability under the missing at random assumption. In terms of precision, IMMRM is the first model, to the best of our knowledge, that combines the precision gain from covariate adjustment, stratified randomization, and adjustment for intermediate outcomes. We prove that, under missing completely at random, (1) IMMRM is asymptotically equally or more precise than MMRM, (2) adjustment for intermediate outcomes by IMMRM will retain or reduce, but not increase the asymptotic variance, and (3) IMMRM could automatically accommodate the precision gain from stratified randomization. In addition, we formally characterize when stratified randomization and intermediate outcomes improve precision beyond covariate adjustment. Computationally, we demonstrate via simulation studies that this precision gain also exists for data with a small sample size under the missing at random assumption.

Our results build on existing work on causal inference under the super-population framework. For the validity of misspecified MMRM and IMMRM, our proof is based on the semiparametric theory \citep{tsiatis2007} and its recent extension to stratified randomization \citep{bugni2019inference}. In the research area of mixed models, although misspecified random effects have been extensively studied \citep{mcculloch2011misspecifying}, the existing results mainly relied on correctly-specified fixed effects and focused on random effects estimation, which is not the topic of this article. For robustly improving precision by adjusting for intermediate outcomes,  \cite{qian2019improving} used the general approach proposed by \cite{Lu2011,van2012targeted} and derived the non-parametric formulas for the precision gain from adjusting for baseline variables and the short-term outcomes. Later, \cite{van2020improving} proposed a robust estimator that achieves such precision gain for the interim analysis. 
More recently, \cite{schuler2022mixed} showed the analysis of covariance, which ignores intermediate outcomes, outperforms MMRM (without visit-by-covariate interaction terms) under two-arm, equal, simple randomization with no missing data. However, the above results fail to handle the complexity of non-monotone censoring at multiple visits or account for stratified randomization. 
{Non-monotone missingness brings a significant challenge to variance characterization, which needs to accommodate every missing pattern. In contrast to monotone missingness yielding $K$ missing patterns where $K$ is the number of visits,  non-monotone missingness can lead to $2^K$ missing patterns, thereby substantially driving up the complexity. Furthermore, stratified randomization will make this characterization even harder since it may alter the variance under each missing pattern differently.}
Finally, our results are extensions of established work on variance reduction by covariate adjustment (\citealp{YangTsiatis2001, Rubin&Laan2008, Jiang2018}, among others) and stratified randomization (\citealp{ye2022toward, wang2021model}, among others). 
{Specifically, \cite{wang2021model} established the asymptotic convergence for M-estimators under stratified randomization. We generalize this result to handle longitudinal data with multiple treatment groups and further characterize the comparison of asymptotic variances among the considered estimators.  
\cite{ye2022toward} proposed robust linearly-adjusted estimators for non-longitudinal randomized trials with  covariate-adaptive randomization. We extend their results to the setting with repeated outcome measures and show the superiority of IMMRM over their method in precision. }

In the next section, we introduce the motivating example. In Section~\ref{sec:def}, we
present our setup, notations, and assumptions. In Section~\ref{sec:estimators}, we describe the ANCOVA estimator, MMRM estimator (with or without visit-by-covariates interactions), and our proposed IMMRM estimator. Our main result is presented in Section~\ref{sec:main-results}, which consists of asymptotic theory, explanations of precision gain from intermediate outcomes, and results for two-arm equal randomization. Simulation study and data application are given in Sections~\ref{sec:simulation} and~\ref{sec:data-application}, respectively. Finally, we provide practical recommendations and discuss future directions in Section~\ref{sec:discussion}.

	
	\section{Motivating example: the IMAGINE-2 study}\label{sec:data-application_trial1}
The trial of ``A Study in Patients With Type 2 Diabetes Mellitus (IMAGINE 2)'' is a  52-week, two-arm, phase 3 randomized clinical study completed in 2014 \citep{davies2016basal}.
The goal of this trial was to evaluate the effect of basal insulin for the treatment of type 2 diabetes
in an insulin-na\"ive population.

Participants were randomly assigned to receive insulin peglispro (treatment, 1003 patients) or insulin glargine (control, 535 patients), aiming to achieve a 2:1 randomization ratio. Randomization was stratified by baseline HbA1c (Hemoglobin A1c, $<8.5\%$ or $\geq8.5\%$), low‐density lipoprotein cholesterol ($<100$mg/dL or $\ge100$mg/dL) and baseline sulphonylurea or meglitinide use (yes or not). 
HbA1c is a continuous measure of average blood glucose values in the prior three months.
The primary outcome was the change in HbA1c at week 52 from baseline (15\% missing outcomes), while intermediate outcomes were measured at week 4 (3\% missing outcomes), week 12 (6\% missing outcomes), week 26 (10\% missing outcomes) and week 39 (13\% missing outcomes).
We focus on estimating the average treatment effect of the primary outcome and adjusting for intermediate outcomes, baseline HbA1c value, and randomization strata.

\section{Definition and assumptions}\label{sec:def}
\subsection{Data generating distributions}
Consider a trial where the outcome is continuous and repeatedly measured at $K$ visits, where $K$ is a positive integer. 
For each participant $i = 1, \dots, n$, the outcome vector at $K$ visits is $\bY_i = (Y_{i1}, \dots, Y_{iK}) \in \mathbb{R}^K$ and the non-missing status at $K$ visits is $\bM_i = (M_{i1}, \dots, M_{iK})$, where $M_{it} = 1$ if $Y_{it}$ is observed at visit $t$, and 0 otherwise. 
Let $A_i$ be a categorical variable taking values in $\{0, 1, \dots, J: J \ge 1\}$,  with $A_i=j$ representing that participant $i$ is assigned to the $j$-th treatment group. By convention, we use $A_i=0$ to denote being assigned to the control group.
Let $\bX_i$ be a vector of baseline variables with length $p$. Throughout, we refer to $Y_{iK}$ as the final outcome, and $(Y_{i1},\dots, Y_{i,K-1})$ as the intermediate outcomes (if $K>1$), for conciseness.

We use the Neyman-Rubin potential outcomes framework \citep{neyman1990}, which assumes $\bY_i = \sum_{j = 0}^J  I\{A_i=j\} \bY_i(j)$, where $I$ is the indicator function and $\bY_i(j) = (Y_{i1}(j), \dots, Y_{iK}(j))$ is the potential outcome vector for treatment group $j = 0,\dots, J$. Analogously to the consistency assumption above, we assume $\bM_i = \sum_{j = 0}^J  I\{A_i=j\} \bM_i(j)$, where $\bM_i(j)$ is the indicator vector of whether  $\bY_i(j)$ would be observed at each of the $K$ visits, if participant $i$ was assigned to treatment group $j$ for $j = 0,\dots, J$.

For each participant $i$, we define the complete data vector as \\
$\boldsymbol{W}_i = (\bY_i(0), \dots, \bY_i(J), \bM_i(0),\dots, \bM_i(J), \bX_i)$ and the observed data vector as $\boldsymbol{O}_i = (\bY_i^o, \bM_i, A_i, \bX_i)$, where $\bY_i^o$ is the vector of observed outcomes, whose dimension may vary across participants. 
For example, if participant $i$ only shows up in visits 1 and $K$, then $\bY_i^o = (Y_{i1}, Y_{iK})$.
For the special case that a participant misses all post-randomization visits but still has baseline information recorded, the observed data vector is $\boldsymbol{O}_i = (\bM_i, A_i, \bX_i)$ with $\bM_i$ being a zero vector. All estimators defined in Section~\ref{sec:estimators} are functions of the observed data $(\boldsymbol{O}_1, \dots, \boldsymbol{O}_n)$. 
We make the following assumptions on $(\bW_1, \dots, \bW_n)$:

\vspace{5pt}
\noindent \textbf{Assumption 1.}
\begin{enumerate}[(a)]
    \item (Super-population) $\bW_i, i = 1, \dots, n$ are independent, identically distributed samples from an unknown distribution $P$.
    \item (Missing completely at random, MCAR) $\bM_i(j)$ is independent of $(\bY_i(j), \bX_i)$ for $j = 0, \dots, J$, and $\bM_i(0), \dots, \bM_i(J)$ are identically distributed.
    \item (Positivity for missing no visit) $P(M_{i1}(j) = 1, \dots, M_{iK}(j) = 1) > 0$ {for all $j=0,\dots, J$.} 
\end{enumerate}

Assumption 1 (a) is the standard assumption for causal inference under the super-population framework. Assumption 1 (b) implies that the missing pattern is independent of outcomes and covariates. This assumption, while restrictive, is needed for achieving valid inference under arbitrary model misspecification (See Section~\ref{sec:main-results}). An alternative and more flexible assumption is missing at random (MAR), i.e., $M_{it}(j)$ is independent of $Y_{it}(j)$ given $\{\bX_i, M_{i1}(j), Y_{i1}(j), \dots, M_{i,t-1}(j), Y_{i,t-1}(j)\}$. Under MAR, however, all models we examine in Section~\ref{sec:estimators} could be biased if they are misspecified, and we empirically test their performance via simulation studies (Section~\ref{sec:simulation}). Finally, Assumption 1 (c) requires the existence of {people who appear in all visits} in each arm, which is a convenient assumption for achieving parameter identifiability in the considered models.

In addition to Assumption 1, we also assume regularity conditions for the estimators defined in Section~\ref{sec:estimators}.
As we show in the Supplementary Material, all estimators we consider (including our proposed estimators) are M-estimators (Section 5 of \citealp{vaart_1998}), 
which is defined as a zero of prespecified estimating functions.
These conditions are similar to the classical non-parametric conditions given in Section 5.3 of \cite{vaart_1998} for proving consistency and asymptotic normality for M-estimators. We provide these regularity conditions in the Supplementary Material.

The parameters of interest are the average treatment effects of the final outcome comparing each treatment group to the control group, i.e., 
\begin{displaymath}
    \bDelta^* = (E[Y_{iK}(1)] - E[Y_{iK}(0)], \dots, E[Y_{iK}(J)] - E[Y_{iK}(0)]),
\end{displaymath}
where $E$ is the expectation with respect to the distribution $P$.
Our results in Section~\ref{sec:main-results} also apply to estimating any linear transformation of $\bDelta^*$, e.g., the average treatment effect comparing any two treatment groups.

\subsection{Simple and stratified randomization}
We consider two types of treatment assignment procedures: simple randomization and stratified randomization. 
For $j = 0,\dots, J$, let $\pi_j$ be the target proportion of participants assigned to the treatment group $j$. We assume that $\sum_{j=0}^J \pi_j  =1$ and $\pi_j > 0$ for all treatment groups. For example, equal randomization refers to the setting where $\pi_0 = \dots =\pi_J$.

Simple randomization allocates treatment by independent draws from a  categorical distribution on $A$, with $P(A=j) = \pi_j$ for $j=0,\dots, J$. 
Then $(A_1,\dots, A_n)$ are independent, identically distributed samples from this categorical distribution, and also independent of $(\bW_1, \dots, \bW_n)$. 

Under stratified randomization, treatment assignment depends on a set of categorical baseline variables, called stratification variables. We use a categorical random variable $S$ with support $\mathcal{S} = \{1,\dots, R\}$ to denote the joint levels created by all    stratification variables. For example, if randomization is stratified by sex (female or male) and weight (normal, overweight, or obesity), then $S$ can take $R=6$ possible values. Each element in $\mathcal{S}$ is referred to as a ``randomization stratum''. 
Within each randomization stratum, permuted blocks are used for sequential treatment allocation. Within each permuted block, fraction $\pi_j$ of participants are assigned to treatment group $j$. After a block is exhausted, a new block is used. 
For each participant $i$, we assume that $S_i$ is encoded as $R-1$ dummy variables (dropping one level to avoid collinearity) in the baseline vector $\bX_i$.

Different from simple randomization, stratified randomization is able to achieve balance within each randomization stratum, i.e., exact fraction  $\pi_j$ of participants are assigned to treatment group $j$. 
Due to this advantage, stratified randomization is used by 70\% of randomized clinical trials, according to a survey by \cite{Lin2015}.
In the meantime, stratified randomization also brings statistical challenge: the treatment assignments $(A_1,\dots, A_n)$ are not independent of each other such that the classical central limit theory cannot be directly applied. To overcome this challenge, recent work (\citealp{ye2022toward, wang2021model} among others) has established the non-parametric asymptotic theory for common estimators in non-longitudinal randomized trials. However, the analysis of robustness and efficiency for longitudinal trials, to the best of our knowledge, remains incomplete under stratified randomization. 

We finish this section by introducing a few additional definitions. For any two symmetric matrices $\bV_1$ and $\bV_2$ with the same dimension, we denote $\bV_1 \succeq \bV_2$ if $\bV_1 - \bV_2$ is positive semi-definite.
For any estimator $\widehat{\bDelta}$ of $\bDelta^*$, we call $\bV$ the asymptotic covariance matrix of $\widehat{\bDelta}$ if $\sqrt{n}(\widehat{\bDelta}_1 -\bDelta^*)$ weakly converges to a multivariate normal distribution with mean $\bzero$ and covariance matrix $\bV$.
If two estimators $\widehat{\bDelta}_1$ and $\widehat{\bDelta}_2$ of $\bDelta^*$ have asymptotic covariance matrix $\bV_1$ and $\bV_2$ respectively,  we call $\widehat{\bDelta}_2$ is (asymptotically) equally or more precise than $\widehat{\bDelta}_1$ if $\bV_1 \succeq\bV_2$. Such an expression is commonly used for scalar estimators, and we extend it to vector estimators. 

\section{Estimators}\label{sec:estimators}
For estimating the average treatment effects $\bDelta^*$, we introduce three commonly used models in Sections~\ref{subsec:ancova}-\ref{subsec:MMRM} and then propose a new working model in Section~\ref{subsec:IMMRM}. We focus on comparing these four models in the rest of this article.


\subsection{The Analysis of Covariance (ANCOVA) estimator}\label{subsec:ancova}
The ANCOVA estimator, $\widehat\bDelta^{(\textup{ANCOVA})}$, is defined as the maximum likelihood estimator (MLE) for parameters $(\beta_{A1}, \dots, \beta_{AJ})$ in the working model
\begin{equation}\label{eq:ANCOVA}
    Y_{iK} = \beta_0 + \sum_{j = 1}^ J \beta_{Aj}I\{A_i = j\} + \bbeta_{\bX}^\top \bX_i + \varepsilon_i,
\end{equation}
where $(\beta_0, \beta_{A1}, \dots, \beta_{AJ}, \bbeta_{\bX})$ are parameters, and $\varepsilon_i$ is the residual error independent of $\bX_i$ following a normal distribution with mean 0 and unknown variance $\sigma^2$. Participants with missing outcomes will be dropped from the model fit.
When outcomes are repeatedly measured, the ANCOVA estimator wastes the information from intermediate outcomes. Although ignoring such information does not affect the robustness of ANCOVA under MCAR \citep{YangTsiatis2001}, intermediate outcomes can be used to improve precision, as we show below. 

\subsection{The mixed-effects model for repeatedly measured outcomes (MMRM)} \label{subsec:MMRM}
There are two MMRM working models that are commonly seen in the analysis of randomized trials: one adjusts for covariates by a constant coefficient vector across visits, while the other includes visit-by-covariates interaction terms. For conciseness, we refer to the former model as ``MMRM-I'' and the latter as ``MMRM-II''.

The MMRM-I working model is defined as, for each $ t=1,\dots,K,$
\begin{equation}\label{eq:MMRM}
    Y_{it} = \beta_{0t} +  \sum_{j = 1}^ J \beta_{Ajt}I\{A_i = j\} + \bbeta_{\bX}^\top \bX_i + \varepsilon_{it}, 
\end{equation}
where $\{(\beta_{0t}, \beta_{A1t}, \dots, \beta_{AJt})\}_{t=1}^K$ and $\bbeta_{\bX}$ are parameters,  and $\boldsymbol{\varepsilon}_i = (\varepsilon_{i1}, \dots, \varepsilon_{iK})$ is the independent residual error following a 
multivariate normal distribution with mean $\bzero$ and unknown covariance $\bSigma$. The covariance matrix $\bSigma$ is assumed to be unstructured; that is, no other assumption is made on $\bSigma$ except that it is positive definite. 
The MMRM-I estimator $ \widehat{\bDelta}^{\textup{(MMRM)}}$ for $\bDelta^*$ is defined as the MLE for parameters $({\beta}_{A1K}, \dots, {\beta}_{AJK})$. Of note, Equation~(\ref{eq:MMRM}) is the marginal model, i.e., the random effects are marginalized and implicitly represented in the covariance matrix $\bSigma$. 

The MMRM-II working model is, for each $t = 1,\dots, K$,
\begin{equation}\label{eq:MMRM-II}
    Y_{it} = \beta_{0t} +  \sum_{j = 1}^ J \beta_{Ajt}I\{A_i = j\} + \bbeta_{\bX t}^\top \bX_i + \varepsilon_{it}, 
\end{equation}
which differs from the MMRM-I working model only on the regression coefficients of $\bX_i$, where $\bbeta_{\bX t}$ is substituted for $\bbeta_{\bX}$. We define the MMRM-II estimator, $\widehat{\bDelta}^{\textup{(MMRM-II)}}$, as the MLE for $(\beta_{A1K}, \dots, \beta_{AJK})$ in model~(\ref{eq:MMRM-II}). Different from MMRM-I, MMRM-II can capture the time-varying correlation of covariates and outcomes and has been recommended for the primary analysis \citep{mallinckrodt2020aligning}. 

In MMRM-I~(\ref{eq:MMRM}) and MMRM-II~(\ref{eq:MMRM-II}), the correlation of covariates and the outcome vector is constant among treatment groups, which we call the homogeneity assumption. In addition, the covariance matrix of $\bY_i$ is also assumed to be the same across treatment groups, which we refer to as the homoscedasticity assumption. Due to these assumptions, although MMRM-I and MMRM-II utilize information from intermediate outcomes, we find that they may be  less precise than ANCOVA if their working models are misspecified (Section~\ref{subsec: intermediate outcomes}). 




\subsection{Improved MMRM:  modeling heterogeneity and heteroscedasticity  among  treatment groups and visits}\label{subsec:IMMRM}

We propose a working model, called ``IMMRM'', that handles both heterogeneity and heteroscedasticity as follows:
for each $t=1,\dots,K,$
\begin{equation}\label{eq:IMMRM}
    Y_{it} = \beta_{0t} + \sum_{j = 1}^ J \beta_{Ajt}I\{A_i = j\} +  \bbeta_{\bX t}^\top\bX_i + \sum_{j = 1}^ J \bbeta_{A\bX jt}^\top I\{A_i = j\}\bX_i  +\sum_{j = 0}^ J\varepsilon_{ijt}I\{A_i = j\},
\end{equation}
where $\bvarepsilon_{ij} = (\varepsilon_{ij1}, \dots, \varepsilon_{ijK})$ has a multivariate normal distribution with mean $\bzero$ and covariance $\bSigma_j$ for each $j = 0, \dots, J$, and $(\bvarepsilon_{i0}, \dots, \bvarepsilon_{iJ})$ are independent of $\bX_i$ and each other. 
The fixed effects in model~(\ref{eq:IMMRM}) are
$(\beta_{0t}, \beta_{Ajt}, \bbeta_{\bX t}, \bbeta_{A \bX j t})$ for $j = 1, \dots, J$. Each $\bSigma_j, j = 0,\dots, J$ is assumed to be positive definite and unstructured. The IMMRM estimator for $\bDelta^*$ is defined as
\begin{displaymath}
 \widehat{\bDelta}^{(\textup{IMMRM)}} =  (\widehat\beta_{A1K} + \widehat\bbeta_{A\bX 1 K}^\top \overline{\bX}, \dots, \widehat\beta_{AJK} + \widehat\bbeta_{A\bX J K}^\top \overline{\bX}),
\end{displaymath}
where $\overline{\bX} = n^{-1}\sum_{i=1}^n \bX_i$ and $(\widehat\beta_{AjK}, \widehat\bbeta_{A\bX jK})$ are MLEs for parameters $(\beta_{AjK}, \bbeta_{A\bX jK})$. 


Compared with MMRM-I and MMRM-II, the IMMRM working model has two improvements. First, the inclusion of treatment-covariates-visits three-way interaction terms allows the relationship between the outcomes and baseline variables to vary across treatment groups and visits.
Such interaction terms models heterogeneity, which has been shown by \cite{tsiatis2007, ye2022toward} as an effective method to improve precision for scalar outcomes. We extend this idea to longitudinal repeated measures data. 
Second, the covariance matrix of $\bY_i$ is modeled separately for each treatment group, which accounts for heteroscedasticity. 
\cite{gosho2018effect} first proposed the idea of modeling heteroscedasticity in MMRM; however, they only provide empirical results to show its benefits.
We show, in Section~\ref{subsec: intermediate outcomes}, that modeling heteroscedasticity is necessary for achieving asymptotic precision gain when repeated measure outcomes are jointly modeled. 

Technically, IMMRM remains a linear mixed model for repeated measures and hence fits into the broad MMRM framework defined by \cite{mallinckrod2008recommendations}. However, IMMRM extends the commonly used MMRM working models, i.e., MMRM-I and MMRM-II, by handling heterogeneity and heteroscedasticity, two techniques that have not been used together in practice or theoretically analyzed for longitudinal randomized trials, thereby representing a methodology improvement with a practical impact.



\section{Main results}\label{sec:main-results}
\subsection{Asymptotic theory}
\begin{theorem}\label{Thm1}
Assume Assumption 1 and regularity conditions.
Consider $\widehat{\bDelta}^{\textup{(est)}}$ for  $\textup{est} \in \{\textup{ANCOVA},\ \textup{MMRM-I},\ \textup{MMRM-II},\ \textup{IMMRM}\}$.

For each of the four estimators, under simple or stratified randomization, we have consistency, i.e., $\widehat{\bDelta}^{\textup{(est)}}\rightarrow \bDelta^*$ in probability, and asymptotic normality, i.e., $\sqrt{n}(\widehat{\bDelta}^{\textup{(est)}}- \bDelta^*)$ weakly converges to a mean-zero multivariate normal distribution, under arbitrary misspecification of its working model.

Denote $\widetilde{\bV}^{\textup{(est)}}$ and $\bV^{\textup{(est)}}$ as the asymptotic covariance matrices of $\widehat{\bDelta}^{\textup{(est)}}$  under simple and stratified randomization, respectively. Then, for $\textup{est} \in \{\textup{ANCOVA},\ \textup{MMRM-I},\ \textup{MMRM-II}\}$,
\begin{equation}\label{eq:partial-order}
 \widetilde\bV^{\textup{(est)}} \succeq \bV^{\textup{(est)}} \succeq \widetilde\bV^{\textup{(IMMRM)}} = \bV^{\textup{(IMMRM)}}.
\end{equation}
In addition, we provide the conditions for $\bV^{\textup{(est)}} = \bV^{\textup{(IMMRM)}}$ in the Supplementary Material.
\end{theorem}

Theorem~\ref{Thm1} has the following implications. 
First, under simple or stratified randomization, each of the ANCOVA, MMRM-I, MMRM-II, and IMMRM estimators is robust under MCAR. 
Second, the IMMRM estimator has the highest precision among the four estimators. 
Such precision gain can be translated into sample size reduction.
Third, unlike the other three estimators, the IMMRM estimator has the same asymptotic covariance matrix under simple or stratified randomization. Therefore, the confidence interval for the IMMRM estimator under stratified randomization can be constructed as if simple randomization were used, without being statistically conservative. 




For performing hypothesis testing and constructing confidence intervals, we provide consistent estimators for the asymptotic covariance matrices $\widetilde\bV^{\textup{(est)}}$ and $\bV^{\textup{(est)}}$ in the Supplementary Material. The  sandwich variance estimator \citep{tsiatis2007} is used to estimate $\widetilde\bV^{\textup{(est)}}$; and the expression of $\widetilde\bV^{\textup{(est)}} - \bV^{\textup{(est)}}$ is derived in the Supplementary Material and approximated by substituting $\widehat{E}$, the expectation with respect to the empirical distribution, for $E$. 

Our proof for Theorem~\ref{Thm1} is provided in the Supplementary Material. In the proof, we derive the influence function for each of the four estimators and prove the asymptotic results by extending Theorem 1 of \cite{wang2021model} and  Lemma~B.2 of \cite{bugni2019inference} to longitudinal data. The major innovation and challenge of the proof are to derive the partial order~(\ref{eq:partial-order}) given the non-monotone missingness. We overcome this challenge by developing a series of inequalities related to functions of $\bM_i$ and positive definite matrices, which are presented in Lemma 1 of the Supplementary Material.  

When the missing data mechanism depends on the treatment, baseline information, or historical outcomes, all estimators considered in Theorem~\ref{Thm1} could be biased if their working models are misspecified, a common property shared by many outcome regression models \citep{wang2019analysis}. However, since MMRM-I and MMRM-II are special cases of IMMRM, IMMRM has additional robustness to model misspecification: it can still be correct when MMRM-I and MMRM-II are wrong, but not vice versa. Furthermore, the sensitivity analysis on missing data assumptions for MMRM \citep{o2014clinical, mallinckrodt2020aligning} can also be applied to IMMRM, allowing for its incorporation into the current practice. In the ensuing real-data-based simulation study, we  show that the IMMRM estimator is less biased and more precise than the others under MAR.


\subsection{How adjustment for intermediate outcomes improves precision}\label{subsec: intermediate outcomes}
Consider the IMMRM working model~(\ref{eq:IMMRM}) with a modification that intermediate outcomes are excluded, i.e.,
\begin{equation}\label{eq:ANHECOVA}
    Y_{iK} = \beta_{0K} + \sum_{j = 1}^ J \beta_{AjK}I\{A_i = j\} + \bbeta_{\bX K}^\top \bX_i + \sum_{j = 1}^ J \beta_{A\bX j K}^\top I\{A_i = j\}\bX_i + \sum_{j = 1}^ J\varepsilon_{ijK} I\{A_i=j\}.
\end{equation}
Then the IMMRM estimator for $\bDelta^*$ given the above working model~(\ref{eq:ANHECOVA}) reduces to the ``ANHECOVA'' estimator proposed by \cite{ye2022toward}. 
As a special case of the IMMRM model~(\ref{eq:IMMRM}), the working model~(\ref{eq:ANHECOVA}) differs from IMMRM only on whether intermediate outcomes are adjusted.
By comparing the asymptotic covariance matrices of the IMMRM estimator with the ANHECOVA estimator, we examine the contribution of intermediate outcomes in improving precision beyond what comes from covariate adjustment and stratified randomization. 

If $K=1$, a case with no intermediate outcomes, the IMMRM estimator and ANHECOVA estimators are equivalent. For $K>2$, the following Corollary shows that adjusting for intermediate outcomes by IMMRM will retain or increase but not lose precision.
\begin{corollary}\label{corollary:intermediate outcomes}
Assume $K>1$, Assumption 1 and regularity conditions in the Supplementary Material.
Let $\bV^{\textup{(ANHECOVA)}}$ be the asymptotic covariance matrix of the ANHECOVA estimator based on the working model~(\ref{eq:ANHECOVA}). Then $\bV^{\textup{(ANHECOVA)}} \succeq \bV^{\textup{(IMMRM)}}$.

Furthermore, $\bV^{\textup{(ANHECOVA)}}= \bV^{\textup{(IMMRM)}}$ if and only if, for each $t = 1,\dots, K-1$ and $j = 0,\dots, J$, $$P(M_{it}(j) = 1, M_{iK}(j) = 0)\ Cov\{Y_{it}(j), Y_{iK}(j) - \boldsymbol{b}_{Kj}^\top \bX_i\} = 0,$$ where $\boldsymbol{b}_{Kj} = Var(\bX_i)^{-1}Cov\{\bX_i, Y_{iK}(j)\}$,  $Cov\{\boldsymbol{U}_1, \boldsymbol{U}_2\} = E[\boldsymbol{U}_1\boldsymbol{U}_1^\top] - E[\boldsymbol{U}_1]E[\boldsymbol{U}_1]^\top$ is the covariance between any random vectors $\boldsymbol{U}_1, \boldsymbol{U}_2$ with finite second moments, and $Var(\bX_i)$ is the covariance matrix of $\bX_i$.
\end{corollary}





Corollary~\ref{corollary:intermediate outcomes} specifies the conditions for when adjusting for intermediate outcomes brings precision gain. For an intermediate visit $t$, $P(M_{it}(j) = 1, M_{iK}(j) = 0) > 0$ implies that a participant has a positive probability to both appear in visit $t$ and miss the last visit $K$; and $Cov\{Y_{it}(j), Y_{iK}(j) - \boldsymbol{b}_{Kj}^\top \bX_i\} \ne 0$ means that, for the treatment group $j$, $Y_{it}(j)$ is correlated with the residual of $Y_{iK}(j)$ after regressing on baseline variables. If an intermediate outcome satisfies the above two conditions for some treatment group $j$, then adjusting for it will lead to precision gain. 
On the contrary, adjusting for an intermediate outcome makes no change on the asymptotic covariance matrix if an intermediate outcome $Y_{it}$ is missing whenever $Y_{iK}$ is missing, or if it is not prognostic to the final outcome after controlling for $\bX_i$ in any treatment group.

Leveraging intermediate outcomes can bring
precision gain only when there are missing final outcomes and the intermediate outcomes are prognostic to the final  outcome beyond what is explained by covariates. This result also applies to MMRM-I and MMRM-II. When either of the two condition fails, ANHECOVA will lead to the same precision as IMMRM and be equally or more precie than MMRM-I and MMRM-II.
This finding generalizes the results of \cite{qian2019improving}, which considers a special case of our setup with $J=K=2$, simple randomization, and monotone censoring. 

Unlike the IMMRM estimator, adjusting for intermediate outcomes by MMRM-I or MMRM-II may cause efficiency loss. 
MMRM-I uses a constant vector to model the effect of covariates on outcomes, which is too stringent to capture its variation across visits and treatment groups. Thus, it can be less precise than ANCOVA, which only models the final outcome. This result is demonstrated in both simulation study and data application.
For the MMRM-II estimator, although simulation studies have shown that it is comparable to or better than ANCOVA \citep{mallinckrodt2020aligning, schuler2022mixed}, our theoretical examination indicates that it is \textit{not} universally correct: MMRM-II can be less precise than ANCOVA when there is heteroscedasticity. 
To make this argument concrete, we construct an analytical counterexample in the supplementary material showing that the MMRM-II estimator has a 5\% larger variance than the ANCOVA estimator. Of note, an exception is the two-arm equal randomization, where MMRM-II is asymptotically equally or more precise than ANCOVA, as discussed below.

\subsection{Special case: two-arm equal randomization}\label{subsec: 1:1-randomization}
In a general setting, e.g., multi-arm trials or unequal randomization, the efficiency comparison among ANCOVA, MMRM-I, and MMRM-II is indeterminate. However, under the two-arm equal randomization, the following corollary implies that the MMRM-II estimator has equal or smaller asymptotic variance than ANCOVA and MMRM-I; in addition, the MMRM-II estimator has the same asymptotic variance under simple or stratified randomization.
\begin{corollary}\label{corollary:1:1-randomization}
Assume $J=1$, $\pi_1=\pi_0$, Assumption 1, and regularity conditions in the Supplementary Material. Then $ V^{\textup{(ANCOVA)}} \ge V^{\textup{(MMRM-II)}}$, $ V^{\textup{(MMRM-I)}} \ge V^{\textup{(MMRM-II)}}$ and $\widetilde{V}^{\textup{(MMRM-II)}} = V^{\textup{(MMRM-II)}}$.
\end{corollary}


\section{Simulation study}\label{sec:simulation}

\subsection{Simulation settings}
We conducted a simulation study assessing the performance of the ANCOVA, MMRM-I, MMRM-II and IMMRM estimators in  four scenarios: small ($50$) versus large ($200$) sample size per arm, and MCAR versus MAR. Across all scenarios, the simulated data were based on the IMAGINE-2 study introduced in Section~\ref{sec:data-application_trial1}, reflecting a real-world data distribution and indicating that all working models are potentially misspecified.
Under the setting of large sample size and MCAR, we demonstrate our results in Section~\ref{sec:main-results}. For small sample size or MAR, we stress-test the performance of all four estimators.

In the simulation, we considered five post-randomization visits ($K=5$), three treatment groups ($J=2$) and four randomization strata ($R=4$), trying to duplicate the setting of the IMAGINE-2 study while considering a multi-arm study. We used three baseline variables from the data: baseline HbA1c ($X_{1i}$), indicator of LDL cholesterol $\geq 2.6$mmol/L ($X_{2i}$), and indicator of baseline HbA1c $\geq 8.5\%$ ($X_{3i}$). 
The simulated data were generated by the following steps.

First, we took the $844$ {people who appear in all visits} from the insulin peglispro arm to serve as the super-population for the control group, where the potential outcome vector for  participant $i$ was denoted as  $\bY_i(0)$. For the other two treatment groups, named TRT1 and TRT2, we generated the potential outcome $\bY_i(1)$ and $\bY_i(2)$ by:
\begin{equation*}\label{eq:simulpo}
    Y_{it}(j) = c_{t}(j) + Y_{it}(0) + \alpha_{t}(j)(X_{1i} - \overline{X_{1i}}) + \gamma_{t}(j)(X_{1i}^2 - \overline{X_{1i}^2}), 
\end{equation*} where, for $j=1, 2$ and $t=1,\dots, K$, $c_{t}(j)$ is a constant specifying the average treatment effect comparing TRT$a$ to the control group at time $t$, $\alpha_t(j)$ and $\gamma_t(j)$ are coefficients that determine the degree of heterogeneity and heteroscedasticity among treatment arms, and $\overline{X_{1i}}$ and $\overline{X_{1i}^2}$ are averages of $X_{1i}$ and $X_{1i}^2$ across $i$ in the control group.
The quadratic terms adds another layer of model misspecification. 
{We set $(c_1(1), \dots, c_K(1)) = (0, 0, 0, 0, 0)$, $(c_1(2), \dots, c_K(2))= (-0.2, -0.5, -0.8, -0.9, -1)^\top$.}
This indicates that the true average treatment effect is $\bDelta^* = (0,-1)$.
{We let $(\alpha_1(1), \dots, \alpha_K(1)) = -10^{-2} \times (1,2,4,5,5)$, $(\alpha_1(2), \dots, \alpha_K(2)) = -10^{-2} \times (0.5, 1, 1.5, 7, 10)^\top$, $(\gamma_1(1), \dots, \gamma_K(1)) = -10^{-2} \times (1,3,3,2,5)^\top$ and $(\gamma_1(2), \dots, \gamma_K(2)) = -10^{-2} \times(1.5, 2, 1, 10, 10)^\top$. }
The negative signs of $\alpha_t(j)$ and  $\gamma_t(j)$ indicated that a higher baseline HbA1c is associated with a larger HbA1c change. 

Next, we independently resample $n$ participants with replacement from the empirical distribution of $\{(\bY_i(0), \bY_i(1), \bY_i(2), X_{1i}, X_{2i}, X_{3i})\}_{i=1}^{844}$. We then applied stratified randomization with a block size of 6 to randomly assign the resampled participants to three treatment arms with 1:1:1 randomization ratio, where randomization strata are defined by the joint levels of $X_{2i}$ and $X_{3i}$. The treatment variable was $A_i$ taking values in $\{0,1,2\}$ and the realized outcome vector was $\bY_i = \sum_{j=0}^2 I\{A_i=j\}\bY_i(j)$.

In the final step, we assigned missing outcomes under monotone censoring. We mimicked the missing data percentages in the IMAGINE-2 study such that $3\%, 6\%, 10\%, 13\%$, and $15\%$ are expected to be missing at visits 1-5, respectively. For MCAR, the censoring time was generated by a logistic regression with an intercept  only. {For MAR, the censoring time was determined by a logistic regression model on the treatment group and quadratic forms of the previous outcome and baseline HbA1c. } The dropout rate was made higher in the control arm and TRT1 compared to TRT2, and the participants are more likely to drop out given a higher HbA1c observed from the previous visit. The complete model specifications for MAR is provided in the Supplementary Material.

Given a simulated data set,  all four estimators were used to estimate $\bDelta^*$. The standard error is calculated by the sandwich variance estimator (using the option ``empirical'' in SAS with an adjustment for the variability in the mean covariates as pointed out by \citealp{qu2015estimation} in estimating the average treatment effect), which does not account for stratified randomization. The 95\% confidence interval was constructed by normal approximations.
The above procedure was repeated 10,000 times for each scenario. 
The performance metrics are the bias, empirical standard error, averaged standard error for each estimate, coverage probability, probability of rejecting the null hypothesis, and the relative mean squared error compared to IMMRM.


\subsection{Simulation results}
The simulation results are summarized in Table~\ref{simulstudy:n50}  for $n=150$ and Table~\ref{simulstudy:n200} for $n=600$.
Across all scenarios, IMMRM maintains a coverage probability close to 95\% and consistently outperforms the other estimators in bias and precision, reflected by its smallest RMSE. 
Furthermore, ANCOVA, MMRM-I, and MMRM-II fail to account for the precision gain from stratified randomization, resulting in overestimating the true standard error (comparing ASE and ESE) and thus losing power.

Under MCAR, all estimators have negligible bias despite their misspecified models, which is consistent with our asymptotic result. The MMRM-I estimator has the largest mean squared error and variance among the four estimators, indicating its inadequate adjustment for intermediate outcomes. When the treatment has no effect (TRT1), IMMRM and MMRM-II have similar performance and are 14\% more precise than ANCOVA. 
{Here, the percent  precision gain is computed as 1 minus the ratio of squared empirical standard error.}
Given a non-zero treatment effect (TRT2),  IMMRM is 15\% more precise than both ANCOVA and MMRM-II. In addition, ANCOVA, MMRM-I, and MMRM-II have an additional 4-35\% variance increase due to ignoring stratified randomization, whereas IMMRM remains anti-conservative. These results further support our theorems in Section~\ref{sec:main-results}.

{Under MAR, ANCOVA, MMRM-I, and MMRM-II have bias  as high as 0.142 and under-coverage as high as 7\%. In contrast, IMMRM remains valid unbiased and precise, a combined result of its augmented robustness and the relatively small proportion of missing outcomes (3-15\% across visits). Our choice of missing outcome proportions duplicates the IMAGINE-2 study, thereby testing the performance of IMMRM in a real-world setting for phase-3 clinical trials.}

{Compared with a large sample size, a small sample size leads to slightly larger bias and less accurate standard error estimates for all estimators due to increased uncertainty. IMMRM is most influenced by the sample size due to its complexity in modeling, leading to underestimation of true standard error and 1-2\% under-coverage; this issue is mitigated in Table~\ref{simulstudy:n200} with a larger sample size. 
To improve the coverage probability under a small sample size, we recommend using finite-sample correction methods for variance estimation, such as applying the ``empirical=firores'' option in SAS \citep{mancl2001covariance}. On this topic, an evaluation of  bias-corrected variance estimators for MMRM can be found in \cite{gosho2017comparison}.}


\begin{table}[p!]
\centering
\caption{Simulation results comparing candidate estimators with $50$ samples per arm under MCAR and MAR.
For each estimator, we estimate the average treatment effect of TRT1 and TRT2, both comparing the control group. The following measures are used: bias, empirical standard error (ESE), averaged standard error (ASE), coverage probability (CP), probability of rejecting the null (PoR), relative mean squared error compared to IMMRM (RMSE). For RMSE, a number bigger than 1 indicates a larger mean squared error than IMMRM.}
\label{simulstudy:n50}
\renewcommand{\arraystretch}{0.8}
\resizebox{\textwidth}{!}{
\begin{tabular}{lllrrrrrr} 
\toprule
                                               &                           &    Group        & Bias    & ESE    & ASE    & CP(\%)   & PoR(\%)  & RMSE\\ 
\midrule
   \multirow{8}{*}{MCAR} & \multirow{2}{*}{ANCOVA}   & TRT1  & 0.004 & 0.229 & 0.238 & 95.7 & 4.3 & 1.190 \\ 
                                               &                           & TRT2 & 0.007 & 0.268 & 0.275 & 95.1 & 95.9 & 1.184 \\ 
\cmidrule{2-9}
                                               & \multirow{2}{*}{MMRM-I}     & TRT1  & 0.006 & 0.232 & 0.264 & 97.2 & 2.8 & 1.222 \\ 
                                              &                           & TRT2 & 0.004 & 0.301 & 0.370 & 98.5 & 82.1 & 1.484 \\ 
\cmidrule{2-9}
                                               & \multirow{2}{*}{MMRM-II} & TRT1  & 0.005 & 0.216 & 0.212 & 94.6 & 5.4 & 1.059 \\
                                               &                           & TRT2 & 0.005 & 0.268 & 0.284 & 96 & 95.6 & 1.179 \\ 
\cmidrule{2-9}
                                               & \multirow{2}{*}{IMMRM}    & TRT1  & 0.002 & 0.210 & 0.199 & 93.1 & 6.9 & - \\ 
                                               &                           & TRT2 &0.003 & 0.247 & 0.239 & 93.6 & 98.7 & - \\ 
\hline
                         \multirow{8}{*}{MAR}  & \multirow{2}{*}{ANCOVA}   & TRT1  & 0.058 & 0.253 & 0.254 & 94.3 & 5.7 & 1.318 \\ 
                                               &                           & TRT2 & 0.142 & 0.298 & 0.295 & 90.8 & 83.6 & 1.556 \\
\cmidrule{2-9}
                                               & \multirow{2}{*}{MMRM-I}     & TRT1  & 0.012 & 0.248 & 0.274 & 96.6 & 3.4 & 1.208 \\
                                               &                           & TRT2 & 0.042 & 0.322 & 0.378 & 97.1 & 75.5 & 1.504 \\
\cmidrule{2-9}
                                              & \multirow{2}{*}{MMRM-II} & TRT1  & 0.028 & 0.230 & 0.220 & 93.4 & 6.6 & 1.055 \\ 
                                               &                           & TRT2 &0.046 & 0.288 & 0.293 & 94.6 & 91.3 & 1.211 \\ 
\cmidrule{2-9}
                                               & \multirow{2}{*}{IMMRM}    & TRT1  & 0 & 0.226 & 0.206 & 92.6 & 7.4 & - \\ 
                                              &                           & TRT2 & 0.004 & 0.265 & 0.245 & 93 & 97.5 & - \\ 
\bottomrule
\end{tabular}
}
\end{table}

\begin{table}[p!]
\centering
\caption{Simulation results comparing candidate estimators with $200$ samples per arm under MCAR and MAR.
For each estimator, we estimate the average treatment effect of TRT1 and TRT2, both comparing the control group. The following measures are used: bias, empirical standard error (ESE), averaged standard error (ASE), coverage probability (CP), probability of rejecting the null (PoR), relative mean squared error compared to IMMRM (RMSE). For RMSE, a number bigger than 1 indicates a larger mean squared error than IMMRM.}
\label{simulstudy:n200}
\renewcommand{\arraystretch}{0.8}
\resizebox{\textwidth}{!}{
\begin{tabular}{lllrrrrrr} 
\toprule
                                               &                           &          Group  & Bias    & ESE    & ASE    & CP(\%)  & PoR(\%) & RMSE \\ 
\midrule
 \multirow{8}{*}{MCAR} & \multirow{2}{*}{ANCOVA}   & TRT1  &0.002 & 0.112 & 0.120 & 96.9 & 3.1 & 1.167 \\ 
                                               &                           & TRT2 & 0.006 & 0.133 & 0.140 & 96 & 100 & 1.196 \\
\cmidrule{2-9}
                                               & \multirow{2}{*}{MMRM-I}     & TRT1  & 0 & 0.116 & 0.133 & 97.8 & 2.2 & 1.241 \\ 
                                               &                           & TRT2 & 0.004 & 0.149 & 0.186 & 98.5 & 100 & 1.500 \\ 
\cmidrule{2-9}
                                               & \multirow{2}{*}{MMRM-II} & TRT1 & 0.002 & 0.105 & 0.107 & 95.7 & 4.3 & 1.028 \\ 
                                               &                           & TRT2 &  0.005 & 0.132 & 0.143 & 96.7 & 100 & 1.176 \\ 
\cmidrule{2-9}
                                               & \multirow{2}{*}{IMMRM}    & TRT1  & 0.001 & 0.104 & 0.102 & 95 & 5.1 & - \\ 
                                               &                           & TRT2 & 0.003 & 0.122 & 0.122 & 95.3 & 100 & - \\ 
\cmidrule{1-9}
                         \multirow{8}{*}{MAR}  & \multirow{2}{*}{ANCOVA}   & TRT1  & 0.051 & 0.126 & 0.130 & 94.3 & 5.7 & 1.480 \\ 
                                               &                           & TRT2 & 0.116 & 0.147 & 0.152 & 88.1 & 100 & 2.114 \\ 
\cmidrule{2-9}
                                               & \multirow{2}{*}{MMRM-I}     & TRT1  & 0.014 & 0.123 & 0.138 & 97.1 & 2.9 & 1.224 \\ 
                                               &                           & TRT2 & 0.047 & 0.162 & 0.192 & 97.3 & 100 & 1.711 \\ 
\cmidrule{2-9}
                                               & \multirow{2}{*}{MMRM-II} & TRT1  & 0.026 & 0.114 & 0.112 & 93.7 & 6.3 & 1.096 \\ 
                                              &                           & TRT2 & 0.043 & 0.142 & 0.148 & 94.7 & 100 & 1.319 \\ 
\cmidrule{2-9}
                                               & \multirow{2}{*}{IMMRM}    & TRT1  & 0.001 & 0.112 & 0.107 & 93.9 & 6.2 & - \\ 
                                               &                           & TRT2 & 0 & 0.129 & 0.126 & 94.2 & 100 & - \\ 
\bottomrule
\end{tabular}
}
\end{table}

{In addition to the two simulation studies, we conducted additional simulations where the treatment effect heterogeneity and heteroscedasticity are removed. Specifically, we modify the distribution of $Y_{it}(j)$ by letting $Y_{it}(j) = c(j) + Y_{it}(0)$ with $c(1)= 0$ and $c(2)=-1$, yielding a constant treatment effect model. As a result, the data generating distribution favors MMRM-II, and we use this scenario to stress-test the performance of IMMRM under homogeneity and homoscedasticity. The complete simulation results are provided as Tables 1 and 2 in the Supplementary Material. In Summary, IMMRM maintains negligible bias and nominal coverage, while its variance and RMSE are 0-3\% larger than MMRM-II. This is expected, since MMRM-II is efficient if correctly specified, and the additional parameters in IMMRM cause the small precision loss in finite samples. Compared to ANCOVA and MMRM-I, IMMRM still shows 3-6\% precision gain given $n=600$ and can yield valid inference under MAR in this scenario.}

\section{Data application}\label{sec:data-application}
We applied the ANCOVA, MMRM-I, MMRM-II, and IMMRM estimators to the IMAGINE-2 study. 
Due to very limited data for participants with no sulfonylurea/meglitinide use at baseline, we dropped this stratification variable in the analysis, resulting in 4 randomization strata in total.
For each estimator, we computed the estimate, 95\% confidence interval based on normal approximations, and the proportional variance reduction (PVR), defined as one minus the variance ratio of the estimator and the MMRM-I estimator. The standard error of each estimator is computed by the sandwich variance estimator.


Table~\ref{table:data-application} summarizes our results of the data application. All estimators have similar treatment effect estimates, while their variances differ. Consistent with the results from our simulation study, IMMRM yields the most precise estimator, which is 2.0\%, 4.7\%, and 12.6\%  more efficient than MMRM-II, ANCOVA, and MMRM-I, respectively, indicating its advantage in increasing power.

\begin{table}[ht]
\centering
\caption{Summary of data application.}\label{table:data-application}
\renewcommand{\arraystretch}{0.8}
\resizebox{\textwidth}{!}{
\begin{tabular}{lrrr}
\hline
Method     &  \makecell[r]{Treatment effect \\ estimate} & \makecell[r]{95\% confidence \\ interval} & \makecell[r]{Proportional variance reduction \\ compared to MMRM-I}  \\
\hline
ANCOVA     & -0.283 & (-0.380, -0.186) & 8.3\% \\
MMRM-I       & -0.296 & (-0.398, -0.194) & -\\
MMRM-II   & -0.299 & (-0.395, -0.203) & 10.8\%\\
IMMRM      & -0.292 & (-0.388, -0.196) & 12.6\% \\
\hline
\end{tabular}}
\end{table}

\section{Discussion}\label{sec:discussion}
For the analysis of longitudinal repeated measures data, we propose the IMMRM estimator that can improve precision for estimating the average treatment effect of the final outcome, compared to standard practice that uses ANCOVA, MMRM-I, or MMRM-II. Such precision gain comes from  appropriately adjusting for intermediate outcomes and accounting for stratified randomization, which can be translated into sample size reduction in trial planning. This result also applies to estimating the expectation of the final outcome for each treatment group.
{In addition to the results of precision we discussed, the proposed estimator is also locally efficient, i.e., asymptotically efficient when the model is correctly specified, as implied by Section 5.5 of \cite{vaart_1998}. However, efficiency may not be achieved when IMMRM is misspecified. In this case, our main theorem shows the superiority of the proposed estimator over the compared methods in precision.}

Like most regression models, the price paid for precision gain is to estimate more parameters. IMMRM involves $Kp(J+1)-p$ more parameters than MMRM-I and $KpJ$ more parameters than MMRM-II. 
Given that most randomized clinical trials involve a few treatment groups and visits relative to the sample size \citep{mallinckrod2008recommendations}, the number of extra parameters is often dominated by the number of covariates.
When the sample size is small, we recommend prespecifying a small number of prognostic baseline variables to avoid losing degrees of freedom.

{In practice, the precision gain of IMMRM compared to ANCOVA or MMRM is likely to occur when the sample size is large compared to the number of parameters. However, it is often difficult to determine how large the sample size should be since the performance of different methods will also vary by the true data-generating distribution (e.g., disease areas), quality of prespecified baseline variables  and intermediate outcomes (e.g., how  prognostic they are) and randomness in the samples \citep{kahan2014risks}. In our simulations with 600 samples per arm, IMMRM has 45 parameters and can approximate the asymptotic results well, while the standard error estimation is less precise when the sample size decreases to 150. 
    Following the recommendation in the literature for covariate adjustment \citep{benkeser2021improving},
    we recommend that the sample size for each model parameter is larger than 20 to avoid overfitting.}


Under the MAR assumption, an alternative approach for estimating the average treatment effect is the targeted minimum loss-based estimation (TMLE,  \citealp{van2012targeted}). This method involves recursively fitting regression models for outcomes and propensity scores for non-missingness, and is consistent as long as one of the two sets of models is correctly specified.  It is an open question, to the best of our knowledge, whether TMLE guarantees an efficiency gain by adjusting for intermediate outcomes while retaining its double-robustness.


{Leveraging post-randomization information to improve precision is not limited to intermediate outcomes. Our main results also apply to adjustment for other post-randomization continuous-valued random variables measured before the final outcomes, such as the body mass index measured at each visit. When these additional variables provide new prognostic information beyond intermediate outcomes and covariates, adding them to the IMMRM model can bring further asymptotic precision gain. }

{Based on MMRM, an alternative approach to maximum likelihood estimation is through maximizing the restricted maximum likelihood (REML), which is known for reducing the bias of the variance estimators in small samples. When MMRM is correctly specified, \cite{jiang2017asymptotic} showed that the  maximum likelihood estimator and REML estimator are equivalent asymptotically.  When MMRM is misspecified, however, \cite{maruo2020note} pointed out that REML can be incompatible with MMRM when the mean and variance parameters are not orthogonal. In this case, the construction of robust REML estimators based on MMRM remains an open question, which we leave for future research.}

{For robust variance estimation of MMRM estimators, we note that some functions in statistical software, e.g., the ``mixed'' procedure in SAS, implicitly assume orthogonality of mean and variance parameters. However, \cite{maruo2020note} showed such functions may lead to biased variance estimators when this  orthogonality assumption fails. To address this issue, \cite{maruo2020note} provided detailed guidance on how to implement MMRM with SAS and obtain unbiased variance estimators, which also applies to our IMRMM estimator. In addition, we also provide R functions to compute the point estimate and consistent variance estimator, which is available at \url{https://github.com/BingkaiWang/MMRM}.}


	\section*{Supplementary materials}
The Supplementary material contains regularity conditions for our theorems, formulas of the asymptotic variances in the theorem and their consistent estimators, proofs, an example which shows MMRM-II is less precise than ANCOVA, the missing data distribution under MAR in the simulation study, and additional simulation studies. 

\section*{Acknowledgement}
We sincerely thank the comments from the editor, Dr. Ashkan Ertefaie, and the anonymous reviewer to help us improve the paper. We also thank Drs. Yongming Qu, Michael Rosenblum, Ting Ye, and Yanyao Yi for their input in the early stage of this paper. 

	\par


\bibliographystyle{chicago}      
\bibliography{references}   

\end{document}


\def\spacingset#1{\renewcommand{\baselinestretch}%
{#1}\small\normalsize} \spacingset{1}

\date{\vspace{-5ex}}

\maketitle

\renewcommand\thesection{\Alph{section}}
\setcounter{section}{0}

\spacingset{1.5}
\setcounter{page}{1}

In Section \ref{sec:notations}, we give some useful notations. In Section~\ref{sec:regulartiy-conditions}, we provide the regularity conditions for Theorem 1 of the main paper. In Section~\ref{sec:variance-estimators}, we define the variance estimators for $\widetilde{V}^{\textup{(est)}}$ and $V^{\textup{(est)}}$ and provide consistent estimators. In Section~\ref{sec:lemmas}, we introduce a few lemmas for proving our main results. In Section~\ref{sec:proofs}, we prove Theorem 1, Corollary 1 and Corollary 2 of the main paper. In Section~\ref{sec: an example MMRM-II}, we give an example where MMRM-II is less precise than ANCOVA. In Section~\ref{sec:MAR}, we provide the missing data mechanism for MAR in the simulation study. In Section~\ref{sec: additional simulation}, we provide additional simulations under homogeneity and homoscedasticity.

\section{Notations}\label{sec:notations}
Throughout the proofs, we use $\bY, \bY(j), Y_{t}, Y_t(j)$ instead of $\bY_i, \bY_i(j), Y_{it}, Y_{it}(j)$ to represent random variables from the distribution $P$ for conciseness. The same notation is used for $\bM, \bM(j), M_t, M_t(j), \bX, A$.

Let $\bone_K$ be the  column vector of length $K$ with each component equal to 1, $\bzero_K$ be the  column vector of length $K$ with each component equal to 0, and $\boldsymbol{e}_t$ be the column vector of length $K$ with the $t$-th entry 1 and the rest 0. Let $\bI_K$ be the $K\times K$ identity matrix. Let $\otimes$ be the Kronecker product. Let $I$ be the indicator function, i.e. $I\{A\} = 1$ if event $A$ is true and 0 otherwise. For any random vector $\boldsymbol{W}$ with finite second-order moment, we define $\widetilde{\boldsymbol{W}} = \boldsymbol{W} - E[\boldsymbol{W}]$ and $Var(\boldsymbol{W}) = E[\widetilde{\boldsymbol{W}}\widetilde{\boldsymbol{W}}^\top]$. For two random vectors $\boldsymbol{W}_1$ and $\boldsymbol{W}_2$, we define $Cov(\boldsymbol{W}_1, \boldsymbol{W}_2) = E[\widetilde{\boldsymbol{W}}_1 \widetilde{\boldsymbol{W}}_2^\top]$. Let $\{0,1\}^K$ be the set of $K$-dimensional binary vectors, i.e. $\{0,1\}^K = \{(x_1, \dots, x_K)^\top: x_t \in \{0,1\}, t = 1, \dots K\}$. For any matrix (or vector), we use $||\cdot||$ to denote its $L_2$ matrix (or vector) norm. 
For any vector $\boldsymbol{v} = (v_1, \dots, v_L)$, we use $diag\{\boldsymbol{v}\}$ or $diag\{v_l: l = 1, \dots, L\}$ to denote an $L\times L$ diagonal matrix with the diagonal entries being $(v_1, \dots, v_L)$. For any sequence $x_1, \dots, x_n, \dots$, we define $P_n x = n^{-1} \sum_{i=1}^n x_i$. 

\section{Regularity conditions}\label{sec:regulartiy-conditions}
The regularity conditions for Theorem 1 are assumed on estimating equations $\bpsi^{\textup{(ANCOVA)}}$, $\bpsi^{\textup{(MMRM)}}$ and $\bpsi^{\textup{(IMMRM)}}$, which are defined by Equations~(\ref{psi:ANCOVA}), (\ref{psi:MMRM}), and (\ref{psi:IMMRM}) below, respectively. 
Each of $\bpsi^{\textup{(ANCOVA)}}$, $\bpsi^{\textup{(MMRM)}}$ and $\bpsi^{\textup{(IMMRM)}}$ is a $q$-dimensional function of random variables $(A, \bX, \bY, \bM)$ and a set of parameters $\btheta \in \mathbb{R}^q$, where $q$ and $\btheta$ vary among estimators, and $\bDelta$ are embedded in $\btheta$. Without causing confusion, we use $\bpsi(A, \bX, \bY, \bM;\btheta)$ to represent any of the above three estimating equations. 
As we show in the proof of Theorem 1 below, each of the ANCOVA, MMRM-I, MMRM-II and IMMRM estimators is an M-estimator, which is defined as the solution of $\bDelta$ to the equations $P_n \bpsi(A, \bX, \bY, \bM;\btheta) = \bzero$.

The regularity conditions are  similar to those used in Section 5.3 of \cite{vaart_1998} in their theorem on estimating equations $\bpsi$ for showing asymptotic normality of M-estimators for independent, identically distributed data. 
The regularity conditions are given below:

\hspace{5pt} (1) $\btheta \in \mathbf{\Theta}$, a compact set in $\mathbb{R}^{q}$.

\hspace{5pt} (2) $E[||\bpsi(j, \bX, \bY(j), \bM(j); \btheta)||^2] < \infty$ for any $\btheta \in \mathbf{\Theta}$  and $j \in \{0, \dots, J\}$.

\hspace{5pt} (3)
There exists a unique solution in the interior of $\mathbf{\Theta}$, denoted as $\underline{\btheta}$, to the equations
\begin{equation*}
\sum_{j=0}^J \pi_j E[\bpsi(j, \bX, Y(j), M(j); \btheta)]=\boldsymbol{0}.
\end{equation*}

\hspace{5pt} (4) For each $j \in \{0,\dots, J\} $, the function $\btheta \mapsto \bpsi(j, \boldsymbol{x}, \boldsymbol{y}, \boldsymbol{m};\btheta)$ is  twice continuously  differentiable for every $(\boldsymbol{x}, \boldsymbol{y}, \boldsymbol{m})$ in the support of  $(\bX, \bY(j), \bM(j))$ and is dominated by an integrable function $\boldsymbol{u}(\bX, \bY(j), \bM(j))$.

\hspace{5pt} (5) There exist a $C > 0$ and and integrable function $v(\bX, \bY(j), \bM(j))$, such that, for each entry $\psi_r, r = 1,\dots, q$, of $\bpsi$, $||\frac{\partial^2}{\partial\btheta\partial\btheta^t }\psi_r(j, \boldsymbol{x}, \boldsymbol{y}, \boldsymbol{m};\btheta)|| < v(\boldsymbol{x}, \boldsymbol{y}, \boldsymbol{m})$ for every $(j, \boldsymbol{x}, \boldsymbol{y}, \boldsymbol{m})$ in the support of $(A, \bX, \bY(j), \bM(j))$  and $\btheta$ with $||\btheta - \underline{\btheta}|| < C$.

\hspace{5pt} (6) $E\left[\Big|\Big|\frac{\partial}{\partial\btheta}\bpsi(j, \bX, \bY(j), \bM(j); \btheta)\Big|_{\btheta  =  \underline{\btheta}}\Big|\Big|^2\right] < \infty$ for $j \in \{0, \dots, J\}$ and 
\begin{displaymath}
    \sum_{j=0}^J \pi_j E\left[\frac{\partial}{\partial\btheta}\bpsi(j, \bX, \bY(j), \bM(j); \btheta)\Big|_{\btheta  =  \underline{\btheta}}\right]
\end{displaymath}
is invertible.

\section{Variance estimators in Theorem 1}\label{sec:variance-estimators}
For an M-estimator $\widehat\btheta \in \mathbb{R}^q$ defined by $P_n \bpsi(A,\bX,\bY,\bM;\btheta) = \bzero$, its sandwich variance estimator under simple randomization is defined as
\begin{align*}
\widetilde{\bV}_n(\bpsi, \widehat\btheta) &= \frac{1}{n} \left\{P_n\left[\frac{\partial}{\partial \btheta }\bpsi(A, \bX, Y, M; \btheta)\bigg|_{\btheta = \widehat{\btheta}}\right]\right\}^{-1}\left\{P_n\left[\bpsi(A, \bX, Y, M; \widehat\btheta)\bpsi(A, \bX, Y, M; \widehat\btheta)^t\right]\right\} \\
&\qquad \left\{P_n\left[\frac{\partial}{\partial \btheta }\bpsi(A, \bX, Y, M; \btheta)\bigg|_{\btheta = \widehat{\btheta}}\right]\right\}^{-1,t}.
\end{align*}
Since $\bDelta$ is embedded in $\btheta$, we can find $\mathbf{C} \in \mathbf{R}^{J\times q}$ such that $\bDelta = \mathbf{C} \btheta$. Then the variance estimator of $\widehat\bDelta$ under simple randomization is defined as $\mathbf{C} \widetilde{\bV}_n(\bpsi, \widehat\btheta)\mathbf{C}^\top$. 

For each $\textup{est} \in \{\textup{ANCOVA}, \textup{MMRM-I}, \textup{MMRM-II}, \textup{IMMRM} \}$, the variance estimator $\widetilde{\bV}_n^{\textup{(est)}}$ is calculated by $\mathbf{C} \widetilde{\bV}_n(\bpsi^{\textup{(est)}}, \widehat\btheta)\mathbf{C}^\top$. We note that the MMRM-I estimator and MMRM-II estimator share the same estimating equations $\bpsi^{\textup{(MMRM)}}$ with different specifications of $\bu(\bX)$ as described in Equation~(\ref{proof:MMRM}).

We next define $\bV_n^{\textup{(est)}}$ for $\textup{est} = \textup{ANCOVA}, \textup{MMRM-I}, \textup{MMRM-II}$. 
Define
\begin{align*}
    \widehat{Var}(\widehat{E}[\bX|S]) &= \sum_{s \in \mathcal{S}} \frac{(P_n I\{S=s\} \bX)(P_n I\{S=s\} \bX)^\top}{P_n I\{S=s\}} - (P_n \bX) (P_n \bX)^\top,\\
    \widehat{Var}(\bX) &= P_n \bX\bX^\top -( P_n \bX) (P_n \bX)^\top, \\
    \widehat{Cov}(\bX, Y_K(j)) &= P_n \frac{I\{A=j\}}{\pi_j}\bX Y_K -  P_n \bX P_n \frac{I\{A=j\}}{\pi_j} Y_K,\\
    \widehat{\bb}_{Kj} &= \widehat{Var}(\bX)^{-1}\widehat{Cov}(\bX, Y_K(j)),\\
    \widehat{\bb}_{K} &= \sum_{j=0}^J \pi_j\widehat{Var}(\bX)^{-1}\widehat{Cov}(\bX, Y_K(j)),\\
    \widehat{\mathbf{z}} &= (\widehat{\bb}_{K0} - \widehat{\bb}_{K},\dots, \widehat{\bb}_{KJ} - \widehat{\bb}_{K}), \\
    \widehat{\bv} &= (\widehat{\bb}_{K0} - \widehat{\bbeta}_{\bX},\dots, \widehat{\bb}_{KJ} - \widehat{\bbeta}_{\bX}), \\
    \bL &= (-\bone_J \quad \bI_J),
\end{align*}
where $\widehat{\bbeta}_{\bX}$ is the MLE of $\bbeta_{\bX}$ in the MMRM-I working model (2).

Following Equation~(\ref{V-ANCOVA}), we define
\begin{align*}
    \bV_n^{\textup{(ANCOVA)}} &= \widetilde\bV_n^{\textup{(ANCOVA)}} - \frac{1}{n}\bL[ diag\{\pi_j^{-1} (\widehat{\bb}_{Kj} - \widehat{\bb}_{K})^\top\widehat{Var}(\widehat{E}[\bX|S])(\widehat{\bb}_{Kj} - \widehat{\bb}_{K}): j = 0,\dots, J\}\\
    &\quad - \widehat{\mathbf{z}}^\top\widehat{Var}(\widehat{E}[\bX|S])  \widehat{\mathbf{z}}] \bL^\top.  
\end{align*}
Following Equation~(\ref{V-MMRM}), we define
\begin{align*}
    \bV_n^{\textup{(MMRM-I)}} &= \widetilde\bV_n^{\textup{(MMRM-I)}} - \frac{1}{n}\bL[ diag\{\pi_j^{-1} (\widehat{\bb}_{Kj} - \widehat{\bbeta}_{\bX})^\top\widehat{Var}(\widehat{E}[\bX|S])(\widehat{\bb}_{Kj} - \widehat{\bbeta}_{\bX}): j = 0,\dots, J\}\\
    &\quad - \widehat{\bv}^\top\widehat{Var}(\widehat{E}[\bX|S])  \widehat{\bv}] \bL^\top.  
\end{align*}
Similarly, we define
\begin{align*}
    \bV_n^{\textup{(MMRM-II)}} &= \widetilde\bV_n^{\textup{(MMRM-II)}} - \frac{1}{n}\bL[ diag\{\pi_j^{-1} (\widehat{\bb}_{Kj} - \widehat{\bb}_{K})^\top\widehat{Var}(\widehat{E}[\bX|S])(\widehat{\bb}_{Kj} - \widehat{\bb}_{K}): j = 0,\dots, J\}\\
    &\quad - \widehat{\mathbf{z}}^\top\widehat{Var}(\widehat{E}[\bX|S])  \widehat{\mathbf{z}}] \bL^\top.  
\end{align*}
\section{Lemmas}\label{sec:lemmas}

\begin{lemma}\label{lemma-Vm}
Let $\bSigma \in \mathbb{R}^{K\times K}$ be a positive definite matrix.
For each $\bm \in \{0,1\}^K\setminus\{\bzero_K\}$, let $n_{\bm} = \sum_{t=1}^K m_{t}$ be the number of ones in $\bm$ and $t_{\bm,1} < \dots < t_{\bm,n_{\bm}}$ denote the ordered list of locations of ones in $\bm$, i.e., $m_{t} = 1$ if $t \in \{t_{\bm,1}, \dots, t_{\bm,n_{\bm}}\}$ and $0$ otherwise. We define $\mathbf{D}_{\bm} = [\boldsymbol{e}_{t_{\bm,1}} \,\,\,\dots \,\,\, \boldsymbol{e}_{t_{\bm,n_{\bm}}}] \in \mathbb{R}^{K\times n_{\bm}}$ and $$\bV_{\bm}(\bSigma) = I\{\bm  \in \{0,1\}^K\setminus\{\bzero_K\}\} \mathbf{D}_{\bm}(\mathbf{D}_{\bm}^\top\bSigma \mathbf{D}_{\bm})^{-1}\mathbf{D}_{\bm}^\top \in \mathbb{R}^{K\times K},$$
which are deterministic functions of $\bm$ and $\bSigma$.

Let $\bM = (M_1, \dots, M_K)$ be a $K$-dimensional binary random vector taking values in $\{0,1\}^K$. We assume that $P(\bM = \bone_K) > 0$. Then the following statements hold.
\begin{enumerate}[(1)]
    \item $E[\bV_{\bM}(\bSigma)]$ is well-defined and positive definite.
    \item $\be_K^\top E[\bV_{\bM}(\bSigma)]^{-1}\be_K \le P(M_K = 1)^{-1}\be_K^\top \bSigma\be_K $. The equality holds if and only if either of the following conditions holds: (i) $K = 1$ or (ii) $P(M_t = 1, M_K = 0) \sigma_{t,K} = 0$ for $t = 1,\dots, K-1$, where $\sigma_{t,K} = \be_t^\top \bSigma \be_K$ is the $(t,K)$-th entry of $\bSigma$.
    
    \item Let $\mathbf{A} \in \mathbb{R}^{K\times K}$ be a positive definite matrix. Then \\
    $\be_K^\top E[\bV_{\bM}(\bSigma)]^{-1}\be_K \le \be_K^\top E[\bV_{\bM}(\mathbf{A})]^{-1}E[\bV_{\bM}(\mathbf{A}) \bSigma  \bV_{\bM}(\mathbf{A})]E[\bV_{\bM}(\mathbf{A})]^{-1}\be_K $. The equality holds if and only if\\ $P(\bM = \bm) \be_K^\top E[\bV_{\bM}(\mathbf{A})]^{-1}\bV_{\bm}(\mathbf{A}) = P(\bM = \bm) \be_K^\top E[\bV_{\bM}(\bSigma)]^{-1}\bV_{\bm}(\bSigma)$ for all $\bm \in \{0,1\}^K$.
    
    \item Let $\mathbf{B} \in \mathbb{R}^{K\times K}$ be a positive semi-definite matrix. Then \\
    $\be_K^\top \mathbf{B}\be_K \le \be_K^\top E[\bV_{\bM}(\bSigma)]^{-1}E[\bV_{\bM}(\bSigma) \mathbf{B} \bV_{\bM}(\bSigma)]E[\bV_{\bM}(\bSigma)]^{-1}\be_K $. The equality holds if and only if $P(\bM = \bm) \be_K^\top E[\bV_{\bM}(\bSigma)]^{-1}\bV_{\bm}(\bSigma) \mathbf{B}$ does not vary across $\bm \in \{0,1\}^K$.
    \item Letting $\mathbf{C} \in \mathbb{R}^{K\times K}$ be a positive definite matrix such that  $\mathbf{C} - \bSigma$ is positive semi-definite, then $\be_K^\top (E[\bV_{\bM}(\mathbf{C})]^{-1} - E[\bV_{\bM}(\bSigma)]^{-1})\be_K \ge \be_K^\top(\mathbf{C} - \bSigma)\be_K$. 
\end{enumerate}
\end{lemma}

\begin{proof}
\noindent (1) By definition, for any $\boldsymbol{m}\in \{0,1\}^K\setminus\{\bzero_K\}$, $\mathbf{D}_{\bm}$ is full column rank. Since $\bSigma$ is positive definite and $n_i \le K$, then $\mathbf{D}_{\bm}^\top\bSigma\mathbf{D}_{\bm}$ is positive definite and $\bV_{\bm} = \mathbf{D}_{\bm}(\mathbf{D}_{\bm}^\top\bSigma\mathbf{D}_{\bm})^{-1}\mathbf{D}_{\bm}^\top$ is positive semi-definite. 
In particular, for the case $\bm = \bone_K$, $\bV_{\bm} = \bSigma^{-1}$ and so is positive definite.
Since $E[\bV_{\bM}(\bSigma)] = \sum_{\boldsymbol{m}\in \{0,1\}^K\setminus\{\bzero_K\}}\bV_{\boldsymbol{m}}P(\bM =\boldsymbol{m})$ and we assume $P(\bM = \bone_K) >0$, then for any $\boldsymbol{x} \in \mathbb{R}^K$,
\begin{align*}
  \boldsymbol{x}^\top E[\bV_{\bM}]\boldsymbol{x} &= 
   \boldsymbol{x}^\top\bV_{\bone_K}\boldsymbol{x}P(\bM =\bone_K) +
  \sum_{\boldsymbol{m}\in \{0,1\}^K\setminus \{\bzero_K, \bone_K\}}\boldsymbol{x}^\top\bV_{\boldsymbol{m}}\boldsymbol{x}P(\bM =\boldsymbol{m}) \\
  &\ge \boldsymbol{x}^\top\bV_{\bone_K}\boldsymbol{x}P(\bM =\bone_K) \\
  & > 0,
\end{align*}
which implies that $E[\bV_{\bM}]$ is positive definite and so invertible.

\vspace{10pt}
\noindent (2) Consider the case $K=1$. we have that $\bSigma$ reduces to a positive number $\sigma$ and $V_{m}(\sigma) = I\{m = 1\} \sigma^{-1}$. Then $\be_K^\top E[V_{m}(\sigma)]^{-1}\be_K  = \frac{1}{P(m=1)}\sigma = \frac{1}{P(M_K = 1)}\be_K^\top \bSigma\be_K $. 
We next consider the case that $K\ge 2$.
Define $\bOmega_K = \{\bm \in \{0,1\}^K\setminus\{\be_K\}: m_K = 1\}$ and $\bOmega_{-K} = \{\bm \in \{0,1\}^K\setminus\{\bzero_K\}: m_K = 0\}$. Then we have
\begin{align*}
    \mathbf{D}_{\bm} = \left(\begin{array}{cc}
    \widetilde{\mathbf{D}}_{\bm}     & \bzero_{K-1} \\
    \bzero_{\sum_{t=1}^K m_t-1}^\top     &  1
    \end{array}\right) \textrm{if } \bm \in \bOmega_K\ \textrm{and } \mathbf{D}_{\bm} = \left(\begin{array}{c}
    \widetilde{\mathbf{D}}_{\bm}  \\
    \bzero_{\sum_{t=1}^K m_t}^\top  
    \end{array}\right) \textrm{if } \bm \in \bOmega_{-K},
\end{align*}
where $\widetilde{\mathbf{D}}_{\bm} \in \mathbb{R}^{(K-1)\times(\sum_{t=1}^K m_t-1)}$ is a matrix taking the first $K-1$ rows and first $(\sum_{t=1}^K m_t-1)$ columns of $\mathbf{D}_{\bm}$ if $\bm \in \bOmega_K$ and $\widetilde{\mathbf{D}}_{\bm} \in \mathbb{R}^{(K-1)\times(\sum_{t=1}^K m_t)}$ is the first $K-1$ rows of $\widetilde{\mathbf{D}}_{\bm}$ if $\bm \in \bOmega_{-K}$. 
We further define $\bSigma_{-K,-K} \in \mathbb{R}^{(K-1)\times (K-1)}$, $\bSigma_{-K,K} \in \mathbb{R}^{K-1}$ and $\sigma = \be_K^\top \bSigma \be_K$ such that
\begin{align*}
    \bSigma = \left(\begin{array}{cc}
    \bSigma_{-K,-K}     &  \bSigma_{-K,K}\\
    \bSigma_{-K,K}^\top     & \sigma
    \end{array}\right).
\end{align*}
Using the formula of block matrix inversion, we can compute that, if $\bm \in \bOmega_K$, then
\begin{align*}
  \bV_{\bm}(\bSigma) =  \left(\begin{array}{cc}
    \bA_{\bm} &  -\sigma^{-1}\bA_{\bm}\bSigma_{-K,K} \\
   -\sigma^{-1}\bSigma_{-K,K}^\top\bA_{\bm}  & \sigma^{-1} + \sigma^{-2}\bSigma_{-K,K}^\top\bA_{\bm}\bSigma_{-K,K}
  \end{array}\right),
\end{align*}
and, if $\bm \in \bOmega_{-K}$, then
\begin{align*}
  \bV_{\bm}(\bSigma) =  \left(\begin{array}{cc}
    \bA_{\bm}   &  \bzero_{K-1} \\
   \bzero_{K-1}^\top & 0
  \end{array}\right),
\end{align*}
where $\bA_{\bm} = \widetilde{\mathbf{D}}_{\bm}\{\widetilde{\mathbf{D}}_{\bm}^\top(\bSigma_{-K,-K}-\bSigma_{-K,K}\bSigma_{-K,K}^\top/\sigma)\widetilde{\mathbf{D}}_{\bm}\}^{-1}\widetilde{\mathbf{D}}_{\bm}^\top \in \mathbb{R}^{(K-1)\times (K-1)}$ if $\bm \in \Omega_K$ and $\bA_{\bm} = \widetilde{\mathbf{D}}_{\bm}\{\widetilde{\mathbf{D}}_{\bm}^\top\bSigma_{-K,-K}\widetilde{\mathbf{D}}_{\bm}\}^{-1}\widetilde{\mathbf{D}}_{\bm}^\top \in \mathbb{R}^{(K-1)\times (K-1)}$ if $\bm \in \Omega_{-K}$. Here $\bA_{\bm}$ is well defined for each $\bm \in \Omega_K \bigcup \Omega_{-K}$ since $\widetilde{\mathbf{D}}_{\bm}^\top\bSigma\widetilde{\mathbf{D}}_{\bm}$ is positive definite. 
In addition, if $\bm = \be_K$, then $\mathbf{D}_{\bm} = \be_K$ and $\bV_{\bm} = \sigma^{-1}\be_K\be_K^\top$. 
Hence
\begin{align*}
    E[\bV_{\bM}(\bSigma)] &= \left(\begin{array}{cc}
    \sum_{\bm \in \Omega_K\bigcup\Omega_{-K}}p_{\bm} \bA_{\bm}    &  -\sum_{\bm \in \Omega_K}p_{\bm}\sigma^{-1}\bA_{\bm}\bSigma_{-K,K}\\
 -\sum_{\bm \in \Omega_K}p_{\bm}\sigma^{-1}\bSigma_{-K,K}^\top\bA_{\bm}  &\sum_{\bm \in \Omega_K}p_{\bm} (\sigma^{-1} + \sigma^{-2}\bSigma_{-K,K}^\top\bA_{\bm}\bSigma_{-K,K}) + p_{\be_K} \sigma^{-1} 
    \end{array}\right).
\end{align*}
Since we have shown that $E[\bV_{\bM}(\bSigma)]$ is positive definite and $P(M_K = 1) = p_{\be_K} + \sum_{\bm \in \Omega_K}p_{\bm}$, by using the formula of block matrix inversion again, we have
\begin{align*}
  & (\be_K^\top E[\bV_{\bm}(\bSigma)]^{-1}\be_K)^{-1}\\
  &= P(M_K = 1)\sigma^{-1} + \sum_{\bm \in \Omega_K}p_{\bm}  \sigma^{-2}\bSigma_{-K,K}^\top\bA_{\bm}\bSigma_{-K,K}\\
   &\  -(\sum_{\bm \in \Omega_K}p_{\bm}\sigma^{-1}\bSigma_{-K,K}^\top\bA_{\bm}) 
    (\sum_{\bm \in \Omega_K \bigcup \Omega_{-K}}p_{\bm} \bA_{\bm} )^{-1}(\sum_{\bm \in \Omega_K}p_{\bm}\sigma^{-1}\bA_{\bm}\bSigma_{-K,K})\\
    &= P(M_K = 1)\sigma^{-1} + \sigma^{-2}\bSigma_{-K,K}^\top\{\mathbf{G}_{\Omega_{K}} - \mathbf{G}_{\Omega_{K}}(\mathbf{G}_{\Omega_{K}} + \mathbf{G}_{\Omega_{-K}})^{-1}\mathbf{G}_{\Omega_{K}}\}\bSigma_{-K,K},
\end{align*}
where $\mathbf{G}_{\Omega_{K}} =   \sum_{\bm \in \Omega_K}p_{\bm} \bA_{\bm}$ and $\mathbf{G}_{\Omega_{-K}} =   \sum_{\bm \in \Omega_{-K}}p_{\bm} \bA_{\bm}$.

For $\mathbf{G}_{\Omega_K}$, we note that $\bone_K \in \Omega_K$ with $p_{\bone_K} > 0$ and $\mathbf{G}_{\bone_K} = \{\bSigma_{-K,-K}-\bSigma_{-K,K}\bSigma_{-K,K}^\top/\sigma\}^{-1}$ is positive definite. Hence $\mathbf{G}_{\Omega_K}$ is positive definite. Furthermore, $\mathbf{G}_{\Omega_{-K}}\succeq \bzero$ by definition. Then
\begin{align*}
  &\mathbf{G}_{\Omega_{K}}^{-1}  - (\mathbf{G}_{\Omega_{K}} + \mathbf{G}_{\Omega_{-K}})^{-1} \\
  &= (\mathbf{G}_{\Omega_{K}} + \mathbf{G}_{\Omega_{-K}})^{-1}\{(\mathbf{G}_{\Omega_{K}} + \mathbf{G}_{\Omega_{-K}}) \mathbf{G}_{\Omega_{K}}^{-1}(\mathbf{G}_{\Omega_{K}} + \mathbf{G}_{\Omega_{-K}}) - (\mathbf{G}_{\Omega_{K}} + \mathbf{G}_{\Omega_{-K}})\}(\mathbf{G}_{\Omega_{K}} + \mathbf{G}_{\Omega_{-K}})^{-1} \\
  &= (\mathbf{G}_{\Omega_{K}} + \mathbf{G}_{\Omega_{-K}})^{-1}(\mathbf{G}_{\Omega_{-K}}+\mathbf{G}_{\Omega_{-K}}\mathbf{G}_{\Omega_{K}}^{-1}\mathbf{G}_{\Omega_{-K}})(\mathbf{G}_{\Omega_{K}} + \mathbf{G}_{\Omega_{-K}})^{-1} \\
  & \succeq \bzero.
\end{align*}
Hence 
\begin{align*}
  & (\be_K^\top E[\bV_{\bm}(\bSigma)]^{-1}\be_K)^{-1}\\
    &= P(M_K = 1)\sigma^{-1} + \sigma^{-2}\bSigma_{-K,K}^\top\{\mathbf{G}_{\Omega_{K}} - \mathbf{G}_{\Omega_{K}}(\mathbf{G}_{\Omega_{K}} + \mathbf{G}_{\Omega_{-K}})^{-1}\mathbf{G}_{\Omega_{K}}\}\bSigma_{-K,K} \\
    &=P(M_K = 1)\sigma^{-1} + \sigma^{-2}\bSigma_{-K,K}^\top\mathbf{G}_{\Omega_{K}}\{\mathbf{G}_{\Omega_{K}}^{-1}  - (\mathbf{G}_{\Omega_{K}} + \mathbf{G}_{\Omega_{-K}})^{-1}\}\mathbf{G}_{\Omega_{K}}\bSigma_{-K,K}  \\
    & \ge P(M_K = 1)\sigma^{-1},
\end{align*}
which completes the proof of $\be_K^\top E[\bV_{\bM}(\bSigma)]^{-1}\be_K \le\frac{1}{P(M_K = 1)}\be_K^\top \bSigma\be_K $.

We next examine when $\be_K^\top E[\bV_{\bM}(\bSigma)]^{-1}\be_K = \frac{1}{P(M_K = 1)}\be_K^\top \bSigma\be_K $. Since $\mathbf{G}_{\Omega_{K}}^{-1}$ is positive semi-definite, then the above derivation shows that the equality holds if and only if $\bSigma_{-K,K}^\top\mathbf{G}_{\Omega_{K}}(\mathbf{G}_{\Omega_{K}} + \mathbf{G}_{\Omega_{-K}})^{-1}\mathbf{G}_{\Omega_{-K}} =\bzero$. Noting that
\begin{displaymath}
\mathbf{G}_{\Omega_{K}}(\mathbf{G}_{\Omega_{K}} + \mathbf{G}_{\Omega_{-K}})^{-1}\mathbf{G}_{\Omega_{-K}} =\mathbf{G}_{\Omega_{K}}- \mathbf{G}_{\Omega_{K}}(\mathbf{G}_{\Omega_{K}} + \mathbf{G}_{\Omega_{-K}})^{-1}\mathbf{G}_{\Omega_{K}}
\end{displaymath}
is symmetric, the equality holds if and only if $\bSigma_{-K,K}^\top\mathbf{G}_{\Omega_{-K}}(\mathbf{G}_{\Omega_{K}} + \mathbf{G}_{\Omega_{-K}})^{-1}\mathbf{G}_{\Omega_{K}} =\bzero$. Since $(\mathbf{G}_{\Omega_{K}} + \mathbf{G}_{\Omega_{-K}})^{-1}\mathbf{G}_{\Omega_{K}}$ is positive definite, we get
\begin{align*}
    &\be_K^\top E[\bV_{\bM}(\bSigma)]^{-1}\be_K = \frac{1}{P(M_K = 1)}\be_K^\top \bSigma\be_K\\ 
    &\Leftrightarrow \bSigma_{-K,K}^\top\mathbf{G}_{\Omega_{-K}}(\mathbf{G}_{\Omega_{K}} + \mathbf{G}_{\Omega_{-K}})^{-1}\mathbf{G}_{\Omega_{K}} = \bzero\\
    &\Leftrightarrow    \bSigma_{-K,K}^\top\mathbf{G}_{\Omega_{-K}} \bSigma_{-K,K} = 0\\    
    &\Leftrightarrow  \sum_{\bm \in \Omega_{-K}}p_{\bm}  \bSigma_{-K,K}^\top\widetilde{\mathbf{D}}_{\bm}\{\widetilde{\mathbf{D}}_{\bm}^\top\bSigma_{-K,-K}\widetilde{\mathbf{D}}_{\bm}\}^{-1}\widetilde{\mathbf{D}}_{\bm}^\top\bSigma_{-K,K} = 0\\
    &\Leftrightarrow p_{\bm}  \bSigma_{-K,K}^\top\widetilde{\mathbf{D}}_{\bm} = 0\ \textrm{for each } \bm \in \Omega_{-K}\\
    &\Leftrightarrow P(M_t = 1, M_K = 0) \sigma_{t, K} = 0 \ \textrm{for each } t = 1, \dots, K-1,
\end{align*}
where $\sigma_{j, K} = \be_j^\top \bSigma \be_K$, which completes the proof.

\vspace{10pt}
\noindent (3) Denote $\boldsymbol{c}_{\bM}^\top(\mathbf{A}) = \be_K^\top E[\bV_{\bM}(\mathbf{A})]^{-1} \bV_{\bM}(\mathbf{A})$ and $\boldsymbol{c}_{\bM}^\top(\bSigma) = \be_K^\top E[\bV_{\bM}(\bSigma)]^{-1} \bV_{\bM}(\bSigma)$.
We have the following derivation:
\begin{align*}
& E[\{\boldsymbol{c}_{\bM}^\top(\mathbf{A}) - \boldsymbol{c}_{\bM}^\top(\bSigma)\} \bSigma \{\boldsymbol{c}_{\bM}(\mathbf{A}) - \boldsymbol{c}_{\bM}(\bSigma)\}] \\
&= E[\boldsymbol{c}_{\bM}^\top(\mathbf{A}) \bSigma \boldsymbol{c}_{\bM}(\mathbf{A})] - E[\boldsymbol{c}_{\bM}^\top(\mathbf{A}) \bSigma \boldsymbol{c}_{\bM}(\bSigma)] - E[ \boldsymbol{c}_{\bM}^\top(\bSigma) \bSigma \boldsymbol{c}_{\bM}(\mathbf{A})] + E[ \boldsymbol{c}_{\bM}^\top(\bSigma) \bSigma \boldsymbol{c}_{\bM}(\bSigma)]\\
&= E[\boldsymbol{c}_{\bM}^\top(\mathbf{A}) \bSigma \boldsymbol{c}_{\bM}(\mathbf{A})] -  \be_K^\top E[\bV_{\bM}(\bSigma)]^{-1}\be_K,
\end{align*}
where the last equation comes from the fact that
\begin{align*}
 \bV_{\bM}(\mathbf{A}) \bSigma \bV_{\bM}(\bSigma)
   &=  \mathbf{D}_{\bm}(\mathbf{D}_{\bm}^\top\mathbf{A} \mathbf{D}_{\bm})^{-1}\mathbf{D}_{\bm}^\top\bSigma \mathbf{D}_{\bm}(\mathbf{D}_{\bm}^\top\bSigma \mathbf{D}_{\bm})^{-1}\mathbf{D}_{\bm}^\top = \bV_{\bM}(\mathbf{A}) \\
  \bV_{\bM}(\bSigma) \bSigma \bV_{\bM}(\bSigma)
   &=  \mathbf{D}_{\bm}(\mathbf{D}_{\bm}^\top\bSigma \mathbf{D}_{\bm})^{-1}\mathbf{D}_{\bm}^\top\bSigma \mathbf{D}_{\bm}(\mathbf{D}_{\bm}^\top\bSigma \mathbf{D}_{\bm})^{-1}\mathbf{D}_{\bm}^\top = \bV_{\bM}(\bSigma).
\end{align*}
Since $\bSigma$ is positive definite, then we have
$E[\boldsymbol{c}_{\bM}^\top(\mathbf{A}) \bSigma \boldsymbol{c}_{\bM}(\mathbf{A})] \ge  \be_K^\top E[\bV_{\bM}(\bSigma)]^{-1}\be_K$, which is the desired inequality. The equality holds if and only if $p_{\bm} \{\boldsymbol{c}_{\bm}^\top(\mathbf{A}) - \boldsymbol{c}_{\bM}^\top(\bSigma)\} = \bzero_K^\top$ for all $\bm \in \{0,1\}^K\setminus\{\bzero_K\}$.

\vspace{10pt}
\noindent (4) 
Define $\boldsymbol{x}_{\bM}^\top= \be_K^\top E[\bV_{\bM}(\bSigma)]^{-1}\bV_{\bM}(\bSigma) \mathbf{B}^{\frac{1}{2}}$. Here $\mathbf{B}^{\frac{1}{2}}$ is well-defined since $\mathbf{B}$ is positive semi-definite. Then we have
\begin{align*}
    &\be_K^\top E[\bV_{\bM}(\bSigma)]^{-1}E[\bV_{\bM}(\bSigma) \mathbf{B} \bV_{\bM}(\bSigma)]E[\bV_{\bM}(\bSigma)]^{-1}\be_K - \be_K^\top \mathbf{B}\be_K \\
    &= E[\boldsymbol{x}_{\bM}^\top\boldsymbol{x}_{\bM}] - E[\boldsymbol{x}_{\bM}]^\top E[\boldsymbol{x}_{\bM}] \\
    &= \sum_{t=1}^K Var(\boldsymbol{x}_{\bM}^\top\be_t) \\
    & \ge 0.
\end{align*}
The equality holds if and only if $Var(\boldsymbol{x}_{\bM}^\top\be_t) = 0$ for $t = 1, \dots, K$, which is equivalent to $\boldsymbol{x}_{\bM}^\top$ being a constant vector.

\vspace{10pt}
\noindent (5) Define $\mathbf{B} = \mathbf{C} - \bSigma$. The statement is proved by the following derivation:
\begin{align*}
    &\be_K^\top E[\bV_{\bM}(\mathbf{C})]^{-1}\be_K\\
    &= \be_K^\top E[\bV_{\bM}(\mathbf{C})]^{-1}E[\bV_{\bM}(\mathbf{C})\mathbf{C}\bV_{\bM}(\mathbf{C})]E[\bV_{\bM}(\mathbf{C})]^{-1}\be_K\\
    &= \be_K^\top E[\bV_{\bM}(\mathbf{C})]^{-1}E[\bV_{\bM}(\mathbf{C})(\mathbf{B}+\bSigma)\bV_{\bM}(\mathbf{C})]E[\bV_{\bM}(\mathbf{C})]^{-1}\be_K \\
    &\ge \be_K^\top E[\bV_{\bM}(\mathbf{C})]^{-1}E[\bV_{\bM}(\mathbf{C})\mathbf{B}\bV_{\bM}(\mathbf{C})]E[\bV_{\bM}(\mathbf{C})]^{-1}\be_K + \be_K^\top E[\bV_{\bM}(\mathbf{\bSigma})]^{-1}\be_K \\
    &\ge \be_K^\top\mathbf{B}\be_K +  \be_K^\top E[\bV_{\bM}(\mathbf{\bSigma})]^{-1}\be_K,
\end{align*}
where the first inequality results from Lemma~\ref{lemma-Vm} (3), and the second inequality comes from Lemma~\ref{lemma-Vm} (4).
\end{proof}

\begin{lemma}[Kronecker product]\label{lemma:kronecker-product}
Let $\mathbf{A} \in \mathbb{R}^{n_1 \times n_2}, \mathbf{B} \in \mathbb{R}^{n_3 \times n_4}, \mathbf{C} \in \mathbb{R}^{n_2 \times n_5}, \mathbf{D} \in \mathbb{R}^{n_4 \times n_6}$ be random matrices. Then
\begin{enumerate}[(1)]
    \item $(\mathbf{A} \otimes \mathbf{B}) (\mathbf{C} \otimes \mathbf{D}) = (\mathbf{A}\mathbf{C}) \otimes (\mathbf{B}\mathbf{D})$, 
    \item  If $\mathbf{A}$ is independent of $(\mathbf{B}, \mathbf{C})$, then $E[\mathbf{A}\otimes \mathbf{B}] = E[\mathbf{A}] \otimes E[\mathbf{B}]$ and $E[(\mathbf{A} \mathbf{C})\otimes \mathbf{B}] =  E[(E[\mathbf{A}] \mathbf{C})\otimes\mathbf{B}]$. 
    \item $(\mathbf{A} \otimes \mathbf{B})^{-1} =\mathbf{A}^{-1}  \otimes \mathbf{B}^{-1}$ if $\mathbf{A}$ and $\mathbf{B}$ are invertible.
    \item $(\mathbf{A} \otimes \mathbf{B})^\top =\mathbf{A}^\top  \otimes \mathbf{B}^\top$.
    \item Suppose $n_1 = n_2$, $n_3 = n_4$, $\mathbf{A}$ has eigenvalues $\lambda_1,\dots, \lambda_{n_1}$, and $\mathbf{B}$ has eigenvalues $\mu_1,\dots, \mu_{n_3}$. Then $\mathbf{A} \otimes \mathbf{B}$ has eigenvalues $\lambda_i\mu_j$ for each $i = 1,\dots, n_1$ and $j = 1, \dots, n_3$.
\end{enumerate}
\end{lemma}

\begin{lemma}\label{lemma-bugni}
Given Assumption 1, for each $j = 0,\dots, J$, let $\bZ_i(j) = \bh_j(\bY_i(j), \bM_i(j), \bX_i) \in \mathbb{R}^q$ for some function $\bh_j$ such that $E[||\bZ_i(j)\bZ_i(j)^\top||] < \infty$. 
Then under stratified randomization,
\begin{align*}
    \frac{1}{\sqrt{n}} \sum_{i=1}^n \sum_{j=0}^J \bigg(I\{A_i=j\} \bZ_i(j) - \pi_j E[\bZ(j)]\bigg) \xrightarrow{d} N(\bzero, \mathbf{G}),
\end{align*}
where
\begin{align*}
  \mathbf{G} = \sum_{j=0}^J \pi_j E[Var\{\bZ(j)|S\}] + Var\left(\sum_{j=0}^J \pi_j E[\bZ(j)|S]\right).
\end{align*}
Furthermore, $$ \sum_{j=0}^J \pi_j E[\bZ(j)\bZ(j)^\top] - E\left[\sum_{j=0}^J \pi_j\bZ(j)\right]E\left[\sum_{j=0}^J \pi_j\bZ(j)\right]^\top-\mathbf{G} = E[\mathbf{U} (diag\{\boldsymbol{\pi}\} -\boldsymbol{\pi}\boldsymbol{\pi}^\top ) \mathbf{U}^\top]$$ is positive semi-definite, where
\begin{align*}
    \mathbf{U} &= (E[\bZ(0)|S], \dots, E[\bZ(J)|S]) \\
    \boldsymbol{\pi} &= (\pi_0, \dots, \pi_J)^\top.
\end{align*}
\end{lemma}
\begin{proof}
Let $\mathcal{S} = \{1,\dots, R\}$ denote the levels in $S$.
Using the fact that $E[\bZ_i(j) |S=S_i] = \sum_{s\in\mathcal{S}} I\{S_i = s\}E[\bZ(j) |S = s]$ and $E[\bZ_i(j)] = \sum_{s \in \mathcal{S}}P(S=s)E[\bZ(j) |S=s]$, we have
\begin{align*}
  &  \frac{1}{\sqrt{n}} \sum_{i=1}^n \sum_{j=0}^J \bigg(I\{A_i=j\} \bZ_i(j)  - \pi_j E[\bZ(j)]\bigg) \\
  &= \frac{1}{\sqrt{n}} \sum_{i=1}^n \sum_{j=0}^J \sum_{s\in\mathcal{S}} I\{A_i=j, S_i =s\} \bigg( \bZ_i(j)  - E[\bZ(j)|S = S_i]\bigg) \\
  &\quad + \sum_{s\in\mathcal{S}} \sqrt{n}  \left(\frac{\sum_{i=1}^n I\{S_i =s\}}{n} - P(S=s)\right) \sum_{j=0}^J \pi_j E[\bZ(j)|S=s] \\
  &\quad + \sum_{s\in\mathcal{S}}\sum_{j=0}^J \frac{1}{\sqrt{n}}\sum_{i=1}^n\left(I\{A_i=j, S_i=s\} - \pi_jI\{S_i=s\}\right) E[\bZ(j)|S=s] \\
  &= (\bone_{(J+1)L}\otimes \bI_q)^\top \mathbb{L}_n^{(1)} + \bu^\top \mathbb{L}_n^{(2)} + \bv_n^\top \mathbb{L}_n^{(3)},
\end{align*}
where
\begin{align*}
    \mathbb{L}_n^{(1)} &= \left(\frac{1}{\sqrt{n}}\sum_{i=1}^n I\{A_i=j, S_i =s\} \{ \bZ_i(j)  - E[\bZ(j)|S = S_i]\}:(j,s)\in\{0,\dots, J\}\times \mathcal{S}\right),\\
    \mathbb{L}_n^{(2)} &= \left(\sqrt{n}  \left\{\frac{\sum_{i=1}^n I\{S_i =s\}}{n} - P(S=s)\right\}: s \in \mathcal{S}\right),\\
    \mathbb{L}_n^{(3)} &=\left(\sqrt{n}  \left\{\frac{\sum_{i=1}^n I\{A_i = j, S_i =s\}}{\sum_{i=1}^n I\{S_i =s\}} - \pi_j\right\}: (j,s)\in\{0,\dots, J\}\times \mathcal{S}\right), \\
    \bu &= \left(\sum_{j=0}^J \pi_j E[\bZ(j)|S=1],\dots, \sum_{j=0}^J \pi_j E[\bZ(j)|S=R]\right)^\top, \\
    \bv_n &= \left(\frac{\sum_{i=1}^n I\{S_i =s\}}{n}E[\bZ(j)|S=s]:(j,s)\in\{0,\dots, J\}\times \mathcal{S}\right)^\top,
\end{align*}
where $(\boldsymbol{x}_{js}:(j,s)\in\{0,\dots, J\}\times \mathcal{S}) = (\boldsymbol{x}_{01}^\top, \dots, \boldsymbol{x}_{0R}^\top,\dots, \boldsymbol{x}_{J1}^\top, \dots, \boldsymbol{x}_{JR}^\top)^\top$ and $(\boldsymbol{x}_{js}:s\in\mathcal{S}) = (\boldsymbol{x}_{j1}^\top, \dots, \boldsymbol{x}_{jR}^\top)^\top$ for any vectors $\boldsymbol{x}_{js} \in \mathbb{R}^q$.

We next show that
\begin{align*}
    \left(\begin{array}{c}
    \mathbb{L}_n^{(1)}   \\
    \mathbb{L}_n^{(2)} \\
    \mathbb{L}_n^{(3)}
    \end{array}\right) \xrightarrow{d} N\left(\left(\begin{array}{c}
        \bzero \\
         \bzero \\
         \bzero
    \end{array}\right), \left(\begin{array}{ccc}
    \bSigma_1 & \bzero & \bzero  \\
    \bzero & diag\{p_\mathcal{S}\} - p_\mathcal{S}p_\mathcal{S}^\top & \bzero\\
    \bzero & \bzero & \bzero_{(J+1)L\times (J+1)L}
    \end{array}\right)\right),
\end{align*}
where
\begin{align*}
    \bSigma_1 &= bdiag\{\pi_jP(S=s)Var\{Z(j)|S=s\}: (j,s)\in\{0,\dots, J\}\times \mathcal{S}\},\\
    p_\mathcal{S} &= (P(S=1), \dots, P(S=R))^\top,
\end{align*}
where $bdiag\{\bV_{js}: (j,s)\in\{0,\dots, J\}\times \mathcal{S}\}$ represents a block diagonal matrix with $\bV_{js}$ being the $\{(s-1)R+j+1\}$-th diagonal block. The proof can be found in Lemma C.1 and C.2 of Appendix C of \cite{bugni2019inference}. The only difference is that $\bZ_i(j)$ is substituted for $Y_i(j)$ and all the arguments still hold.

By the delta method, we have $(\bone_{(J+1)L}\otimes \bI_q)^\top \mathbb{L}_n^{(1)} + \bu^\top \mathbb{L}_n^{(2)} \xrightarrow{d} N(0, \mathbf{G})$ and $\bv_n \xrightarrow{P} \bv$ with $\bv = \left(P(S=s)E[\bZ(j)|S=s]:(j,s)\in\{0,\dots, J\}\times \mathcal{S}\right)$. Using Slutsky's theorem twice, we get the desired asymptotic normal distribution.

Finally, we have the following derivation:
\begin{align*}
   &  \sum_{j=0}^J \pi_j E[\bZ(j)\bZ(j)^\top] - E\left[\sum_{j=0}^J \pi_j\bZ(j)\right]E\left[\sum_{j=0}^J \pi_j\bZ(j)\right]^\top-\mathbf{G} \\
   &= \sum_{j=0}^J \pi_j E[\bZ(j)\bZ(j)^\top] - E\left[\sum_{j=0}^J \pi_j\bZ(j)\right]E\left[\sum_{j=0}^J \pi_j\bZ(j)\right]^\top \\
   & - \sum_{j=0}^J \pi_j E[\bZ(j)\bZ(j)^\top] + \sum_{j=0}^J \pi_j E[E[\bZ(j)|S]E[\bZ(j)|S]^\top] \\
   & - E\left[E\left[\sum_{j=0}^J \pi_j\bZ(j)\bigg|S\right]E\left[\sum_{j=0}^J \pi_j\bZ(j)\bigg|S\right]^\top\right] + E\left[\sum_{j=0}^J \pi_j\bZ(j)\right]E\left[\sum_{j=0}^J \pi_j\bZ(j)\right]^\top \\
   &= \sum_{j=0}^J \pi_j E[E[\bZ(j)|S]E[\bZ(j)|S]^\top] - E\left[E\left[\sum_{j=0}^J \pi_j\bZ(j)\bigg|S\right]E\left[\sum_{j=0}^J \pi_j\bZ(j)\bigg|S\right]^\top\right] \\
   &= E[\mathbf{U} (diag\{\boldsymbol{\pi}\} -\boldsymbol{\pi}\boldsymbol{\pi}^\top ) \mathbf{U}^\top].
\end{align*}
Since $diag\{\boldsymbol{\pi}\} \succeq\boldsymbol{\pi}\boldsymbol{\pi}^\top$, then we get $E[\mathbf{U} (diag\{\boldsymbol{\pi}\} -\boldsymbol{\pi}\boldsymbol{\pi}^\top ) \mathbf{U}^\top]$ is positive semi-definite.
\end{proof}

\begin{lemma}\label{lemma:observed-data-distribution}
Given Assumption 1, under simple or stratified randomization, each data vector $(A_i,\bY_i,\bM_i, \bX_i)$ is identically distributed and, for $i = 1, \dots, n$ , $A_i$ is independent of $\bW_i$ and $P(A_i = j) = \pi_j$. 

Let $P^*$ denote the distribution of $(A_i,\bY_i,\bM_i, \bX_i)$ and $E^*$ be the associated expectation. Define $Z = f(\bY, \bM, \bX)$ and $Z(j) = f(\bY(j), \bM(j), \bX)$ such that $E[Z(j)^2] < \infty$ for $j =0, \dots, J$. Then $E^*[I\{A=j\} Z] = \pi_j E[Z(j)]$ and $E^*[I\{A=j\} Z|S] = \pi_j E[Z(j)|S]$ for $j = 0,\dots, J$.
\end{lemma}
\begin{proof}
See Lemma 4 and Lemma 3 in the Supplementary of \cite{wang2021model}. The only difference of proof is that $A = 1$ is substituted by $A=j$ for $j = 1,\dots, J$, and $(\bY, \bM)$ are substituted for $(Y,M)$.
\end{proof}

\section{Proofs}\label{sec:proofs}
\subsection{Proof of Theorem 1}
\textbf{Outline of the proof:}
Consider the estimator $\widehat{\bDelta}^{\textup{(est)}}$ for each \\ $\textup{est} \in \{\textup{ANCOVA},\ \textup{MMRM-I}, \textup{MMRM-II},\ \textup{IMMRM}\}$.
We first show that $\bDelta^{\textup{(est)}}$ is an M-estimator. We then apply Theorem 1 of \cite{wang2021model} to show that $\bDelta^{\textup{(est)}}$ is model-robust and asymptotically linear with influence function $IF^{\textup{(est)}}$.
The influence function $IF^{\textup{(est)}}$ is the same under simple and stratified randomization.
Next, we prove the asymptotic normality by Lemma~\ref{lemma-bugni}, which is a central limit theorem for sums of random vectors under stratified randomization that generalizes Lemma~B.2 of \cite{bugni2019inference}.
Next, we calculate $IF^{\textup{(est)}}$  and derive the asymptotic covariance matrix i.e., $\bV^{\textup{(est)}}$ and $\widetilde\bV^{\textup{(est)}}$, for which the detailed algebra is given in Lemma~\ref{lemma:if-MMRM} and Lemma~\ref{lemma:if-IMMRM}. 
Finally, we compare the asymptotic covariance matrices, where Lemma~\ref{lemma-Vm} is used to handle missing data.

\begin{proof}[Proof of Theorem 1]
The ANCOVA estimator can be computed by solving \\ $P_n\bpsi^{\textup{(ANCOVA)}}(A, \bX, \bY, \bM; \btheta) = \bzero$, where
\begin{equation}
\bpsi^{\textup{(ANCOVA)}}(A, \bX, \bY, \bM; \btheta) =
    I\{M_K=1\}(Y_{K} - \beta_{0} -  \sum_{j = 1}^ J \beta_{Aj}I\{A = j\} - \bbeta_{\bX}^\top \bX) \left(\begin{array}{c}
        1  \\
        \bA \\
        \bX
    \end{array}\right), \label{psi:ANCOVA}
\end{equation}
where $\bA = (I\{A=1\}, \dots, I\{A=J\})^\top$ is a vector of treatment assignment indicator and $\btheta = (\beta_{0}, \beta_{A1}, \dots, \beta_{AJ}, \bbeta_{\bX}^\top)^\top$. Hence $\widehat{\bDelta}^{\textup{(ANCOVA)}}$ is an M-estimator.

The MMRM-I and MMRM-II working models  can be rewritten in one formula below:
\begin{equation}\label{proof:MMRM}
      \bY = \bbeta_0 + (\bI_K \otimes \bA)^\top\bbeta_{\bA} + \bu(\bX)^\top \bbeta_{\bu(\bX)} + \boldsymbol{\varepsilon},
\end{equation}
where $\bbeta_0 = (\beta_{01}, \dots, \beta_{0K})^\top \in \mathbb{R}^K, \bbeta_A = (\beta_{A11}, \dots, \beta_{AJ1}, \dots, \beta_{A1K}, \dots \beta_{AJK})^\top \in \mathbb{R}^{JK},  \bbeta_{\bu(\bX)} \in \mathbb{R}^q$ are column vectors of parameters,  $\bu(\bX) \in \mathbb{R}^{q_{\bu} \times K}$ is a matrix function of $\bX$,
and the error terms $\boldsymbol{\varepsilon} \sim N(\bzero, \bSigma)$, where  $\bSigma \in \mathbb{R}^{K\times K}$ is a positive-definite covariance matrix.
For MMRM-I, $q_{\bu} = p$ (the dimension of $\bX)$, $\bu(\bX) = \bX \bone_K^\top$ and $\bbeta_{\bu(\bX)} = \bbeta_{\bX}$. 
For MMRM-II, $q = pK$, $\bu(\bX) = \bI_K \otimes \bX $ and $\bbeta_{\bu(\bX)} = (\bbeta_{\bX 1}^\top, \dots, \bbeta_{\bX K}^\top)^\top$.

 Under the working model,  the random error vectors  $\boldsymbol{\varepsilon}_i, i=1,\dots,n,$ are assumed  to be independent (of each other and of $\{(A_i,\bX_i)\}_{i=1}^n$), identically distributed draws from $N(\bzero,\bSigma)$.
Denote  $\bSigma = \bSigma(\balpha)$, where $\balpha = (\alpha_1, \dots, \alpha_L)^\top \in \mathbb{R}^L$ is the vector of unknown parameters in  $\bSigma$.
For example, $\balpha$ consists of the lower triangular and diagonal entries of $\bSigma$ if no structure is assumed on $\bSigma$. 

For each $i$, let $n_i = \sum_{t=1}^K M_{it}$ be the number of non-missing outcomes and $\bY_i^o \in \mathbb{R}^{n_i}$ be the observed outcomes if $n_i > 0$.
Let $t_{i,1} < \dots < t_{i,n_i}$ denote the ordered list of visits when the outcomes are not missing. For example, $t_{i,1}$ is the first non-missing visit for subject $i$. We define $\mathbf{D}_{\bM_i} = [\boldsymbol{e}_{t_{i,1}} \,\,\, \boldsymbol{e}_{t_{i,2}} \,\,\, \dots \,\,\, \boldsymbol{e}_{t_{i,n_i}}] \in \mathbb{R}^{K\times n_i}$. 
We use the subscript $\bM_i$ to note that $\mathbf{D}_{\bM_i}$ is a deterministic function of $\bM_i$.
Then $\bY_i^o = \mathbf{D}_{\bM_i}^\top \bY_i$.
The observed data vector for each $i$ is $(\bY_i^o, \bM_i, A_i, \bX_i)$. 

Denote the full set of parameters as $\btheta = (\bbeta^\top, \balpha^\top)^\top$, where $\bbeta = (\bbeta_0^\top, \bbeta_A^\top, \bbeta_{\bu(\bX)}^\top)^\top$. We further define $\mathbf{Q} = [\bI_K \,\,\,\,  (\bI_K \otimes \bA)^\top \,\,\,\, \bu(\bX)^\top]^\top$.
We let $\mathbf{Q}_i$ denote $\mathbf{Q}$ with  $\bA_i,\bX_i$ substituted for $\bA,\bX$. It follows that  $\mathbf{Q}^\top_i \bbeta = 
\bbeta_0 + (\bI_K \otimes \bA_i)^\top\bbeta_{\bA} + \bu(\bX_i)^\top \bbeta_{\bu(\bX)}$.
Then we have $\bY_i^o|(A_i,\bX_i,\bM_i, \bM_i \ne \bzero_K) \sim N(\mathbf{D}_{\bM_i}^\top\mathbf{Q}_i^\top\bbeta, \mathbf{D}_{\bM_i}^\top\bSigma\mathbf{D}_{\bM_i})$ under the MMRM-I or MMRM-II working model assumptions and missing completely at random (MCAR). The corresponding log likelihood function conditional on $\{A_i,\bX_i,\bM_i\}_{i=1}^n$ is a constant independent of the parameter vector $\btheta$ plus the following:
\begin{align*}
 & -\frac{1}{2} \sum_{i=1}^n I\{\bM_i \ne \bzero_K\}\left\{ \log|\mathbf{D}_{\bM_i}^\top\bSigma\mathbf{D}_{\bM_i}| + (\bY_i^o - \mathbf{D}_{\bM_i}^\top\mathbf{Q}_i^\top\bbeta)^\top(\mathbf{D}_{\bM_i}^\top\bSigma\mathbf{D}_{\bM_i})^{-1}(\bY_i^o - \mathbf{D}_{\bM_i}^\top\mathbf{Q}_i^\top\bbeta)\right\}\\
 &=  - \frac{1}{2}\sum_{i=1}^nI\{\bM_i \ne \bzero_K\}\left\{ \log|\mathbf{D}_{\bM_i}^\top\bSigma\mathbf{D}_{\bM_i}| + (\bY_i - \mathbf{Q}_i^\top\bbeta)^\top\mathbf{D}_{\bM_i}(\mathbf{D}_{\bM_i}^\top\bSigma\mathbf{D}_{\bM_i})^{-1}\mathbf{D}_{\bM_i}^\top(\bY_i - \mathbf{Q}_i^\top\bbeta)\right\}\\
 &= \frac{n}{2}P_n l(\btheta; \bY^o|A,\bX,\bM),
\end{align*}
where we define
\begin{align*}
 & l(\btheta; \bY^o|A,\bX,\bM)\\
 & = -I\{\bM \ne \bzero_K\}\left\{\log|\mathbf{D}_{\bM}^\top\bSigma\mathbf{D}_{\bM}| + (\bY - \mathbf{Q}^\top\bbeta)^\top\mathbf{D}_{\bM}(\mathbf{D}_{\bM}^\top\bSigma\mathbf{D}_{\bM})^{-1}\mathbf{D}_{\bM}^\top(\bY - \mathbf{Q}^\top\bbeta)\right\}.
\end{align*}

To derive the estimating functions for the corresponding maximum likelihood estimator, we use the following results to compute the differential of $l(\btheta; \bY^o|A,\bX,\bM)$ with respect to $\btheta$. By Equation (8.7) of \cite{dwyer1967some}, we have $\frac{\partial \log(|\mathbf{D}_{\bM}^\top\bSigma\mathbf{D}_{\bM}|)}{\partial \bSigma} = \mathbf{D}_{\bM}(\mathbf{D}_{\bM}^\top\bSigma\mathbf{D}_{\bM})^{-1}\mathbf{D}_{\bM}^\top$. Using the chain rule of matrix derivatives \cite[Theorem 8]{macrae1974matrix}, we have
\begin{displaymath}
 \frac{\partial \log(|\mathbf{D}_{\bM}^\top\bSigma\mathbf{D}_{\bM}|)}{\partial \alpha_j} = \mathrm{tr}\left(\frac{\partial \log(|\mathbf{D}_{\bM}^\top\bSigma\mathbf{D}_{\bM}|)}{\partial \bSigma}\frac{\partial \bSigma}{\partial \alpha_j}\right) = \mathrm{tr}\left(\mathbf{D}_{\bM}(\mathbf{D}_{\bM}^\top\bSigma\mathbf{D}_{\bM})^{-1}\mathbf{D}_{\bM}^\top\frac{\partial \bSigma}{\partial \alpha_j}\right).
\end{displaymath}
By Theorem 5 of \cite{macrae1974matrix}, we have
\begin{displaymath}
 \frac{\partial (\mathbf{D}_{\bM}^\top\bSigma\mathbf{D}_{\bM})^{-1}}{\partial \alpha_j} = -(\mathbf{D}_{\bM}^\top\bSigma\mathbf{D}_{\bM})^{-1}\mathbf{D}_{\bM}^\top\frac{\partial \bSigma}{\partial \alpha_j}\mathbf{D}_{\bM} (\mathbf{D}_{\bM}^\top\bSigma\mathbf{D}_{\bM})^{-1}.
\end{displaymath}
Denoting $\bV_{\bM}(\bSigma) = I\{\bM \ne \bzero_K\}\mathbf{D}_{\bM}(\mathbf{D}_{\bM}^\top\bSigma\mathbf{D}_{\bM})^{-1}\mathbf{D}_{\bM}^\top$, we have shown that 
\begin{displaymath}
    \frac{\partial\bV_{\bM}(\bSigma) }{\partial \alpha_j} = -\bV_{\bM}(\bSigma) \frac{\partial \bSigma}{\partial \alpha_j}\bV_{\bM}(\bSigma). \label{V_M_derivative}
\end{displaymath}

Using the above results, the estimating functions for the MLE $\widehat{\btheta}$ for $\btheta$ under MMRM-I  or MMRM-II are
\begin{align*}
    & \bpsi^{\textup{(MMRM)}}(A, \bX, \bY, \bM;\btheta)\\
    &= \left(\begin{array}{c}
    \mathbf{Q} \bV_{\bM}(\bY-\mathbf{Q}^\top\bbeta)    \\
     -\mathrm{tr}(\bV_{\bM}\frac{\partial \bSigma}{\partial \alpha_l}) + (\bY-\mathbf{Q}^\top\bbeta)^\top\bV_{\bM}\frac{\partial \bSigma}{\partial \alpha_l}  \bV_{\bM}(\bY-\mathbf{Q}^\top\bbeta),\ l = 1,\dots, L
 \end{array}\right) \numberthis\label{psi:MMRM},
\end{align*}
which implies that $\widehat{\bDelta}^{\textup{(MMRM-I)}}$ and $\widehat{\bDelta}^{\textup{(MMRM-II)}}$ are M-estimators.
In the above expression of $\bpsi^{\textup{(MMRM)}}$, we omit $(\bSigma)$ from $\bV_{\bM}(\bSigma)$ for conciseness.
We note that $\bV_{\bM}$ is a random matrix taking values in $\mathbf{R}^{K\times K}$ and defined in the same way as in Lemma~\ref{lemma-Vm}.

The IMMRM working model~(3) can be written as
\begin{equation}
  \bY = \bbeta_0 + (\bI_K \otimes \bA)^\top\bbeta_{\bA} + (\bI_K \otimes \bX)^\top \bbeta_{\bI_K \otimes \bX} + (\bI_K \otimes \bX \otimes \bA)^\top\bbeta_{\bA\bX} +  \boldsymbol{\varepsilon}_{\bA},
\end{equation}
where $\bbeta_0$, $\bbeta_A$, $\bA$ are defined in Equation~(\ref{proof:MMRM}), $ \bbeta_{\bI_K \otimes \bX} = (\bbeta_{\bX 1}^\top, \dots, \bbeta_{\bX K}^\top)^\top$, $\bbeta_{\bA\bX}^\top \in \mathbb{R}^{JpK}$ with the $\{Jp(k-1)+J(m-1)+j\}$-th entry being $\beta_{\bA X_m jk}$ for $j = 1,\dots, J, k = 1, \dots, K$ and $m = 1, \dots, p$, and $\boldsymbol{\varepsilon}_{\bA} = \sum_{j=0}^J I\{A=j\}\boldsymbol{\varepsilon}_j$, where $\boldsymbol{\varepsilon}_j\sim N(\bzero, \bSigma_j)$ and $(\boldsymbol{\varepsilon}_0, \dots, \boldsymbol{\varepsilon}_J)$ are independent of each other. 
Let $\balpha_j \in \mathbb{R}^L$ be the unknown parameters in $\bSigma_j$ for $j = 0, \dots, J$. We define $\bgamma = (\bbeta_0^\top, \bbeta_{\bA}^\top, \bbeta_{\bI_K \otimes \bX}^\top, \bbeta_{\bA\bX}^\top)^\top$ and $\btheta = (\bDelta, \bgamma^\top, \balpha_0^\top, \dots,  \balpha_J^\top)^\top$.
Following a similar procedure to MMRM-I, the estimating functions for the MLE $\widehat{\btheta}^{\textup{(IMMRM)}}$  under IMMRM are
\begin{align*}
    & \bpsi^{\textup{(IMMRM)}}(A, \bX, \bY, \bM;\btheta)\\
    &= \left(\begin{array}{c}
   \beta_{AjK} + \bX^\top\bbeta_{A\bX jK}   - \Delta_j,\ j = 1,\dots, J\\
    \mathbf{R} \bV_{A\bM}(\bY-\mathbf{R}^\top\bgamma)    \\
     I\{A=j\}\left(-\mathrm{tr}(\bV_{A\bM}\frac{\partial \bSigma_j}{\partial \alpha_{jl}}) + (\bY-\mathbf{R}^\top\bgamma)^\top\bV_{A\bM}\frac{\partial \bSigma_j}{\partial \alpha_{jl}}  \bV_{A\bM}(\bY-\mathbf{R}^\top\bgamma)\right),\\
    \hspace{8.5cm} \ j = 0,\dots, J, \ l = 1,\dots, L
 \end{array}\right), \numberthis\label{psi:IMMRM}
\end{align*}
where $\mathbf{R} = [\bI_K\ \ (\bI_K \otimes \bA)^\top\ \ (\bI_K \otimes \bX)^\top\ \ (\bI_K \otimes \bX \otimes \bA)^\top]^\top$ and\\ $\bV_{A\bM} = \bV_{\bM}(\sum_{j=0}^JI\{A=j\}\bSigma_j)= I\{\bM \ne \bzero_K\}\mathbf{D}_{\bM}(\sum_{j=0}^JI\{A=j\}\mathbf{D}_{\bM}^\top\bSigma_j\mathbf{D}_{\bM})^{-1}\mathbf{D}_{\bM}^\top$. Hence $\widehat{\bDelta}^{\textup{(IMMRM)}}$, as the first $J$ entries of $\widehat{\btheta}^{\textup{(IMMRM)}}$, is an M-estimator.

For each $\textup{est} \in \{\textup{ANCOVA},\ \textup{MMRM-I}, \textup{MMRM-II},\ \textup{IMMRM}\}$, we have just shown that $\widehat{\bDelta}^{\textup{(est)}}$ is an M-estimator.
By Assumption 1 and regularity conditions,  we apply Theorem 1 of \cite{wang2021model} and get, under simple or stratified randomization,
\begin{equation}
    \sqrt{n}(\widehat{\bDelta}^{\textup{(est)}} - \underline{\bDelta}^{\textup{(est)}}) = \frac{1}{\sqrt{n}} \sum_{i=1}^n IF^{\textup{(est)}}(A_i, \bX_i, \bY_i, \bM_i) + o_p(\bone), \label{eq:MMRM-asymptotic-linearity}\\    
\end{equation}
where $ \underline{\bDelta}^{\textup{(est)}}$ satisfies $E^*[\bpsi^{\textup{(est)}}(A,\bX,\bY, \bM;\btheta)] = \bzero$ with $E^*$ defined in Lemma~\ref{lemma:observed-data-distribution}, and $IF^{\textup{(est)}}$ represents the $J$-dimensional influence function. We note that Theorem 1 of \cite{wang2021model} is developed for binary treatment (i.e. $J=1$) and scalar outcome (i.e. $K=1$), but their proof can be easily generalized to accommodate multiple treatment arms  and repeated measured outcomes (as in Example 3 of \citealp{wang2021model}).


We next show $\widehat{\bDelta}^{\textup{(MMRM-I)}}$ is model-robust.
By $E^*[\bpsi^{\textup{(MMRM)}}(A, \bX, \bY, \bM;\underline\btheta)] = \bzero$, we have
\begin{align*}
 E^*[\ubV_{\bM}\{\bY-\underline\bbeta_0 - (\bI_K \otimes \bA)^\top\underline\bbeta_{\bA}  - \bu(\bX)^\top \underline\bbeta_{\bu(\bX)}\}] &= \bzero,\\
 E^*[(\bI_K \otimes \bA)\ubV_{\bM}\{\bY-\underline\bbeta_0 - (\bI_K \otimes \bA)^\top\underline\bbeta_{\bA} - \bu(\bX)^\top \underline\bbeta_{\bu(\bX)}\}] &= \bzero,
\end{align*}
which are first $2K$ equations in $E^*[\bpsi(A, \bX, \bY, \bM;\underline\btheta)] = \bzero$, where $\ubV_{\bM} = \bV_{\bM}(\underline{\bSigma}^{\textup{(MMRM-I)}})$ with $\underline{\bSigma}^{\textup{(MMRM-I)}}$ being the probability limit of $\bSigma(\widehat{\balpha})$ under MMRM-I. By the MCAR assumption and Lemma~\ref{lemma:observed-data-distribution}, the first $K$ equations imply that
\begin{displaymath}
 E^*[\ubV_{\bM}]\sum_{j=0}^J\{E[\bY(j)]-\underline\bbeta_0 - \underline\bbeta_{Aj}- E[\bu(\bX)^\top \underline\bbeta_{\bu(\bX)}]\} \pi_j= \bzero,
\end{displaymath}
where $\bbeta_{Aj} = (\bbeta_{Aj1}, \dots, \bbeta_{AjK})^\top$ for $j = 1,\dots, J$ and $\bbeta_{Aj} = \bzero_K$. Similarly, the $(K+1)-2K$th equations imply that, for $j = 1,\dots, K$ and $t = 1, \dots, K$,
\begin{displaymath}
 E^*[\be_t^\top\ubV_{\bM}]\{E[\bY(j)]-\underline\bbeta_0 - \underline\bbeta_{Aj}- E[\bu(\bX)^\top \underline\bbeta_{\bu(\bX)}]\} \pi_j= \bzero.
\end{displaymath}
The above two sets of equations and the positivtiy assumption imply that, for $j = 0, \dots, J$,
\begin{displaymath}
  E^*[\ubV_{\bM}]\{E[\bY(j)]-\underline\bbeta_0 - \underline\bbeta_{Aj}- E[\bu(\bX)^\top \underline\bbeta_{\bu(\bX)}]\}= \bzero.
\end{displaymath}
The assumption $P(\bM(j) = \bone_K) > 0$ and Lemma~\ref{lemma-Vm} (1) implies that  $E^*[\ubV_{\bM}]$ is invertible. Then the above equations imply that, for $j = 1, \dots, J$,
\begin{displaymath}
 \underline\bbeta_{Aj}  = \underline\bbeta_{Aj} - \underline\bbeta_{A0} = E[\bY(j)] -E[\bY(0)]
\end{displaymath}
 and hence $\underline{\bDelta}^{\textup{(MMRM-I)}} = (\underline{\beta}_{A1K}, \dots, \underline{\beta}_{AJK})^\top = \bDelta^*$. The above proof also applies to the MMRM-II estimator by substituting $\bV_{\bM}(\underline{\bSigma}^{\textup{(MMRM-II)}})$ for $\underline{\bV}_{\bM}$, which implies that $\widehat{\bDelta}^{\textup{(MMRM-II)}}$ is model-robust. Also, since the ANCOVA estimator is a special case of the MMRM-II estimator setting $K=1$, we get that the ANCOVA estimator is model-robust.

Following a similar procedure, we next show that $\widehat{\bDelta}^{\textup{(IMMRM)}}$ is model-robust. We have that,  for $j = 1, \dots, J$, 
\begin{align*}
    E[\ubV_{j\bM}]\{E[\bY(j)]-\underline\bbeta_0 - \underline\bbeta_{Aj}- E[(\bI_K \otimes \bX)^\top \underline\bbeta_{\bI_K \otimes \bX}] - E[(\bI_K \otimes \bX \otimes \be_j)^\top \underline\bbeta_{\bA\bX}] \}\pi_j= \bzero,
\end{align*}
and 
\begin{align*}
    E[\ubV_{0\bM}]\{E[\bY(0)]-\underline\bbeta_0 -  E[(\bI_K \otimes \bX)^\top \underline\bbeta_{\bI_K \otimes \bX}] \}\pi_0= \bzero,
\end{align*}
where $\ubV_{j\bM} = I\{\bM(j) \ne \bzero_K\}\mathbf{D}_{\bM(j)}(\mathbf{D}_{\bM(j)}^\top\underline{\bSigma}_j^{\textup{(IMMRM)}}\mathbf{D}_{\bM(j)})^{-1}\mathbf{D}_{\bM(j)}^\top$ with $\underline{\bSigma}_j^{\textup{(IMMRM)}}$ being the probability limit of $\bSigma_j(\widehat\balpha_j)$ in the IMMRM model (3) for $j =0,\dots, J$. The assumption $P(\bM(j) = \bone_K) > 0$ and Lemma~\ref{lemma-Vm} (1) implies that  $E[\ubV_{j\bM}]$ is invertible. Thus, for $j = 1,\dots, J$,
\begin{align*}
    E[\bY(j)] - E[\bY(0)] = \underline\bbeta_{Aj} + E[(\bI_K \otimes \bX \otimes \be_j)^\top\underline\bbeta_{\bA\bX}],
\end{align*}
which implies $E[Y_K(j)] - E[Y_K(0)] = \underline\beta_{AjK} + E[\bX]^\top \underline\bbeta_{\bA\bX j K}$. Since the first equation of $\bpsi^{\textup{(IMMRM)}}$ indicates that $\underline\beta_{AjK} + E[\bX]^\top \underline\bbeta_{\bA\bX j K} = \underline{\Delta}_j$, we get $\underline{\Delta}_j = E[Y_K(j)] - E[Y_K(0)] = \Delta_j^*$, which completes the proof of model-robustness of $\widehat{\bDelta}^{\textup{(IMMRM)}}$.

We next prove that  $\sqrt{n}(\widehat{\bDelta}^{\textup{(est)}} - \bDelta^*)$ weakly converges to a normal distribution, under simple or stratified randomization. Given Equations~(\ref{eq:MMRM-asymptotic-linearity}), it suffices to show that $\frac{1}{\sqrt{n}} \sum_{i=1}^n IF^{\textup{(est)}}(A_i, \bX_i, \bY_i, \bM_i)$ weakly converges to a normal distribution. Under simple randomization, $(A_i, \bX_i, \bY_i, \bM_i), i = 1,\dots, n$ are independent to each other and identically distributed. Since the regularity conditions implied that $IF^{\textup{(est)}}$ has finite second moment, then the central limit theorem implies the desired weak convergence. Furthermore, we have $\widetilde{\bV}^{\textup{(est)}} = Var^*(IF^{\textup{(est)}}(A, \bX, \bY, \bM))$.
Under stratified randomization, we define\ $\bZ_i(j) = IF^{\textup{(est)}}(j, \bX_i, \bY_i(j), \bM_i(j))$ for $j = 0,\dots, J$. 
Then $IF^{\textup{(est)}}(A_i, \bX_i, \bY_i, \bM_i) = \sum_{j=0}^J I\{A_i=j\}\bZ_i(j)$.
Since $E^*[IF^{\textup{(est)}}(A, \bX, \bY, \bM)] = \bzero$, Lemma~\ref{lemma:observed-data-distribution} implies that $\sum_{j=0}^J\pi_j E[\bZ_i(j)] = \bzero$. By the regularity conditions, we apply Lemma~\ref{lemma-bugni} and get
\begin{align*}
   \frac{1}{\sqrt{n}} \sum_{i=1}^n IF^{\textup{(est)}}(A_i, \bX_i, \bY_i, \bM_i) &=  \frac{1}{\sqrt{n}} \sum_{i=1}^n \sum_{j=0}^J (I\{A_i=j\}\bZ_i(j) - \pi_j E[\bZ_i(j)]) \xrightarrow{d} N(\bzero, \mathbf{G}),
\end{align*}
which completes the proof of asymptotic normality.
In addition, Lemma~\ref{lemma-bugni} also implies that $\widetilde{\bV}^{\textup{(est)}} \succeq \bV^{\textup{(est)}}$.

For the ANCOVA, MMRM-I, MMRM-II and IMMRM estimators, the influence functions by Lemmas~\ref{lemma:if-MMRM} and \ref{lemma:if-IMMRM} are given below:
\begin{align}
    IF^{\textup{(ANCOVA)}} &= \bL \mathbf{T}^{\textup{(ANCOVA)}} \left\{\bY - \boldsymbol{h}(A,\bX)\right\}, \label{eq:if-ANCOVA}\\
    IF^{\textup{(MMRM-I)}} &=   \bL \mathbf{T}^{\textup{(MMRM-I)}} (\bY - \mathbf{Q}^\top \underline\bbeta),  \label{eq:if-MMRM}\\
    IF^{\textup{(MMRM-II)}} &= \bL \mathbf{T}^{\textup{(MMRM-II)}} \left\{\bY - \boldsymbol{h}(A,\bX)\right\}, \label{eq:if-MMRM-II}\\
    IF^{\textup{(IMMRM)}} &= \bL  \mathbf{T}^{\textup{(IMMRM)}} (\bY - \mathbf{R}^\top \underline{\bgamma}) + \bL\mathbf{r}^\top (\bX - E[\bX]),\label{eq:if-IMMRM}
\end{align}
where
\begin{align*}
\bL &= (-\bone_J \quad \bI_J) \in \mathbb{R}^{J \times (J+1)},\\
 \mathbf{T}^{\textup{(ANCOVA)}} &= \left(\frac{I\{A=0\}}{\pi_0}\frac{I\{M_K=1\}}{P^*(M_K=1)}\be_K,\dots, \frac{I\{A=J\}}{\pi_J}\frac{I\{M_K=1\}}{P^*(M_K=1)}\be_K\right)^\top, \\
    \mathbf{T}^{\textup{(MMRM-I)}} &= \left(\frac{I\{A=0\}}{\pi_0}  \ubV_{\bM}E^*[\ubV_{\bM}]^{-1}\be_K,\dots, \frac{I\{A=J\}}{\pi_J}  \ubV_{\bM}E^*[\ubV_{\bM}]^{-1}\be_K\right)^\top, \\    
    \mathbf{T}^{\textup{(MMRM-II)}} &= \left(\frac{I\{A=0\}}{\pi_0} \overline{\ubV}_{\bM}E^*[\overline{\ubV}_{\bM}]^{-1}\be_K, \dots, \frac{I\{A=J\}}{\pi_J} \overline{\ubV}_{\bM}E^*[\overline{\ubV}_{\bM}]^{-1}\be_K\right)^\top,\\
    \mathbf{T}^{\textup{(IMMRM)}} &= \left(\frac{I\{A=0\}}{\pi_0}  \ubV_{0\bM}E^*[\ubV_{0\bM}]^{-1}\be_K,\dots, \frac{I\{A=J\}}{\pi_J}  \ubV_{J\bM}E^*[\ubV_{J\bM}]^{-1}\be_K\right)^\top,\\
     \bY-\boldsymbol{h}(A,\bX) &= \sum_{j=0}^J  I\{A=j\}\left\{\bY- E[\bY(j)] -Cov^*\left(\bY, \bX\right) Var(\bX)^{-1} (\bX - E[\bX])\right\}, \\
    \bY - \mathbf{Q}^\top \underline\bbeta &= \sum_{j=0}^J I\{A=j\}\left\{\bY(j) - E[\bY(j)] - \bone_K (\bX - E[\bX])^\top  \underline\bbeta_{\bX}\right\}, \\
    \bY - \mathbf{R}^\top \underline{\bgamma} &= \sum_{j=0}^J I\{A=j\}\left\{\bY(j) - E[\bY(j)] - Cov(\bY(j),\bX)Var(\bX)^{-1}(\bX - E[\bX])\right\}, \\
    \mathbf{r} &= (\bb_{K0}, \dots, \bb_{KJ}),
\end{align*}
where 
\begin{align*}
    \underline{\bV}_{\bM} &= \underline{\bV}_{\bM}(\underline{\bSigma}^{\textup{(MMRM-I)}}) = \bV_{\bM}(E^*[(\bY - \mathbf{Q}^\top \underline\bbeta)(\bY - \mathbf{Q}^\top \underline\bbeta)^\top]), \\
    \overline{\ubV}_{\bM} &= \underline{\bV}_{\bM}(\underline{\bSigma}^{\textup{(MMRM-II)}}) = \bV_{\bM}(E^*[\{\bY - h(A,\bX)\}\{\bY - h(A,\bX)\}^\top]), \\
    \ubV_{j\bM} &= \underline{\bV}_{\bM}(\underline{\bSigma}_j^{\textup{(IMMRM)}}) = \bV_{\bM}\left(E^*\left[\frac{I\{A=j\}}{\pi_j} (\bY - \mathbf{R}^\top \underline{\bgamma})  (\bY - \mathbf{R}^\top \underline{\bgamma})^\top \right]\right),\ , j = 0,\dots,J, \\
    \underline\bbeta_{\bX} &= Var(\bX)^{-1}Cov^*(\bY,\bX)\frac{E^*[\ubV_{\bM}]\bone_K }{\bone^\top_K E^*[\ubV_{\bM}] \bone_K}, \\
    \bb_{Kj} &= Var(\bX)^{-1} Cov(\bX, Y_K(j)),\ j = 0, \dots, J.
\end{align*}
Furthermore, we have
\begin{align}
 \widetilde{\bV}^{\textup{(ANCOVA)}} &= \frac{1}{P^*(M_K=1)}\bL\left( diag\{\pi_j^{-1}\be_K^\top \underline\bSigma_j^{\textup{(ANCOVA)}} \be_K: j = 0,\dots, J\} \right) \bL^\top, \label{tilde-V-ANCOVA}\\
\widetilde{\bV}^{\textup{(MMRM-I)}} &= \bL\ diag\bigg\{\be_K^\top E^*[\ubV_{\bM}]^{-1}E^*\left[\pi_j^{-1}\ubV_{\bM}  \underline\bSigma_j^{\textup{(MMRM-I)}} \ubV_{\bM}\right]  E^*[\ubV_{\bM}]^{-1}\be_K: j = 0, \dots, J\bigg\} \ \bL^\top, \label{tilde-V-MMRM}\\
\widetilde{\bV}^{\textup{(MMRM-II)}} &= \bL\ diag\bigg\{\be_K^\top E^*\left[\overline{\ubV}_{\bM}\right]^{-1}E^*\left[\pi_j^{-1}\overline{\ubV}_{\bM} \underline\bSigma_j^{\textup{(MMRM-II)}} \overline{\ubV}_{\bM}\right] E^*\left[\overline{\ubV}_{\bM}\right]^{-1}\be_K: j = 0, \dots, J\bigg\} \ \bL^\top,\label{tilde-V-MMRM-II}\\
\widetilde{\bV}^{\textup{(IMMRM)}} &= \bL\left( diag\{\be_K^\top E^*[\pi_j \ubV_{j\bM}]^{-1} \be_K: j = 0,\dots, J\} + \mathbf{r}^\top Var(\bX) \mathbf{r}\right) \bL^\top,\label{tilde-V-IMMRM}
\end{align}
where
\begin{align*}
    \underline\bSigma_j^{\textup{(ANCOVA)}} &= \underline\bSigma_j^{\textup{(MMRM-II)}}= E^*\left[\frac{I\{A=j\}}{\pi_j} \{\bY -\boldsymbol{h}(A,\bX)\}\{\bY -\boldsymbol{h}(A,\bX)\}^\top \right], \\
    \underline\bSigma_j^{\textup{(MMRM-I)}} &= E^*\left[\frac{I\{A=j\}}{\pi_j} (\bY - \mathbf{Q}^\top \underline\bbeta)  (\bY - \mathbf{Q}^\top \underline\bbeta)^\top \right].
\end{align*}

We next compute $\bV^{\textup{(est)}}$. Lemma~\ref{lemma:if-IMMRM} implies that, for $j = 0, \dots, J$, $E[Y_K(j)|S] = E[Y_K(j)] + \bb_{Kj}^\top (E[\bX|S] - E[\bX])$. 
Then using Equation~(\ref{eq:if-MMRM}), we get
\begin{align*}
    E[IF^{\textup{(MMRM-I)}}(j, \bX_i, \bY_i(j), \bM_i(j)) |S] = \bL \widetilde{\be}_{j+1}\pi_j^{-1}(\bb_{Kj} - \underline{\bbeta}_{\bX})^\top(E[\bX|S] - E[\bX]),
\end{align*}
where $\widetilde{\be}_{j+1} \in \mathbb{R}^{J+1}$ has the $(j+1)$-th entry 1 and the rest 0.
Hence by Lemma~\ref{lemma-bugni}, we have
\begin{align*}
     \bV^{\textup{(MMRM-I)}} &= \widetilde{\bV}^{\textup{(MMRM-I)}} - \bL[ diag\{\pi_j^{-1} (\bb_{Kj} - \underline{\bbeta}_{\bX})^\top Var(E[\bX|S])(\bb_{Kj} - \underline{\bbeta}_{\bX}): j = 0,\dots, J\}\\
    &\quad - \bv^\top Var(E[\bX|S]) \bv] \bL^\top,  \numberthis \label{V-MMRM}
\end{align*}
where $\bv = (\bb_{K0} - \underline{\bbeta}_{\bX}, \dots, \bb_{KJ} - \underline{\bbeta}_{\bX})$.
similarly, we have $\bV^{\textup{(IMMRM)}} = \widetilde{\bV}^{\textup{(IMMRM)}}$ and
\begin{align*}
    \bV^{\textup{(ANCOVA)}} &= \widetilde{\bV}^{\textup{(ANCOVA)}} - \bL[ diag\{\pi_j^{-1} (\bb_{Kj} - \bb_K)^\top Var(E[\bX|S])(\bb_{Kj} - \bb_K): j = 0,\dots, J\}\\
    &\quad - \mathbf{z}^\top Var(E[\bX|S]) \mathbf{z}] \bL^\top, \numberthis\label{V-ANCOVA} \\
     \bV^{\textup{(MMRM-II)}} &= \widetilde{\bV}^{\textup{(MMRM-II)}} - \bL[ diag\{\pi_j^{-1} (\bb_{Kj} - \bb_K)^\top Var(E[\bX|S])(\bb_{Kj} - \bb_K): j = 0,\dots, J\}\\
    &\quad - \mathbf{z} Var(E[\bX|S]) \mathbf{z}^\top] \bL^\top, \numberthis\label{V-MMRM-II}
\end{align*}
where $\bb_K = Var(\bX)^{-1} Cov^*(\bX, Y_K)$ and $\mathbf{z} = (\bb_{K0} - \bb_K, \dots, \bb_{KJ} - \bb_K)$.

We next show $\bV^{\textup{(ANCOVA)}} \succeq \bV^{\textup{(IMMRM)}}$. By the definition of $\underline\bSigma_j^{\textup{(ANCOVA)}}$ and $\underline\bSigma_j^{\textup{(IMMRM)}}$, we have
\begin{align*}
\underline\bSigma_j^{\textup{(ANCOVA)}} - \underline\bSigma_j^{\textup{(IMMRM)}} = \{Cov(\bY(j), \bX) - Cov^*(\bY,\bX)\} Var(\bX)^{-1}\{Cov(\bX, \bY(j)) - Cov^*(\bX, \bY)\}    
\end{align*}
is positive semi-definite, and $$\be_K^\top\underline\bSigma_j^{\textup{(ANCOVA)}}\be_K =   \be_K^\top\underline\bSigma_j^{\textup{(IMMRM)}}\be_K + (\bb_{Kj} - \bb_K)^\top Var(\bX)(\bb_{Kj} - \bb_K).$$
Using Equations~(\ref{tilde-V-ANCOVA}) and (\ref{tilde-V-IMMRM}) and the fact that $Var(\bX)  = E[Var(\bX|S)] + Var(E[\bX|S])$, we have
\begin{align*}
   &\bV^{\textup{(ANCOVA)}} - \widetilde{\bV}^{\textup{(IMMRM)}} \\
   &= \bV^{\textup{(ANCOVA)}} -\widetilde{\bV}^{\textup{(ANCOVA)}} +  \widetilde{\bV}^{\textup{(ANCOVA)}} - \widetilde{\bV}^{\textup{(IMMRM)}} \\   
   &\succeq \bV^{\textup{(ANCOVA)}} -\widetilde{\bV}^{\textup{(ANCOVA)}} +  \widetilde{\bV}^{\textup{(ANCOVA)}}\\
   &\quad -  \bL\left( diag\{P^*(M_K = 1)^{-1}\pi_j^{-1}\be_K^\top \underline\bSigma_j^{\textup{(IMMRM)}}\be_K: j = 0,\dots, J\} + \mathbf{r}^\top Var(\bX) \mathbf{r}\right) \bL^\top \\
    &= \bV^{\textup{(ANCOVA)}} -\widetilde{\bV}^{\textup{(ANCOVA)}}\\
    &\quad + \frac{1}{P^*(M_K=1)}\bL diag\{\pi_j^{-1}(\bb_{Kj} - \bb_K)^\top Var(\bX) (\bb_{Kj} - \bb_K): j = 0,\dots, J\}\bL^\top - \bL\mathbf{r}^\top Var(\bX) \mathbf{r} \bL^\top \\
    &\succeq \bL[ diag\{\pi_j^{-1} (\bb_{Kj} - \bb_K)^\top E[Var(\bX|S)](\bb_{Kj} - \bb_K): j = 0,\dots, J\} - \mathbf{z}^\top E[Var(\bX|S)] \mathbf{z}] \bL^\top \\
    &= \bL \mathbf{U}^\top \{(diag\{\boldsymbol{\pi}\}-\boldsymbol{\pi}\boldsymbol{\pi}^\top) \otimes \bI_p\}\mathbf{U} \bL^\top, \numberthis \label{ANCOVA-IMMRM}
\end{align*}
where $\boldsymbol{\pi} = (\pi_0,\dots, \pi_J)^\top$ and
\begin{displaymath}
  \mathbf{U}^\top = \left(\begin{array}{ccc}
    \pi_0^{-1} (\bb_{K0} - \bb_K)^\top E[Var(\bX|S)]^{\frac{1}{2}}   &  \\
       & \ddots \\
       & & \pi_J^{-1} (\bb_{KJ} - \bb_K)^\top E[Var(\bX|S)]^{\frac{1}{2}}
  \end{array}\right).
\end{displaymath}
In the above derivation, the first ``$\succeq$'' results from Lemma~\ref{lemma-Vm} (2), the second ``$\succeq$'' comes from $P^*(M_K=1) \le 1$ and $\bL\mathbf{z}^\top = \bL \mathbf{r}^\top$.
Since $diag\{\boldsymbol{\pi}\}\succeq\boldsymbol{\pi}\boldsymbol{\pi}^\top$, then $(diag\{\boldsymbol{\pi}\}-\boldsymbol{\pi}\boldsymbol{\pi}^\top) \otimes \bI_p$ is positive semi-definite (by Lemma~\ref{lemma:kronecker-product}) and hence $\bV^{\textup{(ANCOVA)}} \succeq \widetilde{\bV}^{\textup{(IMMRM)}}$.

We next show $\bV^{\textup{(MMRM-I)}} \succeq \bV^{\textup{(IMMRM)}}$. 
Using the definition of $\underline\bSigma_j^{\textup{(MMRM-I)}}$ and $\underline\bSigma_j^{\textup{(IMMRM)}}$, we have
$\underline\bSigma_j^{\textup{(MMRM-I)}} - \underline\bSigma_j^{\textup{(IMMRM)}}  = \bLambda_j$, where
\begin{align*}
   \bLambda_j
  &= \left\{Cov(\bY(j),\bX)-\bone_K\underline{\bbeta}_{\bX}^\top Var(\bX)\right\}Var(\bX)^{-1}\left\{Cov(\bX, \bY(j))-Var(\bX)\underline{\bbeta}_{\bX}\bone_K^\top \right\}
\end{align*}
is positive semi-definite. By Equations~(\ref{tilde-V-MMRM}) and (\ref{tilde-V-IMMRM}), we have
\begin{align*}
   &\widetilde{\bV}^{\textup{(MMRM-I)}} - \widetilde{\bV}^{\textup{(IMMRM)}} \\
   &= \bL\ diag\bigg\{\be_K^\top E^*[\ubV_{\bM}]^{-1}E^*\left[\pi_j^{-1}\ubV_{\bM}  (\underline\bSigma_j^{\textup{(IMMRM)}}+\bLambda_j) \ubV_{\bM}\right]  E^*[\ubV_{\bM}]^{-1}\be_K: j = 0, \dots, J\bigg\} \ \bL^\top\\
    &\qquad -  \bL diag\{\be_K^\top E[\pi_j \ubV_{j\bM}]^{-1} \be_K: j = 0,\dots, J\}\bL^\top - \bL\mathbf{r}^\top Var(\bX) \mathbf{r} \bL^\top \\
    & \succeq \bL diag\bigg\{\be_K^\top E^*[\ubV_{\bM}]^{-1}E^*\left[\pi_j^{-1}\ubV_{\bM} \bLambda_j\ubV_{\bM}\right]  E^*[\ubV_{\bM}]^{-1}\be_K: j = 0, \dots, J\bigg\}\bL^\top - \bL\mathbf{r}^\top Var(\bX)^{-1} \mathbf{r} \bL^\top \\
    & \succeq \bL diag\{\pi_j^{-1}\be_K^\top \bLambda_j \be_K: j = 0,\dots, J\}\bL^\top - \bL\mathbf{r}^\top Var(\bX)^{-1} \mathbf{r} \bL^\top \\
    &= \bL diag\{\pi_j^{-1} (\bb_{Kj}- \underline{\bbeta}_{\bX})^\top Var(\bX) (\bb_{Kj}- \underline{\bbeta}_{\bX}): j = 0,\dots, J\}\bL^\top - \bL\bv^\top Var(\bX)^{-1} \bv \bL^\top,
\end{align*}
where the first ``$\succeq$'' results from Lemma~\ref{lemma-Vm} (3), the second ``$\succeq$'' results from Lemma~\ref{lemma-Vm} (4) and the last equation comes from $\be_K^\top \bLambda_j \be_K = (\bb_{Kj}- \underline{\bbeta}_{\bX})^\top Var(\bX) (\bb_{Kj}- \underline{\bbeta}_{\bX})$ and $\bL\bv^\top = \bL\mathbf{r}^\top$. By Equation~(\ref{V-MMRM}) and $Var(\bX)  = E[Var(\bX|S)] + Var(E[\bX|S])$, we have
\begin{align*}
   & \bV^{\textup{(MMRM-I)}} - \widetilde{\bV}^{\textup{(IMMRM)}} \\
   &= \bV^{\textup{(MMRM-I)}} -\widetilde{\bV}^{\textup{(MMRM-I)}} +  \widetilde{\bV}^{\textup{(MMRM-I)}} - \widetilde{\bV}^{\textup{(IMMRM)}} \\
   &\succeq - \bL[ diag\{\pi_j^{-1} (\bb_{Kj}- \underline{\bbeta}_{\bX})^\top Var(E[\bX|S])(\bb_{Kj}- \underline{\bbeta}_{\bX}): j = 0,\dots, J\} + \bv Var(E[\bX|S]) \bv^\top] \bL^\top \\
    &\quad +\bL diag\{\pi_j^{-1}(\bb_{Kj}- \underline{\bbeta}_{\bX})^\top Var(\bX) (\bb_{Kj}- \underline{\bbeta}_{\bX}): j = 0,\dots, J\}\bL^\top - \bL\bv^\top Var(\bX)\bv \bL^\top \\
    &= \bL[ diag\{\pi_j^{-1} (\bb_{Kj}- \underline{\bbeta}_{\bX})^\top E[Var(\bX|S)](\bb_{Kj}- \underline{\bbeta}_{\bX}): j = 0,\dots, J\} - \bv E[Var(\bX|S)] \bv^\top] \bL^\top \\
    &= \bL \mathbf{Z}^\top \{(diag\{\boldsymbol{\pi}\}-\boldsymbol{\pi}\boldsymbol{\pi}^\top) \otimes \bI_p\}\mathbf{Z} \bL^\top, \numberthis \label{MMRM-IMMRM}
\end{align*}
where
\begin{displaymath}
  \mathbf{Z}^\top = \left(\begin{array}{ccc}
    \pi_0^{-1} (\bb_{K0}- \underline{\bbeta}_{\bX})^\top E[Var(\bX|S)]^{\frac{1}{2}}   &  \\
       & \ddots \\
       & & \pi_J^{-1} (\bb_{KJ}- \underline{\bbeta}_{\bX})^\top E[Var(\bX|S)]^{\frac{1}{2}}
  \end{array}\right).
\end{displaymath}
Since $diag\{\boldsymbol{\pi}\}\succeq\boldsymbol{\pi}\boldsymbol{\pi}^\top$, then we get $\bV^{\textup{(MMRM-I)}} \succeq \widetilde{\bV}^{\textup{(IMMRM)}} = \bV^{\textup{(IMMRM)}}$.

Next, for showing $\bV^{\textup{(MMRM-II)}} \succeq \bV^{\textup{(IMMRM)}}$, we can follow a similar proof as in the previous paragraph, where $\underline{\bSigma}^{\textup{(MMRM-I)}}$ is substituted by $\underline{\bSigma}^{\textup{(MMRM-II)}}$, and get
\begin{align*}
    \bV^{\textup{(MMRM-II)}} - \widetilde{\bV}^{\textup{(IMMRM)}} \succeq \bL \mathbf{U}^\top \{(diag\{\boldsymbol{\pi}\}-\boldsymbol{\pi}\boldsymbol{\pi}^\top) \otimes \bI_p\}\mathbf{U} \bL^\top,
\end{align*}
which is positive semi-definite.

Finally, we give the necessary and sufficient conditions for $\bV^{\textup{(est)}} = \bV^{\textup{(IMMRM)}}$, $\textup{(est)} \in \{\textup{ANCOVA}, \textup{MMRM-I}, \textup{MMRM-II}\}$ in Proposition~\ref{prop:iff} below.
\end{proof}

\begin{proposition}\label{prop:iff}
Assume $K > 1$, Assumption 1 and regularity conditions in the Supplementary Material. 
For $t = 1,\dots, K$ and $j = 0,\dots, J$, we denote $\boldsymbol{b}_{tj} = Var(\bX)^{-1}Cov\{\bX, Y_t(j)\}$.

Then $\bV^{\textup{(ANCOVA)}} = \bV^{\textup{(IMMRM)}}$ if and only if either of the following two sets of conditions (a-b) holds:
\begin{enumerate}[(a)]
    \item for each $j = 0,\dots, J$, $P(M_K(j) = 1) = 1$ and \\ $(1-2I\{J=1\}\pi_0)(\boldsymbol{b}_{Kj} - \boldsymbol{b}_{K0})^\top E[Var(\bX|S)] = \bzero$;
    \item for each $t = 1,\dots, K-1$ and $j = 0,\dots, J$, \\
    $P(M_t(j) = 1, M_K(j) = 0)\ Cov\{Y_t(j) - \boldsymbol{b}_{tj}^\top \bX, Y_K(j) - \boldsymbol{b}_{Kj}^\top \bX\} = 0$ and $\boldsymbol{b}_{Kj} = \boldsymbol{b}_{K0}$.
\end{enumerate}
In addition, $\bV^{\textup{(MMRM-I)}} = \bV^{\textup{(IMMRM)}}$ if and only if
\begin{enumerate}[(a')]
    \item  for $j = 0,\dots, J$ and $\bm \in \{0,1\}^K\setminus\{\bzero_K\}$ with $P^*(\bM = \bm)$, $\be_K^\top \{E[\ubV_{j\bM}]^{-1}\ubV_{j\bm} - E^*[\ubV_{\bM}]^{-1}\ubV_{\bm}\} = \bzero$,
    \item for $j = 0,\dots, J$ and $\bm \in \{0,1\}^K\setminus\{\bzero_K\}$ with $P^*(\bM = \bm)$, $\be_K^\top E[\ubV_{\bM}]^{-1}\ubV_{\bm} \bLambda_j$ is a constant vector, 
    \item for $j = 0,\dots, J$, $\left[\boldsymbol{b}_{Kj}-\underline{\bbeta}_{\bX}^{\textup{(MMRM-I)}} - I\{J=1\}\pi_j (\boldsymbol{b}_{Kj} - \boldsymbol{b}_{K0})\right]^\top  E[Var(\bX|S)] = \bzero$,
\end{enumerate}
where $\underline{\bbeta}_{\bX}^{\textup{(MMRM-I)}} $ is the probability limit of $\bbeta_{\bX}$ in the MMRM-I working model.

Also, $\bV^{\textup{(MMRM-II)}} = \bV^{\textup{(IMMRM)}}$ if and only if
\begin{enumerate}[(a'')]
    \item  for $j = 0,\dots, J$ and $\bm \in \{0,1\}^K\setminus\{\bzero_K\}$ with $P^*(\bM = \bm)$, $\be_K^\top \{E[\ubV_{j\bM}]^{-1}\ubV_{j\bm} - E^*[\overline\ubV_{\bM}]^{-1}\overline\ubV_{\bm}\} = \bzero$,
    \item for $j = 0,\dots, J$ and $\bm \in \{0,1\}^K\setminus\{\bzero_K\}$ with $P^*(\bM = \bm)$,\\ $\be_K^\top E[\overline\ubV_{\bM}]^{-1}\overline\ubV_{\bm} \{Cov(\bY(j), \bX) - Cov^*(\bY,\bX)\} Var(\bX)^{-1}\{Cov(\bX, \bY(j)) - Cov^*(\bX, \bY)\}$ is a constant vector, 
    \item  for $j = 0,\dots, J$, $(1-2I\{J=1\}\pi_0)(\boldsymbol{b}_{Kj} - \boldsymbol{b}_{K0})^\top E[Var(\bX|S)] = \bzero.$
\end{enumerate}

\end{proposition}
\begin{proof}
We first derive the necessary and sufficient conditions for $\bV^{\textup{(ANCOVA)}} = \bV^{\textup{(IMMRM)}}$. Recall the derivation, i.e., Equations~(\ref{ANCOVA-IMMRM}), for showing $\bV^{\textup{(ANCOVA)}} \succeq \bV^{\textup{(IMMRM)}}$ in the proof of Theorem 1. By check the two inequalities and the last row in Equations~(\ref{ANCOVA-IMMRM}). We have $\bV^{\textup{(ANCOVA)}} = \bV^{\textup{(IMMRM)}}$ if and only if the following three conditions hold:
\begin{enumerate}[(i)]
    \item $\bL\ diag\{\pi_j^{-1}\be_K^\top (\frac{1}{P^*(M_K=1)}\underline{\bSigma}_j^{\textup{(IMMRM)}} - E^*[ \ubV_{j\bM}]^{-1}) \be_K: j = 0,\dots, J\}\bL = \bzero$,
    \item $\{1-P^*(M_K=1)\} \bL diag\{\pi_j^{-1}(\bb_{Kj} - \bb_K)^\top Var(\bX) (\bb_{Kj} - \bb_K): j = 0,\dots, J\}\bL^\top = \bzero$,
    \item $\bL \mathbf{U}^\top \{(diag\{\boldsymbol{\pi}\}-\boldsymbol{\pi}\boldsymbol{\pi}^\top) \otimes \bI_p\}\mathbf{U} \bL^\top = \bzero$.
\end{enumerate}
For Condition (i), Lemma~\ref{lemma-Vm} (2) and the assumption that $K>1$ imply that the equation holds if and only if $P^*(M_K=0, M_t = 1) \be_t^\top \underline{\bSigma}_j^{\textup{(IMMRM)}} \be_K = 0$ for $t = 1, \dots, K-1$ and $j = 0,\dots, J$. Equation~(\ref{eq:if-IMMRM}) implies that  $\be_t^\top \underline{\bSigma}_j^{\textup{(IMMRM)}}\be_K = Cov(Y_t(j)-\bb_{tj}^\top \bX, Y_K(j)-\bb_{Kj}^\top \bX)$. The MCAR assumption implies that $P^*(M_K=0, M_t = 1) = P(M_K(j)=0, M_t(j) = 1)$ for $j = 0,\dots, J$. Hence Condition (i) is equivalent to
\begin{enumerate}[(i)]
    \item $ P(M_K(j)=0, M_t(j) = 1)Cov(Y_t(j)-\bb_{tj}^\top \bX, Y_K(j)-\bb_{Kj}^\top \bX)$ for $t = 1,\dots, K-1$ and $j = 0,\dots, J$.
\end{enumerate}
Condition (ii) is equivalent to
\begin{enumerate}[(ii)]
    \item $P^*(M_K=1)$ or $Cov(\bY_K(j) - \bY_K(0), \bX) = \bzero$.
\end{enumerate}
For Condition (iii), since $\bL \mathbf{U}^\top \{(diag\{\boldsymbol{\pi}\}-\boldsymbol{\pi}\boldsymbol{\pi}^\top) \otimes \bI_p\}\mathbf{U} \bL^\top$ is positive semi-definite, then it is $\bzero$ if and only if all of its diagonal entries are 0.
Denoting $\boldsymbol{u}_j = E[Var(\bX|S)]^{\frac{1}{2}}(\bb_{Kj} - \bb_K)$, we get that the $(j,j)$-th entry of matrix $\bL \mathbf{U}_{S=s} (diag\{\boldsymbol{\pi}\}-\boldsymbol{\pi}\boldsymbol{\pi}^\top)\mathbf{U}_{S=s} \bL^\top$ is
\begin{align*}
    &\pi_j^{-1}\boldsymbol{u}_j^\top\boldsymbol{u}_j + \pi_0^{-1} \boldsymbol{u}_0^\top\boldsymbol{u}_0 - (\boldsymbol{u}_j - \boldsymbol{u}_0)^\top(\boldsymbol{u}_j - \boldsymbol{u}_0) \\
    &= \frac{1}{\pi_0\pi_j}\left(\pi_0 \boldsymbol{u}_j + \pi_j \boldsymbol{u}_0\right)^\top\left(\pi_0 \boldsymbol{u}_j + \pi_j \boldsymbol{u}_0\right) + (1-\pi_0-\pi_j)\left(\frac{1}{\pi_j} \boldsymbol{u}_0^\top\boldsymbol{u}_0 + \frac{1}{\pi_0} \boldsymbol{u}_j^\top\boldsymbol{u}_j \right),
\end{align*}
which is equal to 0 if and only if either $\boldsymbol{u}_0 = \boldsymbol{u}_j = \bzero$, or $\pi_0 + \pi_j = 1$ and $\pi_0 \boldsymbol{u}_j + \pi_j \boldsymbol{u}_0= \bzero$. The former case is equivalent to $E[Var(\bX|S)]^{\frac{1}{2}}(\bb_{Kj} - \bb_{K0}) = \bzero$ for $j = 1,\dots, K$; and the later case is equivalent to $J = 1$ and $(\pi_1 - \pi_0)E[Var(\bX|S)]^{\frac{1}{2}}(\bb_{K1} - \bb_K)$. Hence Condition (iii) is equivalent to
\begin{enumerate}[(iii)]
    \item $(I\{J=1\}(\pi_j-\pi_0) + I\{J>1\})E[Var(\bX|S)](\bb_{Kj} - \bb_K) = \bzero$ for $j = 1, \dots, J$.
\end{enumerate}
Combining Conditions (i-iii) together, we observe that, $P^*(M_K=1) = 1$ in Condition (ii) implies Condition (i), and $Cov(\bY_K(j) - \bY_K(0), \bX) = \bzero$ in Condition (ii) implies Condition (iii). As a result, the three conditions can be summarized into two conditions, which are
\begin{enumerate}[(a)]
    \item for each $j = 1,\dots, J$, $P(M_K(j) = 1) = 1$ and \\ $(1-2I\{J=1\}\pi_0)(\bb_{Kj} - \bb_{K0})^\top E[Var(\bX|S)] = \bzero$;
    \item for each $t = 1,\dots, K-1$ and $j = 0,\dots, J$, \\
    $P(M_t(j) = 1, M_K(j) = 0)\ Cov\{Y_t(j) - \bb_{tj}^\top \bX, Y_K(j) - \bb_{Kj}^\top \bX\} = 0$ and $\bb_{Kj} = \bb_{K0}$.
\end{enumerate}

For the MMRM-I estimator, the derivations, i.e., Equations~(\ref{MMRM-IMMRM}), imply that $\bV^{\textup{(MMRM-I)}} = \bV^{\textup{(IMMRM)}}$ if and only if the following conditions hold:
\begin{enumerate}[(i')]
    \item  for $j = 0,\dots, J$ and $\bm \in \{0,1\}^K\setminus\{\bzero_K\}$ with $P^*(\bM = \bm)$, $\be_K^\top \{E[\ubV_{j\bM}]^{-1}\ubV_{j\bm} - E^*[\ubV_{\bM}]^{-1}\ubV_{\bm}\} = \bzero$,
    \item for $j = 0,\dots, J$ and $\bm \in \{0,1\}^K\setminus\{\bzero_K\}$ with $P^*(\bM = \bm)$, $\be_K^\top E[\ubV_{\bM}]^{-1}\ubV_{\bm} \bLambda_j$ is a constant vector, 
    \item $\bL \mathbf{Z}^\top \{(diag\{\boldsymbol{\pi}\}-\boldsymbol{\pi}\boldsymbol{\pi}^\top) \otimes \bI_p\}\mathbf{Z} \bL^\top = \bzero$.
\end{enumerate}
Similar to the analysis for Condition (iii), Condition (iii') is equivalent to 
\begin{enumerate}[(iii')]
    \item for $j = 1,\dots, J$,
$[\bb_{Kj}-\underline{\bbeta}_{\bX}^{\textup{(MMRM-I)}} - I\{J=1\}\pi_j (\bb_{Kj} - \bb_{K0})]^\top  E[Var(\bX|S)] = \bzero$,
\end{enumerate}
which is the necessary condition given in Corollary 1.

For the MMRM-II estimator, similarly, we have $\bV^{\textup{(MMRM-I)}} = \bV^{\textup{(IMMRM)}}$ if and only if the following conditions hold:
\begin{enumerate}[(i'')]
    \item  for $j = 0,\dots, J$ and $\bm \in \{0,1\}^K\setminus\{\bzero_K\}$ with $P^*(\bM = \bm)$, $\be_K^\top \{E[\ubV_{j\bM}]^{-1}\ubV_{j\bm} - E^*[\overline\ubV_{\bM}]^{-1}\overline\ubV_{\bm}\} = \bzero$,
    \item for $j = 0,\dots, J$ and $\bm \in \{0,1\}^K\setminus\{\bzero_K\}$ with $P^*(\bM = \bm)$,\\ $\be_K^\top E[\overline\ubV_{\bM}]^{-1}\overline\ubV_{\bm} \{Cov(\bY(j), \bX) - Cov^*(\bY,\bX)\} Var(\bX)^{-1}\{Cov(\bX, \bY(j)) - Cov^*(\bX, \bY)\}$ is a constant vector, 
    \item  $\bL \mathbf{U}^\top \{(diag\{\boldsymbol{\pi}\}-\boldsymbol{\pi}\boldsymbol{\pi}^\top) \otimes \bI_p\}\mathbf{U} \bL^\top = \bzero$.
\end{enumerate}
The Condition (iii'') is the same as Condition (iii), which is the necessary condition shown in Corollary 1.
\end{proof}

\begin{lemma}\label{lemma:if-MMRM}
Assume the same assumption as in Theorem 1. Then the influence functions of $\widehat{\bDelta}^{\textup{(ANCOVA)}}$, $\widehat{\bDelta}^{\textup{(MMRM-I)}}$ and $\widehat{\bDelta}^{\textup{(MMRM-II)}}$ are given by Equations~(\ref{eq:if-ANCOVA}), (\ref{eq:if-MMRM}) and (\ref{eq:if-MMRM-II}), respectively.
Under simple randomization, the asymptotic covariance matrix of $\widehat{\bDelta}^{\textup{(ANCOVA)}}$, $\widehat{\bDelta}^{\textup{(MMRM-I)}}$ and $\widehat{\bDelta}^{\textup{(MMRM-II)}}$ are given by Equations~(\ref{tilde-V-ANCOVA}), (\ref{tilde-V-MMRM}) and (\ref{tilde-V-MMRM-II}), respectively.
\end{lemma}
\begin{proof}
We first derive the influence function for the MMRM-I estimator.
Theorem 1 of \cite{wang2021model} implies that  $IF^{\textup{(MMRM-I)}}(A, \bX, \bY, \bM;\btheta) = \mathbf{B}^{-1}\bpsi^{\textup{(MMRM)}}(A, \bX, \bY, \bM;\btheta)$, where $\mathbf{B} = E^*\left[\frac{\partial}{\partial\btheta}\bpsi^{\textup{(MMRM)}}(A, \bX, \bY, \bM; \btheta)\Big|_{\btheta = \underline{\btheta}}\right]$.
Using the formula (\ref{proof:MMRM}) of $\boldsymbol{\psi}^{\textup{(MMRM)}}$,
we can show that
\begin{align*}
     \mathbf{B} 
     &= \left[\begin{array}{cccc}
        -E^*[\ubV_{\bM}]  & - E^*[ \ubV_{\bM}(\bI_K  \otimes \bA)^\top] & -E^*[\ubV_{\bM}\bu(\bX)^\top ] & \mathbf{0}\\
        - E^*[(\bI_K  \otimes \bA)\ubV_{\bM}]  & - E^*[(\bI_K  \otimes \bA)\ubV_{\bM}(\bI_K  \otimes \bA)^\top] & - E^*[(\bI_K  \otimes \bA)\ubV_{\bM}\bu(\bX)^\top ] & \mathbf{0}\\
        -E^*[\bu(\bX)\ubV_{\bM}] & - E^*[\bu(\bX)\ubV_{\bM}(\bI_K  \otimes \bA)^\top] & -E^*[\bu(\bX)\ubV_{\bM}\bu(\bX)^\top] & \mathbf{B}_{34} \\
        \mathbf{0} & \mathbf{0} & \mathbf{B}_{34}^\top & \mathbf{B}_{44}
     \end{array}\right],
\end{align*}
where $\mathbf{B}_{34} \in \mathbb{R}^{q\times r}$ and $\mathbf{B}_{44} \in \mathbb{R}^{r\times r}$ are matrices not related to the influence function of $\widehat{\bDelta}^{\textup{(MMRM-I)}}$. The zeros in the above matrix result from the following derivation:
\begin{align*}
 E^*\left[\ubV_{\bM}\frac{\partial \bSigma}{\partial \alpha_j}  \ubV_{\bM}(\bY-\mathbf{Q}^\top\underline\bbeta)\right] &=  E^*\left[\ubV_{\bM}\frac{\partial \bSigma}{\partial \alpha_j}  \ubV_{\bM}\right]E^*[\bY-\underline\bbeta_0 -(\bI_K \otimes \bA)^\top\underline\bbeta_{\bA} - \bu(\bX)^\top \underline\bbeta_{\bu(\bX)}] \\
 & = \bzero,
\end{align*}
and, similarly, $E^*\left[(\bI_K\otimes \bA)\ubV_{\bM}\frac{\partial \bSigma}{\partial \alpha_j}  \ubV_{\bM}(\bY-\mathbf{Q}^\top\underline\bbeta)\right] = \bzero$.
By the regularity conditions, $\mathbf{B}$ is invertible.
To compute $\mathbf{B}^{-1}$, we define
\begin{displaymath}
 \mathbf{D} = \left(\begin{array}{cccc}
    \bI_K & \mathbf{0}  & \mathbf{0} & \mathbf{0} \\
    - E^*[\bI_K  \otimes \bA]  & \bI_K & \mathbf{0} & \mathbf{0} \\
    - E^*[\bu(\bX)] & \mathbf{0} & \bI_q  & \mathbf{0} \\
       \mathbf{0} & \mathbf{0}& \mathbf{0} & \bI_r 
 \end{array}\right),
\end{displaymath}
and $\mathbf{F} = \mathbf{D} \mathbf{B} \mathbf{D}^\top$.
Since $\mathbf{D}$ is a lower triangular matrix and hence invertible, then $\mathbf{F}$ is invertible.
Since MCAR implies that $E^*[\bu(\bX)\bV_{\bM}] = E^*[\bu(\bX)]E^*[\ubV_{\bM}]$ and $E^*[(\bI_K  \otimes \bA)\ubV_{\bM}] = E^*[\bI_K  \otimes \bA]E^*[\ubV_{\bM}]$, then
\begin{displaymath}
 \mathbf{F} = \left(\begin{array}{cccc}
        -E^*[\ubV_{\bM}]  & \mathbf{0} & \mathbf{0} & \mathbf{0}\\
        \mathbf{0}  & - E^*[\ubV_{\bM}] \otimes Var^*(\bA)  & \mathbf{0} & \mathbf{0}\\
        \mathbf{0} & \mathbf{0} & -Var\{\bu(\bX)E^*[\ubV_{\bM}]^{\frac{1}{2}} \} & \mathbf{B}_{34} \\
        \mathbf{0} & \mathbf{0} & \mathbf{B}_{34}^\top & \mathbf{B}_{44}
     \end{array}\right),
\end{displaymath}
where $Var^*(\bA) = E^*[\bA\bA^\top] - E^*[\bA]E^*[\bA]^\top$, $\mathbf{F}_{33} \in \mathbb{R}^{q\times q}$ is a matrix not related to the influence function of $\widehat{\bbeta}_A$. 
Then
\begin{align*}
 \mathbf{B}^{-1} &= \mathbf{D}^\top \mathbf{F}^{-1} \mathbf{D} \\
 &= \mathbf{D}^\top\left[\begin{array}{cccc}
        -E^*[\ubV_{\bM}]^{-1}  & \mathbf{0} & \mathbf{0} & \mathbf{0}\\
        \mathbf{0}  & - E^*[\ubV_{\bM}]^{-1} \otimes Var^*(\bA)^{-1} & \mathbf{0} & \mathbf{0}\\
        \mathbf{0} & \mathbf{0} & \tilde{\mathbf{B}}_{33} & \tilde{\mathbf{B}}_{34} \\
        \mathbf{0} & \mathbf{0} & \tilde{\mathbf{B}}_{34}^\top & \tilde{\mathbf{B}}_{44}
     \end{array}\right]\mathbf{D}\\
    &= \left[\begin{array}{cccc}
        \tilde{\mathbf{B}}_{11}  & \mathbf{C}^\top & -E^*[\bu(\bX)^\top]\tilde{\mathbf{B}}_{33} & -E^*[\bu(\bX)^\top]\tilde{\mathbf{B}}_{34}\\
        \mathbf{C}  & - E^*[\ubV_{\bM}]^{-1} \otimes Var^*(\bA)^{-1} &  \mathbf{0}& \mathbf{0}\\
        -\tilde{\mathbf{B}}_{33}E^*[\bu(\bX)] & \mathbf{0} & \tilde{\mathbf{B}}_{33} & \tilde{\mathbf{B}}_{34} \\
       -\tilde{\mathbf{B}}_{34}^\top E^*[\bu(\bX)] & \mathbf{0} & \tilde{\mathbf{B}}_{34}^\top & \tilde{\mathbf{B}}_{44}
     \end{array}\right],
\end{align*}
where $\mathbf{C} = \{E^*[\ubV_{\bM}]^{-1} \otimes Var^*(\bA)^{-1}\} E^*[\bI_K  \otimes \bA]$ and $\tilde{\mathbf{B}}_{11} \in \mathbb{R}^{K\times K}$, $\tilde{\mathbf{B}}_{33} \in \mathbb{R}^{q\times q}$, $\tilde{\mathbf{B}}_{34} \in \mathbb{R}^{q\times r}$ and $\tilde{\mathbf{B}}_{44} \in \mathbb{R}^{r\times r}$ are matrices that are not related to the influence function of $\widehat{\bbeta}_{\bA}$ (as shown below).
Since $\bbeta_{\bA}$ are the $(K+1)$-th, $\dots, ((J+1)K)$-th entries in $\btheta$, we need the $(K+1)$-th, $\dots, ((J+1)K)$-th rows of $\mathbf{B}^{-1}$ to derive the influence function for $\widehat\bbeta_{\bA}$, which are
\begin{displaymath}
 \left[\mathbf{C} \,\,\, -\{E^*[\ubV_{\bM}]^{-1} \otimes Var^*(\bA)^{-1}\} \,\,\, \mathbf{0} \,\,\, \mathbf{0}\right].
\end{displaymath}
Then the influence function for $\widehat\bbeta_{\bA}$ is
\begin{equation*}
\{E^*[\ubV_{\bM}]^{-1} \otimes Var^*(\bA)^{-1}\}\{\bI_K \otimes (\bA-E^*[\bA])\}
\ubV_{\bM}  (\bY - \mathbf{Q}^\top \underline\bbeta),
\end{equation*}
which implies that the influence function for $\widehat\bDelta^{\textup{(MMRM-I)}}$ is
\begin{equation*}
 IF^{\textup{(MMRM-I)}}=   Var^*(\bA)^{-1}(\bA-E^*[\bA]) \be_K^\top E^*[\ubV_{\bM}]^{-1} \ubV_{\bM}  (\bY - \mathbf{Q}^\top \underline\bbeta).
\end{equation*}
Since $ Var^*(\bA)^{-1}(\bA-E^*[\bA]) =  \bL (\frac{I\{A=0\}}{\pi_0},\dots, \frac{I\{A=J\}}{\pi_J})^\top$, we get the desired formula of $IF^{\textup{(MMRM-I)}}$. 

We next compute $\bY - \mathbf{Q}^\top \underline\bbeta$. By $E^*[\bpsi^{\textup{(MMRM)}}(A, \bX, \bY, \bM;\btheta)] = \bzero$, we have $\underline{\bbeta} = E^*[\mathbf{Q}\ubV_{\bM}\mathbf{Q}]^{-1} E^*[\mathbf{Q}\ubV_{\bM}\bY]$. 
Recalling $\bu(\bX) = \bX\bone_K^\top$ for the MMRM-I estimator and following a similar procedure for calculating $\mathbf{B}^{-1}$, we have
\begin{align*}
  \underline{\bbeta}  &= \left(\begin{array}{ccc}
    \bI_K & \mathbf{0}  & \mathbf{0}  \\
    - E^*[\bI_K  \otimes \bA]  & \bI_K & \mathbf{0} \\
    - E^*[\bu(\bX)] & \mathbf{0} & \bI_q 
 \end{array}\right)^\top\left(\begin{array}{ccc}
        E^*[\ubV_{\bM}]^{-1}  & \mathbf{0} & \mathbf{0}\\
        \mathbf{0}  &  E^*[\ubV_{\bM}]^{-1} \otimes Var^*(\bA)^{-1}  & \mathbf{0}\\
        \mathbf{0} & \mathbf{0} & Var\{\bu(\bX)E^*[\ubV_{\bM}]^{\frac{1}{2}} \}^{-1}
     \end{array}\right)\\
    & \left(\begin{array}{ccc}
    \bI_K & \mathbf{0}  & \mathbf{0}  \\
    - E^*[\bI_K  \otimes \bA]  & \bI_K & \mathbf{0} \\
    - E^*[\bu(\bX)] & \mathbf{0} & \bI_q 
 \end{array}\right)\left(\begin{array}{c}
    E^*[\ubV_{\bM} \bY]   \\
    E^*[(\ubV_{\bM} \bY) \otimes \bA] \\
    E^*[\bu(\bX)\ubV_{\bM} \bY]
 \end{array}\right),
\end{align*}
which implies $\underline{\bbeta}_{\bA} = E^*[\widetilde{\bY} \otimes \widetilde{\bA}]$ and $\underline{\bbeta}_{\bX}^\top = \frac{\bone^\top_K E^*[\ubV_{\bM}] }{\bone^\top_K E^*[\ubV_{\bM}] \bone_K}Cov^*(\bY,\bX) Var(\bX)^{-1}$. Since $\bbeta_0$ satisfies $E^*[\bY-\underline\bbeta_0 -(\bI_K \otimes \bA)^\top\underline\bbeta_{\bA} - \bu(\bX)^\top \underline\bbeta_{\bu(\bX)}] = \bzero$, we get \\ $\bY - \mathbf{Q}^\top \underline\bbeta = \widetilde{\bY} - \underline\bbeta_{\bA}^\top(\bI_K \otimes \widetilde\bA) - \bone_K\bbeta_{\bX}^\top \widetilde{\bX}$. Then direct calculation gives the desired formula of $\bY - \mathbf{Q}^\top \underline\bbeta$.

We next calculate $\widetilde{\bV}^{\textup{(MMRM-I)}}$. Since $\bSigma$ is unstructured, the second set of estimating equations $\bpsi^{\textup{(MMRM)}}$ implies that, for each $r, s = 1,\dots, K$ and $j = 0, \dots, J$, we have
\begin{align*}
  0 & =E^*\left[-\textrm{tr}(\ubV_{\bM}  (\be_r\be_s^\top + \be_s\be_r^\top))  + (\bY-\mathbf{Q}^\top \underline\bbeta)^\top\ubV_{\bM}(\be_r\be_s^\top + \be_s\be_r^\top)) \ubV_{\bM}(\bY-\mathbf{Q}^\top \underline\bbeta) \right]  \\
  &= E^*\left[2 \be_r^\top\{-\ubV_{\bM} + \ubV_{\bM}(\bY-\mathbf{Q}^\top \underline\bbeta)(\bY-\mathbf{Q}^\top \underline\bbeta)^\top\ubV_{\bM}\} \be_s\right],
\end{align*}
which implies that, for $j = 0, \dots, J$, $    E^*[-\ubV_{\bM}] + E^*[\ubV_{\bM}(\bY-\mathbf{Q}^\top \underline\bbeta)(\bY-\mathbf{Q}^\top \underline\bbeta)^\top\ubV_{\bM}] = \bzero$.
Since $E^*[\ubV_{\bM}] = E^*[\ubV_{\bM}\underline{\bSigma}^{\textup{(MMRM-I)}}\ubV_{\bM}]$, the regularity condition (3) implies that $\underline{\bSigma}^{\textup{(MMRM-I)}} = E^*[(\bY-\mathbf{Q}^\top \underline\bbeta)(\bY-\mathbf{Q}^\top \underline\bbeta)^\top]$. Thus,
\begin{align*}
    &\widetilde{\bV}^{\textup{(MMRM-I)}}\\
    &= E[IF^{\textup{(MMRM-I)}} IF^{\textup{(MMRM-I)}}{}^\top] \\
    &= \bL\ diag\bigg\{\be_K^\top E^*[\ubV_{\bM}]^{-1}E^*\left[\frac{I\{A=j\}}{\pi_j^2}\ubV_{\bM}  (\bY - \mathbf{Q}^\top \underline\bbeta)  (\bY - \mathbf{Q}^\top \underline\bbeta)^\top \ubV_{\bM}\right] \\
    &\qquad E^*[\ubV_{\bM}]^{-1}\be_K: j = 0, \dots, J\bigg\} \ \bL^\top \\
    &=\bL\ diag\bigg\{\be_K^\top E^*[\ubV_{\bM}]^{-1}E^*\left[\pi_j^{-1}\ubV_{\bM}  \underline\bSigma_j^{\textup{(MMRM-I)}} \ubV_{\bM}\right]  E^*[\ubV_{\bM}]^{-1}\be_K: j = 0, \dots, J\bigg\} \ \bL^\top.
\end{align*}

For the MMRM-II estimator, we can follow a similar proof for the MMRM-I estimator and get the desired influence function and asymptotic covariance matrix. The only difference comes from $\bu(\bX) = \bI_K \otimes \bX$.

For the ANCOVA estimator, we observed that it is a special case of the MMRM-I estimator setting $K=1$. Then we have $\ubV_{\bM} = I\{M_K=1\} (\be_K^\top \underline{\bSigma}^{\textup{(ANCOVA)}} \be_K)^{-1}$, which naturally implies the desired influence function and asymptotic covariance matrix. 
\end{proof}

\begin{lemma}\label{lemma:if-IMMRM}
Assume the same assumption as in Theorem 1 and assume $\bSigma_j, j = 0,\dots, J$ are unstructured. Then the influence function of $\widehat{\bDelta}^{\textup{(IMMRM)}}$ is given by Equation~(\ref{eq:if-IMMRM}) and the asymptotic covariance matrix of $\widehat{\bDelta}^{\textup{(IMMRM)}}$ is given by Equation~(\ref{tilde-V-IMMRM}). 
\end{lemma}
\begin{proof}
Following a similar procedure as in Lemma~\ref{lemma:if-MMRM}, we get that the influence function for $\widehat{\bDelta}^{\textup{(IMMRM)}}$ is
$IF^{\textup{(IMMRM)}} = \bU_1 + \bU_2$, where
\begin{align*}
    \bU_1 &= (\be_K^\top \otimes \bI_J)\mathbf{H}^{-1}(\bI_K \otimes \bA - E^*[\bV_{A\bM} \otimes \bA]E^*[\bV_{A\bM}]^{-1}) \bV_{A\bM} (\bY - \mathbf{R}^\top \underline{\bgamma}) \\
    \bU_2 &= (\be_K \otimes \bI_J)^\top \underline\bbeta_{\bA} + (\be_K \otimes \bX \otimes \bI_J)^\top \underline\bbeta_{\bA\bX} - \bDelta^*.
\end{align*}
where $\mathbf{H} = E^*[\bV_{A\bM} \otimes \bA\bA^\top] - E^*[\bV_{A\bM} \otimes \bA]E^*[\bV_{A\bM}]^{-1} E^*[\bV_{A\bM} \otimes \bA^\top]$.

We next compute $\bU_1$. Define $\boldsymbol{\delta}_j \in \mathbb{R}^{J}$ has the $J$-th entry 1 and the rest 0. We have
\begin{align*}
    E^*[\bV_{A\bM} \otimes \bA\bA^\top]^{-1} = \left(\sum_{j=1}^J E[\pi_j\ubV_{j\bM}] \otimes \boldsymbol{\delta}_j\boldsymbol{\delta}_j^\top \right)^{-1} = \sum_{j=1}^J E[\pi_j\ubV_{j\bM}]^{-1} \otimes \boldsymbol{\delta}_j\boldsymbol{\delta}_j^\top
\end{align*}
and hence
\begin{align*}
&  -E^*[\bV_{A\bM}] + E^*[\bV_{A\bM} \otimes \bA^\top]E^*[\bV_{A\bM} \otimes \bA\bA^\top]^{-1}E^*[\bV_{A\bM} \otimes \bA]\\
&= -E[\bV_{A\bM}] + \left(\sum_{j'=1}^J E[\pi_{j'}\ubV_{j'\bM}] \otimes \boldsymbol{\delta}_{j'}^\top\right)\left(\sum_{j=1}^J E[\pi_j\ubV_{j\bM}]^{-1} \otimes \boldsymbol{\delta}_j\boldsymbol{\delta}_j^\top\right)\left(\sum_{j''=1}^J E[\pi_{j''}\ubV_{j''\bM}] \otimes \boldsymbol{\delta}_{j''}\right) \\
&= -E^*[\bV_{A\bM}] + \sum_{j=1}^J E[\pi_j\ubV_{j\bM}] \otimes \boldsymbol{\delta}_j^\top\boldsymbol{\delta}_j\\
&= -E[\pi_0\ubV_{0\bM}].
\end{align*}
Using the Woodbury matrix identity, we get
\begin{align*}
  \mathbf{H}^{-1} &=   E^*[\bV_{A\bM} \otimes \bA\bA^\top]^{-1} - E^*[\bV_{A\bM} \otimes \bA\bA^\top]^{-1} E^*[\bV_{A\bM} \otimes \bA] \\
  &\qquad (-E^*[\bV_{A\bM}] + E^*[\bV_{A\bM} \otimes \bA^\top]E^*[\bV_{A\bM} \otimes \bA\bA^\top]^{-1}E^*[\bV_{A\bM} \otimes \bA])^{-1} \\
  &\qquad \qquad E^*[\bV_{A\bM} \otimes \bA^\top] E^*[\bV_{A\bM} \otimes \bA\bA^\top]^{-1} \\
  &= \sum_{j=1}^J E[\pi_j\ubV_{j\bM}]^{-1} \otimes \boldsymbol{\delta}_j\boldsymbol{\delta}_j^\top - (\bI_K \otimes \bone_J)(-E[\pi_0\ubV_{0\bM}])^{-1}(\bI_K \otimes \bone_J^\top)\\
  &= \sum_{j=1}^J E[\pi_j\ubV_{j\bM}]^{-1} \otimes \boldsymbol{\delta}_j\boldsymbol{\delta}_j^\top +E[\pi_0\ubV_{0\bM}]^{-1} \otimes \bone_J\bone_J^\top
\end{align*}
Then we get 
\begin{align*}
  &  \mathbf{H}^{-1}(\bI_K \otimes \bA - E^*[\bV_{A\bM} \otimes \bA]E^*[\bV_{A\bM}]^{-1}) \\
  &= \left(\sum_{j=1}^J E[\pi_j\ubV_{j\bM}]^{-1} \otimes \boldsymbol{\delta}_j\boldsymbol{\delta}_j^\top +E[\pi_0\ubV_{0\bM}]^{-1} \otimes \bone_J\bone_J^\top\right) \left(\bI_K \otimes \bA - \sum_{j=1}^J E[\pi_j\ubV_{j\bM}]E^*[\bV_{A\bM} ]^{-1}\otimes \boldsymbol{\delta}_j\right) \\
  &= \sum_{j=1}^J E[\pi_j\ubV_{j\bM}]^{-1} \otimes \boldsymbol{\delta}_j I\{A=j\} + E[\pi_0\ubV_{0\bM}]^{-1} \otimes \bone_J (1-I\{A=0\}) - \sum_{j=1}^J E^*[\bV_{A\bM} ]^{-1}\otimes \boldsymbol{\delta}_j\\
  &\quad - E[\pi_0\ubV_{0\bM}]^{-1}(E^*[\bV_{A\bM} ]-E[\pi_0\ubV_{0\bM}])E^*[\bV_{A\bM} ]^{-1} \otimes \bone_J \\
  &=  \sum_{j=1}^J E[\pi_j\ubV_{j\bM}]^{-1} \otimes \boldsymbol{\delta}_j I\{A=j\}-E[\pi_0\ubV_{0\bM}]^{-1} \otimes \bone_J I\{A=0\},
\end{align*}
which implies that $\bU_1 = \bL  \mathbf{T}^{\textup{(IMMRM)}} \bV_{A\bM}(\bY - \mathbf{R}^\top \underline{\bgamma})$. By $E^*[\bpsi^{\textup{(IMMRM)}}(A, \bX, \bY, \bM;\btheta)] = \bzero$, we have $\underline{\bgamma} = E^*[\mathbf{R}\ubV_{A\bM}\mathbf{R}]^{-1} E^*[\mathbf{R}\ubV_{A\bM}\bY]$. Following a similar procedure for calculating $\underline{\bbeta}$ in Lemma~\ref{lemma:if-MMRM}, we get
\begin{align*}
    \bY - \mathbf{R}^\top \underline{\bgamma} &= \sum_{j=0}^J I\{A=j\}\left\{\bY(j) - E[\bY(j)] - Cov(\bY(j),\bX)Var(\bX)^{-1}(\bX - E[\bX])\right\}.
\end{align*}

We next compute $\bU_2$. For each $j = 1,\dots, J$, we have the $j$-th entry of $\bU_2$ is
\begin{align*}
   \boldsymbol{\delta}_j^\top\bU_2 = \underline\beta_{AjK} + \underline\bbeta_{A\bX j K}^\top \bX - \Delta_j^*.
\end{align*}
Since we have shown the model-robustness of $\widehat{\bDelta}^{\textup{(IMMRM)}}$ in the proof of Theorem 1, then $\Delta_j^* = \underline\beta_{AjK} + \underline\bbeta_{A\bX j K}^\top E[\bX]$ and hence $\boldsymbol{\delta}_j^\top\bU_2 = \underline\bbeta_{A\bX j K}^\top (\bX - E[\bX])$.
By the formula of $\bY - \mathbf{R}^\top \underline{\bgamma}$, we have
\begin{align*}
\underline\bbeta_{A\bX j K}^\top (\bX - E[\bX]) &= (I\{A=j\}-I\{A=0\}) \be_K^\top  (\bY - \mathbf{R}^\top \underline{\bgamma})\\
&= Cov(Y_K(j) - Y_K(0), \bX)Var(\bX)^{-1} (\bX - E[\bX]) \\
&= (\bb_{Kj} - \bb_{K0})^\top (\bX - E[\bX]),
\end{align*}
which implies $\bU_2 = \bL \mathbf{r}^\top (\bX - E[\bX])$.

We next compute $\widetilde{\bV}^{\textup{(IMMRM)}}$, which is
\begin{align*}
    \widetilde{\bV}^{\textup{(IMMRM)}} &= E^*[IF^{\textup{(IMMRM)}}IF^{\textup{(IMMRM)}}{}^\top] = E^*[\bU_1\bU_1^\top] + E^*[\bU_1\bU_2^\top] + E^*[\bU_2\bU_1^\top] + E^*[\bU_2\bU_2^\top].
\end{align*}
For $E^*[\bU_1\bU_2^\top]$, since $E^*[\bpsi^{\textup{(IMMRM)}}(A,\bX,\bY,\bM)] = \bzero$ imply that $E^*[\bV_{A\bM} (\bY - \mathbf{R}^\top \underline{\bgamma})|\bX] =\bzero$ and $E^*[(\bI_K \otimes \bA)\bV_{A\bM} (\bY - \mathbf{R}^\top \underline{\bgamma})|\bX] =\bzero$, then $E^*[\bU_1|\bX] = \bzero$ and hence $E^*[\bU_1\bU_2^\top] = E^*[E^*[\bU_1|\bX]\bU_2^\top] = \bzero$. Thus, $\widetilde{\bV}^{\textup{(IMMRM)}} =  E^*[\bU_1\bU_1^\top] + E^*[\bU_2\bU_2^\top]$.
For $E^*[\bU_1\bU_1^\top]$, since $\bSigma_j$ is unstructured for each $j = 0, \dots, J$, the second set of estimating equations $\bpsi^{\textup{(IMMRM)}}$ implies that, for each $r, s = 1,\dots, K$ and $j = 0, \dots, J$, we have
\begin{align*}
  \bzero & =E^*\left[-\textrm{tr}(\bV_{A\bM} I\{A=j\} (\be_r\be_s^\top + \be_s\be_r^\top))  \right. \\
  &\quad\qquad \left. + (\bY-\mathbf{R}^\top\underline\bgamma)^\top\bV_{A\bM}I\{A=j\}(\be_r\be_s^\top + \be_s\be_r^\top)) \bV_{A\bM}(\bY-\mathbf{R}^\top\underline\bgamma) \right]  \\
  &= E^*\left[2I\{A=j\} \be_r^\top\{-\bV_{A\bM} + \bV_{A\bM}(\bY-\mathbf{R}^\top\underline\bgamma)(\bY-\mathbf{R}^\top\underline\bgamma)^\top\bV_{A\bM}\} \be_s\right],
\end{align*}
which implies that, for $j = 0, \dots, J$, 
\begin{equation}\label{eq:vmj}
    E[-\pi_j\ubV_{j\bM}] + E^*[\ubV_{j\bM}I\{A=j\}(\bY-\mathbf{R}^\top\underline\bgamma)(\bY-\mathbf{R}^\top\underline\bgamma)^\top\ubV_{j\bM}] = \bzero.
\end{equation}
Hence
\begin{align*}
    E^*[\bU_1\bU_1^\top]    &= E^*[\bL  \mathbf{T}^{\textup{(IMMRM)}} \bV_{A\bM}(\bY - \mathbf{R}^\top \underline{\bgamma})(\bY - \mathbf{R}^\top \underline{\bgamma})^\top \bV_{A\bM} \mathbf{T}^{\textup{(IMMRM)}}\bL^\top] \\
    &= \bL\ diag\left\{\frac{1}{\pi_j^2}\be_K^\top E[\ubV_{j\bM}]^{-1}E^*[\bV_{A\bM}I\{A=j\}(\bY - \mathbf{R}^\top \underline{\bgamma})\right.\\
    &\qquad  (\bY - \mathbf{R}^\top \underline{\bgamma})^\top \bV_{A\bM}]E[\ubV_{j\bM}]^{-1}\be_K: j =0,\dots, J\bigg\}\bL^\top\\
    &= \bL\ diag\{\be_K^\top E[\pi_j\ubV_{j\bM}]^{-1}\be_K: j =0,\dots, J\}\bL^\top,
\end{align*}
where the last equation results from Equation~(\ref{eq:vmj}).
Since $ E^*[\bU_2\bU_2^\top] = \bL\mathbf{r}^\top Var(\bX) \mathbf{r} \bL^\top$, we get the desired formula of $\widetilde{\bV}^{\textup{(IMMRM)}}$.
\end{proof}

\subsection{Proof of Corollary 1}
\begin{proof}
The ANHECOVA estimator is a special case of the IMMRM model setting $K=1$. Hence the consistency and asymptotic normality under simple or stratified randomization are implied by Theorem 1.
Furthermore, by Equation~(\ref{tilde-V-IMMRM}), we have
\begin{align*}
    \widetilde{\bV}^{\textup{(ANHECOVA)}} &= \bL\left( diag\{P^*(M_K=1)^{-1}\pi_j^{-1}\be_K^\top \underline{\bSigma}^{\textup{(IMMRM)}} \be_K: j = 0,\dots, J\} + \mathbf{r}^\top Var(\bX) \mathbf{r}\right) \bL.
\end{align*}
Then, by Lemma~\ref{lemma-Vm} (2),
\begin{align*}
  &\widetilde{\bV}^{\textup{(ANHECOVA)}} - \widetilde{\bV}^{\textup{(IMMRM)}}  \\
  &= \bL diag\{\pi_j^{-1}\be_K^\top (P^*(M_K=1)^{-1}\underline{\bSigma}^{\textup{(IMMRM)}} - E[\ubV_{jM}]^{-1} )\be_K: j = 0,\dots, J\} \bL \\
  &\succeq \bzero,
\end{align*}
with equality holds if and only if $P(M_t(j) = 1, M_K(j) = 0)\ Cov\{Y_t(j) - \bb_{tj}^\top \bX, Y_K(j) - \bb_{Kj}^\top \bX\} = 0$ for each $t = 1,\dots, K-1$ and $j = 0,\dots, J$.
The final result comes from the fact that $ Cov\{\bb_{tj}^\top \bX, Y_K(j) - \bb_{Kj}^\top \bX\} = \bzero$.
\end{proof}

\subsection{Proof of Corollary 2}
\begin{proof}
By Equations~(\ref{V-ANCOVA}) and (\ref{V-MMRM-II}), we have
\begin{align*}
     &\widetilde{\bV}^{\textup{(ANCOVA)}} - \bV^{\textup{(ANCOVA)}} \\
     &=\widetilde{\bV}^{\textup{(MMRM-II)}} - \bV^{\textup{(MMRM-II)}}\\
     &= \bL[ diag\{\pi_j^{-1} (\bb_{Kj} - \bb_K)^\top Var(E[\bX|S])(\bb_{Kj} - \bb_K): j = 0,\dots, J\} - \mathbf{z}^\top Var(E[\bX|S]) \mathbf{z}] \bL^\top.
\end{align*}
If $J=1$ and $\pi_1=\pi_0 = 0.5$, then $\bL = (-1,1)$, $\bb_K = 0.5 (\bb_{K1} + \bb_{K0})$ and $\mathbf{z} \bL^\top = \bb_{K1} - \bb_{K0}$. Hence
\begin{align*}
    & \widetilde{V}^{\textup{(ANCOVA)}} - V^{\textup{(ANCOVA)}} \\
    &= \sum_{j=0}^1 2 \{\bb_{K1}-0.5 (\bb_{K1} + \bb_{K0})\}^\top  Var(E[\bX|S]) \{\bb_{K1}-0.5 (\bb_{K1} + \bb_{K0})\} \\
    &\quad - (\bb_{K1} - \bb_{K0})^\top  Var(E[\bX|S]) (\bb_{K1} - \bb_{K0}) \\
    &= 0.
\end{align*}

We next compare $\bV^{\textup{(ANCOVA)}}$ and $\bV^{\textup{(MMRM-II)}}$. By the definition of $\underline{\bSigma}^{\textup{(ANCOVA)}}$ and $\underline{\bSigma}_j^{\textup{(ANCOVA)}}$, we have $2\underline{\bSigma}^{\textup{(ANCOVA)}} = \underline{\bSigma}_0^{\textup{(ANCOVA)}} + \underline{\bSigma}_1^{\textup{(ANCOVA)}}$. Then Equation~(\ref{tilde-V-ANCOVA}) implies that $\widetilde{V}^{\textup{(ANCOVA)}} = 4 P^*(M_K=1)^{-1} \be_K^\top \underline{\bSigma}^{\textup{(ANCOVA)}}\be_K$. Similarly, we have $\widetilde{V}^{\textup{(MMRM-II)}} = 4 \be_K^\top E^*[\overline{\ubV}_{\bM}]^{-1}\be_K$. Since $\underline{\bSigma}^{\textup{(ANCOVA)}} = \underline{\bSigma}^{\textup{(MMRM-II)}}$, then Lemma~\ref{lemma-Vm} (2) implies that $\widetilde{V}^{\textup{(ANCOVA)}} -  \widetilde{V}^{\textup{(MMRM-II)}}  \ge \be_K^\top(\underline{\bSigma}^{\textup{(ANCOVA)}} - \underline{\bSigma}^{\textup{(MMRM-II)}})\be_K = 0$.

Finally, we show $V^{\textup{(MMRM-I)}} \ge V^{\textup{(MMRM-II)}}$. Under two-armed equal randomization,  Equation~(\ref{V-MMRM}) implies that
\begin{align*}
    & \widetilde{V}^{\textup{(MMRM-I)}} - V^{\textup{(MMRM-I)}} = 4 \{\bb_{K}-\underline{\bbeta}_{\bX}\}^\top  Var(E[\bX|S])  \{\bb_{K}-\underline{\bbeta}_{\bX}\}.
\end{align*}
In addition, since $2\underline{\bSigma}^{\textup{(MMRM-I)}} = \underline{\bSigma}_0^{\textup{(MMRM-I)}} + \underline{\bSigma}_1^{\textup{(MMRM-I)}}$, Equation~(\ref{tilde-V-MMRM}) implies that $\widetilde{\bV}^{\textup{(MMRM-I)}} = 4 \be_K^\top E^*[\ubV_{\bM}]^{-1}\be_K$. 
By the definition of $\underline{\bSigma}^{\textup{(MMRM-I)}}$ and $\underline{\bSigma}^{\textup{(MMRM-II)}}$, we have
\begin{align*}
 & \be_K^\top(\underline{\bSigma}^{\textup{(MMRM-I)}} - \underline{\bSigma}^{\textup{(MMRM-II)}} )\be_K \\
 &= Var^*(Y_K) + \underline{\bbeta}_{\bX}^\top Var(\bX)\underline{\bbeta}_{\bX} - \underline{\bbeta}_{\bX}^\top Cov^*(\bX, Y_K) - Cov^*(Y_K, \bX) \underline{\bbeta}_{\bX} \\
 &\quad - Var^*(Y_K) - \bb_K^\top Var(\bX) \bb_K \\
 &= \{\bb_{K}-\underline{\bbeta}_{\bX}\}^\top  Var(\bX) \{\bb_{K}-\underline{\bbeta}_{\bX}\}.
\end{align*}
Hence
\begin{align*}
   & V^{\textup{(MMRM-I)}} - V^{\textup{(MMRM-II)}}  \\
   &= V^{\textup{(MMRM-I)}} - \widetilde{V}^{\textup{(MMRM-I)}} + \widetilde{V}^{\textup{(MMRM-I)}} - V^{\textup{(MMRM-II)}} \\
   &= -4 \{\bb_{K}-\underline{\bbeta}_{\bX}\}^\top  Var(E[\bX|S])  \{\bb_{K}-\underline{\bbeta}_{\bX}\} + 4\be_K^\top (E^*[\ubV_{\bM}]^{-1} - E^*[\overline\ubV_{\bM}]^{-1})\be_K \\
   &\ge  -4 \{\bb_{K}-\underline{\bbeta}_{\bX}\}^\top  Var(E[\bX|S])  \{\bb_{K}-\underline{\bbeta}_{\bX}\} + 4\be_K^\top (\underline{\bSigma}^{\textup{(MMRM-I)}} - \underline{\bSigma}^{\textup{(MMRM-II)}}) \be_K \\
   &= 4 \{\bb_{K}-\underline{\bbeta}_{\bX}\}^\top  E[Var(\bX|S)] \{\bb_{K}-\underline{\bbeta}_{\bX}\} \\
   &\ge 0,
\end{align*}
where the inequality comes from Lemma~\ref{lemma-Vm} (5).
\end{proof}



\section{A counterexample showing MMRM-II is less precise than ANCOVA} \label{sec: an example MMRM-II}

We assume Assumption 1,  simple randomization, $\pi_1 = \frac{1}{3}$,  $\pi_0 = \frac{2}{3}$, and
\begin{align*}
    E[\bY(1)] &= E[\bY(0)] = \bzero,\\
    Var(\bY(1)) &= \left(\begin{array}{rr}
    4     &  -3\\
    -3    &  4
    \end{array}\right), \quad Var(\bY(0)) = \left(\begin{array}{rr}
    4    &  3\\
    3  &  4
    \end{array}\right), \\
    p_{(1,0)} = p_{(0,1)} = p_{(1,1)} &= \frac{1}{3},
\end{align*}
where $p_{\bm} = P^*(\bM = \bm)$ for $\bm \in \{0,1\}^2$. Then we have $\underline{\bSigma}_j^{\textup{(ANCOVA)}} = \underline{\bSigma}_j^{\textup{(MMRM-II)}} = Var(\bY(j))$ for $j = 0,1$. Furthermore,
\begin{displaymath}
\underline{\bSigma}^{\textup{(MMRM-II)}} = \pi_0 \underline{\bSigma}_0^{\textup{(MMRM-II)}} + \pi_1 \underline{\bSigma}_1^{\textup{(MMRM-II)}} = \left(\begin{array}{rr}
    4     &  1\\
    1     &  4
    \end{array}\right)
\end{displaymath}
We define
\begin{displaymath}
\mathbf{C} = \pi_0^{-1} \underline{\bSigma}_0^{\textup{(ANCOVA)}} + \pi_1^{-1} \underline{\bSigma}_1^{\textup{(ANCOVA)}} = \frac{9}{2} \left(\begin{array}{rr}
   4    &  -1\\
    -1   & 4
    \end{array}\right).
\end{displaymath}
Then, by Equation~(\ref{tilde-V-ANCOVA}),
\begin{displaymath}
Var(\widehat{\Delta}^{\textup{(ANCOVA)}}) = \frac{1}{P(M_K=1)} \be_K^\top \mathbf{C} \be_K = 27.
\end{displaymath}
To compute $Var(\widehat{\Delta}^{\textup{(MMRM-II)}})$, recall that $\overline{\ubV}_{\bM} = \bV_{\bM}(\underline{\bSigma}^{\textup{(MMRM-II)}})$. Then
\begin{align*}
    \overline{\ubV}_{(1,0)} = \left(\begin{array}{rr}
    \frac{1}{4}     &  0\\
    0     &  0
    \end{array}\right), \quad \overline{\ubV}_{(0,1)} = \left(\begin{array}{rr}
    0     &  0\\
    0     &  \frac{1}{4}
    \end{array}\right), \quad  \overline{\ubV}_{(1,1)} = \left(\begin{array}{rr}
    \frac{4}{15}     &  -\frac{1}{15}\\
    -\frac{1}{15}     &  \frac{4}{15}
    \end{array}\right),
\end{align*}
which implies
\begin{displaymath}
E^*[\overline{\ubV}_{\bM}] = \sum_{\bm}p_{\bm}\overline{\ubV}_{\bm}  = \frac{1}{180}\left(\begin{array}{rr}
    31    &  -4\\
    -4    & 31
    \end{array}\right), \quad E^*[\overline{\ubV}_{\bM}]^{-1} = \frac{4}{21}\left(\begin{array}{rr}
   31     &  4\\
    4     &  31
    \end{array}\right).
\end{displaymath}
In addition, we have
\begin{align*}
   \overline{\ubV}_{(1,0)}\mathbf{C}\overline{\ubV}_{(1,0)} & = \left(\begin{array}{rr}
    \frac{9}{8}     &  0\\
    0     &  0
    \end{array}\right), \
    \overline{\ubV}_{(0,1)}\mathbf{C}\overline{\ubV}_{(0,1)} = \left(\begin{array}{rr}
    0     &  0\\
    0     &  \frac{9}{8}
    \end{array}\right), \
    \overline{\ubV}_{(1,1)}\mathbf{C}\overline{\ubV}_{(1,1)} =\frac{1}{50} \left(\begin{array}{rr}
   76     &  -49\\
    -49    &  76
    \end{array}\right),     
\end{align*}
which implies
\begin{displaymath}
E^*[\overline{\ubV}_{\bM}\mathbf{C}\overline{\ubV}_{\bM}] = \sum_{\bm}p_{\bm}\overline{\ubV}_{\bm}\mathbf{C} \overline{\ubV}_{\bm}=  \frac{1}{600}\left(\begin{array}{rr}
    529     &  -196\\
    -196    &  529
    \end{array}\right).
\end{displaymath}
Then, by $\underline{\bSigma}_j^{\textup{(MMRM-II)}} = \underline{\bSigma}_j^{\textup{(ANCOVA)}}$ and Equation~(\ref{tilde-V-MMRM-II}), we have
\begin{align*}
Var(\widehat{\Delta}^{\textup{(MMRM-II)}}) &= \be_K^\top E^*[\overline{\ubV}_{\bM}]^{-1} E^*[\overline{\ubV}_{\bM}\mathbf{C}\overline{\ubV}_{\bM}] E^*[\overline{\ubV}_{\bM}]^{-1}\be_K \\
&= \frac{4}{21} (4\,\, 31) \frac{1}{600}\left(\begin{array}{rr}
    529     &  -196\\
    -196    &  529
    \end{array}\right) \frac{4}{21} \left(\begin{array}{rr}
    4    \\
    31
    \end{array}\right) \\
    &= \frac{12486}{21^2} \approx 28.31.
\end{align*}
Since $27 <28.31$, we have $Var(\widehat{\Delta}^{\textup{(ANCOVA)}}) < Var(\widehat{\Delta}^{\textup{(MMRM-II)}})$.

\section{Missing data mechanism in the simulation study} \label{sec:MAR}
Given $(\bY_i(0), \bY_i(1), \bY_i(2), A_i, X_{i1})$ defined in Section 6.1 of the main paper, we define $R_{it}(j)$ as the residual of $Y_{it}(j)$ regressing on $X_{i1}$ by a simple linear regression.
We then define the censoring time $C_i$ by the following sequential conditional model
\begin{enumerate}
    \item $P(C_i=1|X_{i1},A_i) = 1 - expit[logit(0.99) - 0.14 I\{A_i<2\} X_{i1} -0.12 I\{A_i=2\} X_{i1}]    $,
    \item $P(C_i=2|X_{i1},A_i, C_i > 1, Y_{i1}) = 1 - expit[logit(0.969)  - 0.7 I\{A_i<2\} R_{i1}(A_i) -0.5 I\{A_i=2\} R_{i1}(A_i)]$,
    \item $P(C_i=3|X_{i1},A_i, C_i > 2, Y_{i1},Y_{i2}) = 1 - expit[logit(0.958)  - 0.72 I\{A_i<2\} R_{i2}(A_i) -0.51 I\{A_i=2\} R_{i2}(A_i)]$,
    \item $P(C_i=4|X_{i1},A_i, C_i > 3, Y_{i1},Y_{i2},Y_{i3}) = 1 - expit[logit(0.967)  - 0.74 I\{A_i<2\} R_{i3}(A_i) -0.52 I\{A_i=2\} R_{i3}(A_i)]$,
    \item $P(C_i=5|X_{i1},A_i, C_i > 4, Y_{i1},Y_{i2},Y_{i3},Y_{i4}) = 1 - expit[logit(0.977)  - 0.76 I\{A_i<2\} R_{i4}(A_i) -0.53 I\{A_i=2\} R_{i4}(A_i)]$.
\end{enumerate}
Once we have $C_i$, then $\bM_i$ is defined as $M_{it}=1$ if $t\le C_i$ and $0$ otherwise.

\section{Additional simulations under homogeneity and homoscedasticity}\label{sec: additional simulation}

\begin{table}[p!]
\centering
\caption{Simulation results comparing candidate estimators with $50$ samples per arm under MCAR and MAR with no heterogeneity or heteroscedasticity.
For each estimator, we estimate the average treatment effect of TRT1 and TRT2, both comparing the control group. The following measures are used: bias, empirical standard error (ESE), averaged standard error (ASE), coverage probability (CP), probability of rejecting the null (PoR), relative mean squared error compared to IMMRM (RMSE). For RMSE, a number bigger than 1 indicates a larger mean squared error than IMMRM.}
\label{simulstudy:n50-homo}
\renewcommand{\arraystretch}{0.8}
\resizebox{\textwidth}{!}{
\begin{tabular}{lllrrrrrr} 
\toprule
                                               &                           &    Group        & Bias    & ESE    & ASE    & CP(\%)   & PoR(\%)  & RMSE\\ 
\midrule
   \multirow{8}{*}{MCAR} & \multirow{2}{*}{ANCOVA}   & TRT1  & 0.005 & 0.197 & 0.190 & 94.6 & 5.4 & 1.016 \\
                                               &                           & TRT2 & 0.004 & 0.196 & 0.190 & 94.6 & 99.9 & 1.032 \\ 
\cmidrule{2-9}
                                               & \multirow{2}{*}{MMRM-I}     & TRT1  & 0.004 & 0.198 & 0.205 & 95.9 & 4.1 & 1.026 \\ 
                                              &                           & TRT2 & 0.004 & 0.197 & 0.206 & 95.9 & 99.8 & 1.048 \\ 
\cmidrule{2-9}
                                               & \multirow{2}{*}{MMRM-II} & TRT1  & 0.006 & 0.193 & 0.186 & 94.3 & 5.7 & 0.976 \\
                                               &                           & TRT2 &  0.005 & 0.191 & 0.186 & 94.2 & 99.9 & 0.981 \\  
\cmidrule{2-9}
                                               & \multirow{2}{*}{IMMRM}    & TRT1  & 0.006 & 0.195 & 0.183 & 93.8 & 6.2 & - \\ 
                                               &                           & TRT2 &0.005 & 0.193 & 0.183 & 93.5 & 99.9 & - \\ 
\hline
                         \multirow{8}{*}{MAR}  & \multirow{2}{*}{ANCOVA}   & TRT1  & -0.001 & 0.197 & 0.191 & 94.9 & 5.2 & 0.99 \\ 
                                               &                           & TRT2 &  0.005 & 0.197 & 0.191 & 94.3 & 99.7 & 1.010 \\  
\cmidrule{2-9}
                                               & \multirow{2}{*}{MMRM-I}     & TRT1  & 0 & 0.197 & 0.208 & 96.6 & 3.4 & 0.995 \\ 
                                               &                           & TRT2 & -0.002 & 0.196 & 0.207 & 96 & 99.7 & 1.003 \\  
\cmidrule{2-9}
                                              & \multirow{2}{*}{MMRM-II} & TRT1  & -0.001 & 0.195 & 0.187 & 94.9 & 5.2 & 0.969 \\ 
                                               &                           & TRT2 &0 & 0.193 & 0.186 & 93.9 & 99.9 & 0.974 \\ 
\cmidrule{2-9}
                                               & \multirow{2}{*}{IMMRM}    & TRT1  & -0.002 & 0.198 & 0.185 & 93.7 & 6.3 & - \\ 
                                              &                           & TRT2 & 0 & 0.196 & 0.184 & 93.4 & 99.8 & - \\ 
\bottomrule
\end{tabular}
}
\end{table}

\begin{table}[p!]
\centering
\caption{Simulation results comparing candidate estimators with $200$ samples per arm under MCAR and MAR with no heterogeneity or heteroscedasticity.
For each estimator, we estimate the average treatment effect of TRT1 and TRT2, both comparing the control group. The following measures are used: bias, empirical standard error (ESE), averaged standard error (ASE), coverage probability (CP), probability of rejecting the null (PoR), relative mean squared error compared to IMMRM (RMSE). For RMSE, a number bigger than 1 indicates a larger mean squared error than IMMRM.}
\label{simulstudy:n200-homo}
\renewcommand{\arraystretch}{0.8}
\resizebox{\textwidth}{!}{
\begin{tabular}{lllrrrrrr} 
\toprule
                                               &                           &          Group  & Bias    & ESE    & ASE    & CP(\%)  & PoR(\%) & RMSE \\ 
\midrule
 \multirow{8}{*}{MCAR} & \multirow{2}{*}{ANCOVA}   & TRT1  &0.002 & 0.098 & 0.097 & 94.8 & 5.2 & 1.032 \\ 
                                               &                           & TRT2 & 0 & 0.098 & 0.097 & 95.4 & 100  & 1.032 \\  
\cmidrule{2-9}
                                               & \multirow{2}{*}{MMRM-I}     & TRT1  & 0.002 & 0.099 & 0.104 & 96.3  & 3.7 & 1.054 \\ 
                                               &                           & TRT2 & 0.001 & 0.099 & 0.104 & 96.1 & 100 & 1.054 \\ 
\cmidrule{2-9}
                                               & \multirow{2}{*}{MMRM-II} & TRT1 & 0.002 & 0.096 & 0.095 & 94.4 & 5.6 & 1.000 \\ 
                                               &                           & TRT2 &  0.001 & 0.096 & 0.095 & 95.2 & 100  & 0.989 \\
\cmidrule{2-9}
                                               & \multirow{2}{*}{IMMRM}    & TRT1  & 0.002 & 0.096 & 0.095 & 94.4 & 5.6 & - \\ 
                                               &                           & TRT2 & 0.001 & 0.096 & 0.095 & 94.9 & 100  & - \\
\cmidrule{1-9}
                         \multirow{8}{*}{MAR}  & \multirow{2}{*}{ANCOVA}   & TRT1  &-0.001 & 0.099 & 0.098 & 94.7  & 5.4  & 1.032 \\ 
                                               &                           & TRT2 & 0.002 & 0.098 & 0.098 & 94.6  & 100  & 1.043 \\  
\cmidrule{2-9}
                                               & \multirow{2}{*}{MMRM-I}     & TRT1  & 0 & 0.099 & 0.105 & 96.2  & 3.8 & 1.032 \\ 
                                               &                           & TRT2 & -0.002 & 0.099 & 0.105 & 96  & 100  & 1.054 \\ 
\cmidrule{2-9}
                                               & \multirow{2}{*}{MMRM-II} & TRT1  &  0 & 0.097 & 0.096 & 94.8 & 5.2 & 0.989 \\
                                              &                           & TRT2 & -0.001 & 0.096 & 0.095 & 94.7  & 100  & 0.989 \\ 
\cmidrule{2-9}
                                               & \multirow{2}{*}{IMMRM}    & TRT1  & 0 & 0.097 & 0.096 & 94.4 & 5.6 & - \\ 
                                               &                           & TRT2 & -0.001 & 0.096 & 0.095 & 94.5 & 100  & - \\ 
\bottomrule
\end{tabular}
}
\end{table}

{\small
\bibliographystyle{apalike}
\bibliography{references}
}